\documentclass{article}

\usepackage{etoolbox}

\newtoggle{for_arxiv}
\toggletrue{for_arxiv}

\nottoggle{for_arxiv}{%
    \usepackage[accepted]{aistats2019} 
}

\usepackage{microtype}
\usepackage{graphicx}
\usepackage{subfigure}
\usepackage{booktabs} 
\usepackage{siunitx}
\usepackage{hyperref}
\usepackage{xargs}[2008/03/08]

\iftoggle{for_arxiv} {
    \usepackage[authoryear]{natbib}
} {
    \usepackage[round]{natbib}

}

\usepackage{prettyref}
\usepackage{refstyle}

\usepackage{amsmath}
\usepackage{amssymb}
\usepackage{amsfonts}
\usepackage{amsthm}
\usepackage{mathrsfs}
\usepackage{mathtools}
\usepackage{colonequals}
\usepackage{algpseudocode, algorithm} 

\usepackage{listings}
\usepackage{pdfpages}

\AtBeginDocument{\providecommand\assuref[1]{\ref{assu:#1}}}
\AtBeginDocument{\providecommand\condref[1]{\ref{cond:#1}}}
\AtBeginDocument{\providecommand\lemref[1]{\ref{lem:#1}}}
\AtBeginDocument{\providecommand\defrefref[1]{\ref{defref:#1}}}
\AtBeginDocument{\providecommand\proprefref[1]{\ref{propref:#1}}}
\AtBeginDocument{\providecommand\thmrefref[1]{\ref{thmref:#1}}}
\AtBeginDocument{}

\theoremstyle{plain}
    \ifx\thechapter\undefined
      \newtheorem{assumption}{\protect\assumptionname}
    \else
      \newtheorem{assumption}{\protect\assumptionname}[chapter]
    \fi
\theoremstyle{definition}
    \ifx\thechapter\undefined
      \newtheorem{defn}{\protect\definitionname}
    \else
      \newtheorem{defn}{\protect\definitionname}[chapter]
    \fi
\theoremstyle{plain}
    \ifx\thechapter\undefined
      \newtheorem{condition}{\protect\conditionname}
    \else
      \newtheorem{condition}{\protect\conditionname}[chapter]
    \fi
\theoremstyle{plain}
    \ifx\thechapter\undefined
      \newtheorem{prop}{\protect\propositionname}
    \else
      \newtheorem{prop}{\protect\propositionname}[chapter]
    \fi
\theoremstyle{plain}
    \ifx\thechapter\undefined
      \newtheorem{lem}{\protect\lemmaname}
    \else
      \newtheorem{lem}{\protect\lemmaname}[chapter]
    \fi
\theoremstyle{plain}
    \ifx\thechapter\undefined
	    \newtheorem{thm}{\protect\theoremname}
	  \else
      \newtheorem{thm}{\protect\theoremname}[chapter]
    \fi
\theoremstyle{plain}
    \ifx\thechapter\undefined
  \newtheorem{cor}{\protect\corollaryname}
\else
      \newtheorem{cor}{\protect\corollaryname}[chapter]
    \fi

\newref{cond}{refcmd={Condition \ref{#1}}}
\newref{assu}{refcmd={Assumption \ref{#1}}}
\newref{propref}{refcmd={Proposition \ref{#1}}}
\newref{corref}{refcmd={Corrolary \ref{#1}}}
\newref{lem}{refcmd={Lemma \ref{#1}}}
\newref{defref}{refcmd={Definition \ref{#1}}}
\newref{thmref}{refcmd={Theorem \ref{#1}}}

\makeatother

\providecommand{\assumptionname}{Assumption}
\providecommand{\conditionname}{Condition}
\providecommand{\corollaryname}{Corollary}
\providecommand{\definitionname}{Definition}
\providecommand{\lemmaname}{Lemma}
\providecommand{\propositionname}{Proposition}
\providecommand{\theoremname}{Theorem}

\global\long\def\at#1{\rvert_{#1}}

\global\long\def\norm#1{\left\Vert #1\right\Vert }

\global\long\def\mbe{\mathbb{E}}

\global\long\def\thetahat{\hat{\theta}}

\global\long\def\onevec{1_{w}}

\global\long\def\thetawiggle{\tilde{\theta}}

\global\long\def\thetaone{\hat{\theta}_{1}}

\global\long\def\thetaij{\hat{\theta}_{\textrm{IJ}}}

\newcommandx\thetaw[1][usedefault, addprefix=\global, 1=w]{\hat{\theta}\left(#1\right)}

\global\long\def\hone{H_{1}}

\global\long\def\constij{C_{\textrm{IJ}}}

\global\long\def\consth{C_{h}}

\global\long\def\constg{C_{g}}

\global\long\def\constw{C_{w}}

\newcommandx\constop[1][usedefault, addprefix=\global, 1={\text{\,}}]{C_{op}^{#1}}

\global\long\def\liph{L_{h}}

\global\long\def\thetasize{\Delta_{\theta}}

\global\long\def\deltasize{\Delta_{\delta}}

\global\long\def\thetaball{B_{\Delta_{\theta}}}

\global\long\def\thetadeltaball{B_{C_{op}\delta}}

\global\long\def\hint{\tilde{H}}

\usepackage[]{graphicx}
\usepackage[]{color}
\makeatletter
\def\maxwidth{ %
  \ifdim\Gin@nat@width>\linewidth
    \linewidth
  \else
    \Gin@nat@width
  \fi
}
\makeatother

\definecolor{fgcolor}{rgb}{0.345, 0.345, 0.345}

\usepackage{framed}
\makeatletter
 {\par\unskip\endMakeFramed%
 \at@end@of@kframe}
\makeatother

\definecolor{shadecolor}{rgb}{.97, .97, .97}
\definecolor{messagecolor}{rgb}{0, 0, 0}
\definecolor{warningcolor}{rgb}{1, 0, 1}
\definecolor{errorcolor}{rgb}{1, 0, 0}
\newenvironment{knitrout}{}{} 

\usepackage{alltt}

\global\long\def\wcv{W_k}
\global\long\def\wboot{W^*_B}
\global\long\def\thetapw{\thetahat(w)}
\global\long\def\wdiff{\Delta w}

\newcommand{\fig}[1]{Fig.~\ref{fig:#1}}

\newcommand{\sect}[1]{Section~\ref{sec:#1}}
\newcommand{\subsect}[1]{Section~\ref{subsec:#1}}

\newcommand{\corollary}[1]{Corollary~\ref{cor:#1}}

\newcommand{\appsect}[1]{Appendix~\ref{sec:#1}}

\newcommand{\coreassum}{Assumptions \ref{assu:paper_smoothness}--\ref{assu:paper_lipschitz} }
\newcommand{\paperallcoreassum}{Assumptions \ref{assu:paper_smoothness}--\ref{assu:paper_weight_bounded} }

\iftoggle{for_arxiv}{%
    \title{A Swiss Army Infinitesimal Jackknife}
    \author{
      Ryan Giordano\\ \texttt{rgiordano@berkeley.edu}
      \and
      Will Stephenson\\ \texttt{wtstephe@mit.edu}
      \and
      Runjing Liu\\ \texttt{runjing\_liu@berkeley.edu}
      \and
      Michael I.~Jordan\\ \texttt{jordan@cs.berkeley.edu}
      \and
      Tamara Broderick\\ \texttt{tbroderick@csail.mit.edu}
    }
}

\begin{document}

\iftoggle{for_arxiv} {
    \maketitle
}

\nottoggle{for_arxiv} {
\twocolumn[
    \aistatstitle{A Swiss Army Infinitesimal Jackknife}

    \aistatsauthor{
        Ryan Giordano \And
        Will Stephenson \And
        Runjing Liu \And
        Michael I.~Jordan \And
        Tamara Broderick
    }
    \aistatsaddress{
        UC Berkeley \And
        MIT \And
        UC Berkeley \And
        UC Berkeley \And
        MIT
    }
]
}

\begin{abstract}
The error or variability of machine learning algorithms is often assessed by
repeatedly re-fitting a model with different weighted versions of the observed
data. The ubiquitous tools of cross-validation (CV) and the bootstrap are
examples of this technique. These methods are powerful in large part due to
their model agnosticism but can be slow to run on modern, large data sets due to
the need to repeatedly re-fit the model. In this work, we use a linear
approximation to the dependence of the fitting procedure on the weights,
producing results that can be faster than repeated re-fitting by an order of
magnitude. This linear approximation is sometimes known as the ``infinitesimal
jackknife'' in the statistics literature, where it is mostly used as a
theoretical tool to prove asymptotic results. We provide explicit finite-sample
error bounds for the infinitesimal jackknife in terms of a small number of
simple, verifiable assumptions. Our results apply whether the weights and data
are stochastic or deterministic, and so can be used as a tool for proving the
accuracy of the infinitesimal jackknife on a wide variety of problems. As a
corollary, we state mild regularity conditions under which our approximation
consistently estimates true leave-$k$-out cross-validation for any fixed $k$.
These theoretical results, together with modern automatic differentiation
software, support the application of the infinitesimal jackknife to a wide
variety of practical problems in machine learning, providing a ``Swiss Army
infinitesimal jackknife.'' We demonstrate the accuracy of our methods on a range
of simulated and real datasets.

\end{abstract}

\section{Introduction}\label{sec:introduction}

Statistical machine learning methods are increasingly deployed in real-world
problem domains where they are the basis of decisions affecting individuals'
employment, savings, health, and safety. Unavoidable randomness in data
collection necessitates understanding how our estimates, and resulting
decisions, might have differed had we observed different data. Both cross
validation (CV) and the bootstrap attempt to diagnose this variation and are
widely used in classical data analysis. But these methods are often
prohibitively slow for modern, massive datasets, as they require running a
learning algorithm on many slightly different datasets.  In this work, we
propose to replace these many runs with a single perturbative approximation. We
show that the computation of this approximation is far cheaper than the
classical methods, and we provide theoretical conditions that establish its
accuracy.

Many data analyses proceed by minimizing a loss function of exchangeable data.
Examples include empirical loss minimization and M-estimation based on product
likelihoods. Since we typically do not know the true distribution generating the
data, it is common to approximate the dependence of our estimator on the data
via the dependence of the estimator on the empirical distribution. In
particular, we often form a new, proxy dataset using random or deterministic
modifications of the empirical distribution, such as randomly removing $k$
datapoints for leave-$k$-out CV. A proxy dataset obtained in this way can be
represented as a weighting of the original data. From a set of such proxy
datasets we can obtain estimates of uncertainty, including estimates of bias,
variance, and prediction accuracy.

As data and models grow, the cost of repeatedly solving a large optimization
problem for a number of different values of weights can become impractically
large. Conversely, though, larger datasets often exhibit greater regularity; in
particular, under fairly general conditions, limit laws based on independence
imply that an optimum exhibits diminishing dependence on any fixed set of data
points.  We use this observation to derive a linear approximation to resampling
that needs to be calculated only once, but which nonetheless captures the
variability inherent in the repeated computations of classical CV.  Our method
is an instance of the \emph{infinitesimal jackknife} (IJ), a general methodology
that was historically a precursor to cross-validation and the
bootstrap~\citep{jaeckel:1972:infinitesimal, efron:1982:jackknife}. Part of our
argument is that variants of the IJ should be reconsidered
for modern large-scale applications because, for smooth optimization problems,
the IJ can be calculated automatically with modern
automatic differentiation tools \citep{baydin:2015:automatic}.

By using this linear approximation, we incur the cost of forming and inverting a
matrix of second derivatives with size equal to the dimension of the parameter
space, but we avoid the cost of repeatedly re-optimizing the objective. As we
demonstrate empirically, this tradeoff can be extremely favorable in many
problems of interest.

Our approach aims to provide a felicitous union of two schools of thought. In
statistics, the IJ is typically used to prove normality or
consistency of other estimators~\citep{
fernholz:1983:mises,shao:1993:jackknifemestimator,shao:2012:jackknife}. However,
the conditions that are required for these asymptotic analyses to hold are
prohibitively restrictive for machine learning---specifically, they require
objectives with bounded gradients. A number of recent papers in machine learning
have provided related linear approximations for the special case of
leave-one-out cross-validation~\citep{KohL17, RadM18, BeiramiRST17}, though
their analyses lack the generality of the statistical perspective.

We combine these two approaches by modifying the proof of the Fr{\'e}chet
differentiability of M-estimators developed by \citet{clarke:1983:uniqueness}.
Specifically, we adapt the proof away from the question of Fr{\'e}chet
differentiability within the class of all empirical distributions to the
narrower problem of approximating the exact re-weighting on a particular dataset
with a potentially restricted set of weights.  This limitation of what we expect
from the approximation is crucial; it allows us to bound the error in terms of a
complexity measure of the set of derivatives of the observed objective function,
providing a basis for non-asymptotic applications in large-scale machine
learning, even for objectives with unbounded derivatives.  Together with modern
automatic differentiation tools, these results extend the use of the
IJ to a wider range of practical problems. Thus, our
``Swiss Army infinitesimal jackknife,'' like the famous Swiss Army knife, is a
single tool with many different functions.

\section{Methods and Results}\label{sec:methods}

\subsection{Problem definition}\label{subsec:problem_definition}
We consider the problem of estimating an unknown parameter
$\theta\in\Omega_{\theta}\subseteq\mathbb{R}^{D}$, with a compact
$\Omega_{\theta}$ and a dataset of size $N$. Our analysis will proceed entirely
in terms of a fixed dataset, though we will be careful to make assumptions that
will plausibly hold for all $N$ under suitably well-behaved random sampling. We
define our estimate, $\thetahat\in\Omega_{\theta}$, as the root of a weighted
estimating equation. For each $n=1, \ldots, N$, let $g_{n}\left(\theta\right)$ be a
function from $\Omega_{\theta}$ to $\mathbb{R}^{D}$. Let $w_{n}$ be a real
number, and let $w$ be the vector collecting the $w_n$.
Then $\thetahat$ is defined as the quantity that satisfies
\begin{align}
\thetapw:= &
    \quad\theta\textrm{ such that }
    \frac{1}{N}\sum_{n=1}^{N}w_{n}g_{n}\left(\theta\right) = 0.
    \label{eq:estimating_equation}
\end{align}
We will impose assumptions below that imply at least local uniqueness of
$\thetapw$; see the discussion following \assuref{paper_hessian} in
\subsect{assumptions}.

As an example, consider a family of continuously differentiable loss functions
$f\left(\cdot, \theta\right)$ parameterized by $\theta$ and evaluated at data
points $x_{n},n=1, \ldots, N$. If we want to solve the optimization problem
$\thetahat=\underset{\theta\in\Omega_{\theta}}
    {\mathrm{argmin}}\frac{1}{N}\sum_{n=1}^{N}f\left(x_{n},\theta\right),$
then we take $g_{n}\left(\theta\right)=\partial
f\left(x_{n},\theta\right)/\partial\theta$ and $w_{n}\equiv1$. By keeping our
notation general, we will be able to analyze a more general class of problems,
such as multi-stage optimization (see \sect{genomics}). However, to aid
intuition, we will sometimes refer to the $g_{n}\left(\theta\right)$ as
``gradients'' and their derivatives as ``Hessians.''

When \eqref{estimating_equation} is not degenerate (we articulate precise
conditions below), $\thetahat$ is a function of the weights through solving the
estimating equation, and we write $\thetapw$ to emphasize this.
We will focus on the case where we have solved \eqref{estimating_equation} for
the weight vector of all ones, $\onevec:=\left(1, \ldots, 1\right)$, which we denote
$\thetaone := \thetahat\left(\onevec\right)$.

A re-sampling scheme can be specified by choosing a set
$W\subseteq\mathbb{R}^{N}$ of weight vectors. For example, to approximate
leave-$k$-out CV, one repeatedly computes $\thetapw$ where $w$ has $k$ randomly
chosen zeros and all ones otherwise.  Define $\wcv$ as the set of every possible
leave-$k$-out weight vector.  Showing that our approximation is good for all
leave-$k$-out analyses with probability one is equivalent to showing that the
approximation is good for all $w \in \wcv$.

In the case of the bootstrap, $W$ contains a fixed number $B$ of randomly chosen
weight vectors,
$w_{b}^{*}\stackrel{iid}{\sim}\mathrm{Multinomial}\left(N,N^{-1}\right)$ for
$b=1, \ldots, B$, so that $\sum_{n=1}^N w_{bn}^{*} = N$ for each $b$. Note that
while $w_n$ or $w_{bn}^{*}$ are scalars, $w_{b}^{*}$ is a vector of length $N$.
The distribution of
$\thetahat\left(w_{b}^{*}\right)-\thetahat\left(\onevec\right)$ is then used to
estimate the sampling variation of $\thetaone$. Define this set $\wboot =
\{w_{1}^{*},\ldots,w_{B}^{*}\}$.  Note that $\wboot$ is stochastic and is a
subset of all weight vectors that sum to $N$.

In general, $W$ can be deterministic or stochastic, may contain integer or
non-integer values, and may be determined independently of the data or jointly
with it. As with the data, our results hold for a given $W$, but in a way that
will allow natural high-probability extensions to stochastic $W$.

\subsection{Linear approximation}

The main problem we solve is the computational expense involved in evaluating
$\thetapw$ for all the $w\in W$. Our contribution is to use only quantities
calculated from $\thetaone$ to approximate $\thetapw$ for all $w\in W$, without
re-solving \eqref{estimating_equation}. Our approximation is based on the
derivative $\frac{d\thetapw}{dw^{T}}$, whose existence depends on the
derivatives of $g_{n}\left(\theta\right)$, which we assume to exist, and which
we denote as $h_{n}\left(\theta\right):=\frac{\partial
g_{n}\left(\theta\right)}{\partial\theta^{T}}$. We use this notation because
$h_{n}\left(\theta\right)$ would be the Hessian of a term of the objective in
the case of an optimization problem. We make the following definition for
brevity.
\begin{defn}
The fixed point equation and its derivative are given respectively by
\begin{align*}
G\left(\theta,w\right) &:=
    \frac{1}{N}\sum_{n=1}^{N}w_{n}g_{n}\left(\theta\right) \\
H\left(\theta,w\right) &:=
    \frac{1}{N}\sum_{n=1}^{N}w_{n}h_{n}\left(\theta\right).
\end{align*}
\end{defn}
Note that $G\left(\thetapw,w\right)=0$ because $\thetapw$ solves
\eqref{estimating_equation} for $w$. We define
$\hone:=H\left(\thetaone,\onevec\right)$ and define the weight difference as
$\wdiff=w-\onevec\in\mathbb{R}^{N}$. When $\hone$ is invertible, one can use the
implicit function theorem and the chain rule to show that the derivative of
$\thetapw$ with respect to $w$ is given by
\begin{align*}
\frac{d\thetapw}{dw^{T}}\at{\onevec}\wdiff & =
    -\hone^{-1}\frac{1}{N}\sum_{n=1}^{N}g_{n}
    \left(\thetaone\right)\wdiff \\
    & =-\hone^{-1}G\left(\thetaone,\wdiff\right).
\end{align*}
This derivative allows us to form a first-order approximation to $\thetapw$
at $\thetaone$.
\begin{defn}
\label{defref:ij_definition}Our linear approximation to $\thetapw$
is given by
\begin{align*}
\thetaij\left(w\right) & :=\thetaone-\hone^{-1}G\left(\thetaone,\wdiff\right).
\end{align*}
\end{defn}
We use the subscript ``IJ'' for ``infinitesimal jackknife,'' which is the name
for this estimate in the statistics literature
\citep{jaeckel:1972:infinitesimal,shao:1993:jackknifemestimator}. Because
$\thetaij$ depends only on $\thetaone$ and $\wdiff$, and not on
solutions at any other values of $w$, there is no
need to re-solve \eqref{estimating_equation}. Instead, to calculate $\thetaij$
one must solve a linear system involving $\hone$. Recalling that $\theta$ is
$D$-dimensional, the calculation of $\hone^{-1}$ (or a factorization that
supports efficient solution of linear systems) can be $O\left(D^{3}\right)$.
However, once $\hone^{-1}$ is calculated or $\hone$ is factorized, calculating our approximation
$\thetaij\left(w\right)$ for each new weight costs only as
much as a single matrix-vector multiplication. Furthermore, $\hone$ often has a
sparse structure allowing $\hone^{-1}$ to be calculated more efficiently than a
worst-case scenario (see \sect{genomics} for an example). In more
high-dimensional examples with dense Hessian matrices, such as neural networks,
one may need to turn to approximations such as stochastic second-order methods
\citep{KohL17, agarwal:2016:lissa} and conjugate gradient
\citep{wright:1999:optimization}. Indeed, even in relatively small or sparse
problems, the vast bulk of the computation required to calculate $\thetaij$ is
in the computation of $\hone^{-1}$. We leave the important question of
approximate calculation of $\hone^{-1}$ for future work.

\subsection{Assumptions and results}\label{subsec:assumptions}

We now state our key assumptions and results, which are sufficient conditions
under which $\thetaij(w)$ will be a good approximation to $\thetapw$. We defer
most proofs to \appsect{appendix_proofs}. We use $\norm{\cdot}_{op}$
to denote the matrix operator norm, $\norm{\cdot}_{2}$ to denote the $L_{2}$
norm, and $\norm{\cdot}_{1}$ to denote the $L_{1}$ norm. For quantities like $g$
and $h$, which have dimensions $N\times D$ and $N\times D\times D$ respectively,
we apply the $L_p$ norm to the vectorized version of arrays.
For example,
$\frac{1}{\sqrt{N}} \norm{h\left(\theta\right)}_{2} =
\sqrt{
\frac{1}{N}\sum_{n=1}^{N}\sum_{i=1}^{D}\sum_{j=1}^{D}
\left[h_{n}\left(\theta\right)\right]_{ij}^{2}
}$ which is the square root of a sample average over $n\in[N]$.

We state all assumptions and results for a fixed $N$, a given estimating
equation vector $g\left(\theta\right)$, and a fixed class of weights $W$.
Although our analysis proceeds
with these quantities fixed, we are careful to make only assumptions that can
plausibly hold for all $N$ and/or for randomly chosen $W$
under appropriate regularity conditions.
\begin{assumption}[Smoothness] \label{assu:paper_smoothness}
For all $\theta\in\Omega_{\theta}$, each $g_{n}\left(\theta\right)$
is continuously differentiable in $\theta$.
\end{assumption}
The smoothness in \assuref{paper_smoothness} is necessary
for a local approximation like \defrefref{ij_definition} to have
any hope of being useful.
\begin{assumption}[Non-degeneracy]
\label{assu:paper_hessian}
For all $\theta\in\Omega_{\theta}$, $H\left(\theta,\onevec\right)$
is non-singular, with
$\sup_{\theta\in\Omega_{\theta}}\norm{H\left(\theta,\onevec\right)^{-1}}_{op}
\le\constop < \infty$.
\end{assumption}
Without \assuref{paper_hessian}, the derivative in \defrefref{ij_definition}
would not exist. For an optimization problem, \defrefref{ij_definition} amounts
to assuming that the Hessian is strongly positive definite, and, in general,
assures that the solution $\thetaone$ is unique.  Under our assumptions, we will
show later that, additionally, $\thetapw$ is unique in a neighborhood of
$\thetaone$; see \lemref{continuous_invertibility} of \appsect{appendix_proofs}.
Furthermore, by fixing $\constop$, if we want to apply \assuref{paper_hessian}
for $N\rightarrow\infty$, we will require that $\hone$ remains strongly positive
definite.

\begin{assumption}[Bounded averages] \label{assu:paper_bounded}
There exist finite constants $\constg$ and $\consth$ such that
$\sup_{\theta\in\Omega_{\theta}}
    \frac{1}{\sqrt{N}} \norm{g\left(\theta\right)}_{2}\le\constg<\infty
\quad\textrm{and}\quad
\sup_{\theta\in\Omega_{\theta}}
    \frac{1}{\sqrt{N}} \norm{h\left(\theta\right)}_{2} \le\consth<\infty$.
\end{assumption}
\assuref{paper_bounded} essentially states that the sample variances of the
gradients and Hessians are uniformly bounded. Note that it does not require that
these quantities are bounded term-wise. For example, we allow
$\sup_{n}\norm{g_{n}\left(\theta\right)}_2^2
\underset{N\rightarrow\infty}{\longrightarrow}\infty$, as long as
$\sup_{n}\frac{1}{N}\norm{g_{n}\left(\theta\right)}^{2}_2$ remains bounded. This
is a key advantage of the present work over many past applications of the IJ to
M-estimation, which require $\sup_n \norm{g_n(\theta)}_2^2$ to be uniformly
bounded for all $N$ \citep{shao:2012:jackknife, BeiramiRST17}.

In both machine learning and
statistics, $\sup_n\norm{g_n(\theta)}_2^2$ is rarely bounded, though
$\frac{1}{N}\norm{g(\theta)}_2^2$ often is.  As a simple example, suppose
that $\theta \in \mathbb{R}^1$, $x_n \sim \mathcal{N}(0, 1)$, and
$g_n = \theta - x_n$, as would arise from the squared error loss
$f_n\left(x_n, \theta\right) = \frac{1}{2}\left(\theta - x_n\right)^2$.
Fix a $\theta$ and let $N \rightarrow \infty$.  Then
$\sup_n\norm{g_n(\theta)}_2^2 \rightarrow \infty$ because
$\sup_n |x_n| \rightarrow \infty$, but
$\frac{1}{N}\norm{g(\theta)}_2^2 \rightarrow \theta^2 + 1$ by the law of
large numbers.
\begin{assumption}[Local smoothness] \label{assu:paper_lipschitz}
There exists a $\thetasize>0$ and a finite constant $\liph$ such that,
$\norm{\theta-\thetaone}_{2} \le \thetasize$ implies that
$\frac{\norm{h\left(\theta\right)-h\left(\thetaone\right)}_{2}}
    {\sqrt{N}}\le\liph\norm{\theta-\thetaone}_{2}$.
\end{assumption}
The constants defined in \assuref{paper_lipschitz} are needed to calculate
our error bounds explicitly.

\coreassum are quite general and should be expected to hold for many reasonable
problems, including holding uniformly asymptotically with high probability for
many reasonable data-generating distributions, as the following lemma shows.
\begin{lem}[The assumptions hold under uniform convergence]
\label{lem:assumptions_hold}
Let $\Omega_{\theta}$ be a compact set, and let
$g_{n}\left(\theta\right)$ be twice continuously differentiable IID random
functions for $n \in [N]$.  (The function is random but $\theta$ is not---for example,
$\mbe\left[g_n(\theta)\right]$ is still a function of $\theta$.)
Define
$r_{n}\left(\theta\right) :=
    \frac{\partial^{2}g{}_{n}\left(\theta\right)}
    {\partial\theta\partial\theta}$,
so $r_{n}\left(\theta\right)$ is a $D\times D\times D$ tensor.

Assume that we can exchange integration and differentiation, that
$\mbe\left[h_{n}\left(\theta\right)\right]$ is non-singular for
all $\theta\in\Omega_{\theta}$,
and that all of
$\mbe\left[\sup_{\theta\in\Omega_{\theta}}\norm{g_{n}\left(\theta\right)}_{2}^{2}\right]$,
$\mbe\left[\sup_{\theta\in\Omega_{\theta}}\norm{h_{n}\left(\theta\right)}_{2}^{2}\right]$,
and
$\mbe\left[\sup_{\theta\in\Omega_{\theta}}\norm{r_{n}\left(\theta\right)}_{2}^{2}\right]$
are finite.

Then $\lim_{N\rightarrow\infty}P\left(\textrm{\coreassum\ hold}\right)=1$.
\end{lem}
\lemref{assumptions_hold} follows from
the uniform convergence results of Theorems 9.1 and 9.2 in
\citet{keener:2011:theoretical}.
See \appsect{use_cases} for a detailed proof.  A common example to which
\lemref{assumptions_hold} would apply is where $x_n$ are well-behaved
IID data and $g_n(\theta) = \gamma(x_n, \theta)$ for an appropriately
smooth estimating function $\gamma(\cdot, \theta)$.
See \citet[Chapter 9]{keener:2011:theoretical} for more details and examples,
including applications to maximum likelihood estimators on unbounded domains.

\coreassum apply to the estimating equation.  We also require
a boundedness condition for $W$.
\begin{assumption}[Bounded weight averages] \label{assu:paper_weight_bounded}
The quantity
$\frac{1}{\sqrt{N}}\norm w_{2}$ is uniformly bounded for $w\in W$ by a finite
constant $\constw$.
\end{assumption}
Our final requirement is considerably more restrictive, and
contains the essence of whether or not $\thetaij(w)$ will be a good approximation
to $\thetapw$.
\begin{condition}[Set complexity]
\label{cond:paper_uniform_bound}There exists a $\delta\ge0$ and
a corresponding set $W_{\delta}\subseteq W$ such that
\begin{align*}
\max_{w\in W_{\delta}}\sup_{
    \theta\in\Omega_{\theta}} &\norm{\frac{1}{N}\sum_{n=1}^{N}
    \left(w_{n}-1\right)g_{n}\left(\theta\right)}_{1}  \le\delta
\quad\textrm{and} \\
\max_{w\in W_{\delta}}\sup_{
    \theta\in\Omega_{\theta}} &\norm{\frac{1}{N}\sum_{n=1}^{N}
    \left(w_{n}-1\right)h_{n}\left(\theta\right)}_{1}  \le\delta.
\end{align*}
\end{condition}
\condref{paper_uniform_bound} is central to establishing when the approximation
$\thetaij\left(w\right)$ is accurate. For a given $\delta$, $W_{\delta}$ will be
the class of weight vectors for which $\thetaij(w)$ is accurate to within order
$\delta$. Trivially, $\onevec\in W_{\delta}$ for $\delta=0$, so $W_{\delta}$ is
always non-empty, even for arbitrarily small $\delta$. The trick will be to
choose a small $\delta$ that still admits a large class $W_{\delta}$ of weight
vectors. In \sect{methods_examples} we will discuss
\condref{paper_uniform_bound} in more depth, but it will help to first state our
main theorem.
\begin{defn}
\label{defref:constants_definition}  The following constants are given by
quantities in \paperallcoreassum.
\begin{align*}
    \constij &:= 1+D \constw \liph \constop \\
    \deltasize &:=
        \min\left\{ \thetasize\constop[-1],
                    \frac{1}{2}\constij^{-1}\constop[-1]\right\}.
\end{align*}
\end{defn}
Note that, although the parameter dimension $D$ occurs explicitly only once in
\defrefref{constants_definition}, all of $\constw$, $\constop$, and $\liph$ in
general might also contain dimension dependence. Additionally, the bound
$\delta$ in \condref{paper_uniform_bound}, a measure of the set complexity of
the parameters, will typically depend on dimension. However, the particular
place where the parameter dimension enters will depend on the problem and
asymptotic regime, and our goal is to provide an adaptable toolkit for a wide
variety of problems.

We are now ready to state our main result.
\begin{thm}[Error bound for the approximation]
\label{thmref:paper_ij_error}Under
\paperallcoreassum and \condref{paper_uniform_bound},
\begin{align*}
\delta\le \deltasize \Rightarrow
\max_{w\in W_{\delta}}\norm{\thetaij\left(w\right)-{\thetapw}}_{2}
    \le 2 \constop[2] \constij \delta^{2}.
\end{align*}
\end{thm}
We stress that \thmrefref{paper_ij_error} bounds only the difference between
$\thetaij(w)$ and $\thetapw$.  \thmrefref{paper_ij_error} alone does not
guarantee that $\thetaij(w)$ converges to any hypothetical infinite population
quantity. We see this as a strength, not a weakness.  To begin with, convergence
to an infinite population requires stronger assumptions.  Contrast, for example,
the Fr{\'e}chet differentiability work of \citet{clarke:1983:uniqueness}, on which
our work is based, with the stricter requirements in the proof of consistency in
\cite{shao:1993:jackknifemestimator}.  Second, machine learning problems may not
naturally admit a well-defined infinite population, and the dataset at hand may
be of primary interest. Finally, by  analyzing a particular sample rather than a
hypothetical infinite population, we can bound the error in terms of the
quantities $\constij$ and $\deltasize$, which can actually be calculated from the
data at hand.

Still, \thmrefref{paper_ij_error} is useful to prove asymptotic
results about the difference $\norm{\thetaij\left(w\right)-{\thetapw}}_{2}$.
As an illustration, we now show that the uniform consistency of leave-$k$-out
CV follows from \thmrefref{paper_ij_error} by a straightforward
application of H{\"o}lder's inequality.
\begin{cor}[Consistency for leave-$k$-out CV]
\label{cor:paper_k_means_consistent}
Assume that \paperallcoreassum hold uniformly for all $N$. Fix an integer
$k$, and let
$$
W_{k}:=\left\{ w:w_{n}=0\textrm{ in }k\textrm{ entries and }1\textrm{ otherwise}\right\} .
$$
Then, for all $N$, there exists a constant $C_K$ such that
\begin{align*}
\sup_{w\in W_{k}}\norm{\thetaij\left(w\right)-{\thetapw}}_{2}
    & \le C_K \frac{\norm g_{\infty}^2}{N^2} \\
    & \le C_K \frac{\max\left\{\constg, \consth\right\}^2}{N}.
\end{align*}
\end{cor}
\begin{proof}
For $w\in W_{k}$, $\frac{\norm{\wdiff}_{2}}{\sqrt{N}}=\sqrt{\frac{K}{N}}.$
Define $C_{gh} := \max\left\{\constg, \consth\right\}$.
By \assuref{paper_bounded},
$\norm g_{2}/\sqrt{N}\le C_{gh}$ and $\norm h_{2}/\sqrt{N}\le C_{gh}$
for all $N$. By H{\"o}lder's inequality,
\begin{align*}
\lefteqn{ \sup_{w\in W}\sup_{\theta\in\Omega_{\theta}}
    \norm{\frac{1}{N}\sum_{n=1}^{N}\left(w_{n}-1\right)g_{n}
        \left(\theta\right)}_{1} } \\
    & \le \sup_{w\in W} \norm{w - \onevec}_1
    	\sup_{\theta\in\Omega_{\theta}} \frac{\norm g_{\infty}}{N}
    =
    	K \frac{\norm g_{\infty}}{N} \le
    	K \frac{C_{gh}}{\sqrt{N}},
\end{align*}
with a similar bound for $\norm h_{2}$. Consequently, for $N$ large
enough, \condref{paper_uniform_bound} is satisfied with $W_{\delta}=W_{k}$
and either $\delta=K \frac{\norm g_{\infty}}{N}$
or $\delta=K \frac{C_{gh}}{\sqrt{N}}$.
The result then follows from \thmrefref{paper_ij_error}.
\end{proof}

\section{Examples}\label{sec:methods_examples}

The moral of \thmrefref{paper_ij_error} is that, under \paperallcoreassum and
\condref{paper_uniform_bound},
$\norm{\thetaij-{\thetaw}}=O\left(\delta^{2}\right)$ for $w\in W_{\delta}$. That
is, if we can make $\delta$ small enough, $W_{\delta}$ big enough, and still
satisfy \condref{paper_uniform_bound}, then $\thetaij\left(w\right)$ is a good
approximation to $\thetaw$ for ``most'' $w$, where ``most'' is defined as the size
of $W_\delta$. So it is worth taking a moment to develop some intuition for
\condref{paper_uniform_bound}. We have already seen in
\corollary{paper_k_means_consistent} that $\thetaij$ is, asymptotically, a good
approximation for leave-$k$-out CV uniformly in $W$. We now discuss some additional
cases: first, a worst-case example for which $\thetaij$ is not expected to work,
second the bootstrap, and finally we revisit leave-one-out cross
validation in the context of these other two methods.

First, consider a pathological example. Let $W_{full}$ be the set of all weight
vectors that sum to $N$. Let
$n^{*}=\max_{n\in[N]}\norm{g_{n}\left(\thetaone\right)}_{1}$ be the index of the
gradient term with the largest $L_{1}$ norm, and let $w_{n^{*}}=N$ and $w_n=0$
for $n \ne n^{*}$. Then
\begin{align*}
&\sup_{\theta\in\Omega_{\theta}}
    \norm{\frac{1}{N}\sum_{n=1}^{N}
        \left(w_{n}-1\right)g_{n}\left(\theta\right)}_{1} \\
&\quad =\sup_{\theta\in\Omega_{\theta}}
    \norm{g_{n^{*}}\left(\theta\right)-
        \frac{1}{N}\sum_{n=1}^{N}g_{n}\left(\theta\right)}_{1}
    \ge\norm{g_{n^{*}}\left(\thetaone\right)}_{1}.
\end{align*}
(The last inequality uses the fact that $G\left(\thetaone,\onevec\right)=0$.) In
this case, unless the largest gradient,
$\norm{g_{n^{*}}\left(\thetaone\right)}_{1}$, is small,
\condref{paper_uniform_bound} will not be satisfied for small $\delta$, and we
would not expect $\thetaij$ to be a good estimate for $\thetaw$ for all $w\in
W_{full}$. The class $W_{full}$ is too expressive. In the language of
\condref{paper_uniform_bound}, for some small fixed $\delta$, $W_{\delta}$ will
be some very restricted subset of $W_{full}$ in most realistic situations.

Now, suppose that we are using $B$ bootstrap weights,
$w_{b}^{*}\stackrel{iid}{\sim} \mathrm{Multinomial}\left(N,N^{-1}\right)$ for
$b=1,...,B$, and analyzing an optimization problem as defined in
\subsect{problem_definition}.
For a given $w_{b}^{*}$, a dataset $x_{1}^*,...,x_{N}^*$ formed by taking
$w_{b,n}^{*}$ copies of datapoint $x_n$ is equivalent in distribution to
$N$ IID samples with replacement from the empirical distribution
on $\left(x_{1},...,x_{N}\right)$.  In this notation, we then have
\begin{align*}
&\frac{1}{N}\sum_{n=1}^{N}\left(w_{b}^{*}-1\right)g_{n}\left(\theta\right) = \\
&\quad \frac{1}{N}\sum_{n=1}^{N}\frac{\partial
    f\left(\theta,x_{n}^{*}\right)}{\partial\theta} -
        \frac{1}{N}\sum_{n=1}^{N}\frac{\partial
    f\left(\theta,x_{n}\right)}{\partial\theta}.
\end{align*}
In this case, \condref{paper_uniform_bound} is a uniform bound on a centered
empirical process of derivatives of the objective function. Note that estimating
sample variances by applying the IJ with bootstrap weights is equivalent to the
ordinary delta method based on an asymptotic normal approximation \citep[Chapter
21]{efron:1982:jackknife}.  In order to provide an approximation to the
bootstrap that retains benefits (such as the faster-than-normal convergence to
the true sampling distribution described by \citet{hall:2013:bootstrap}), one
must consider higher-ordered Taylor expansions of $\thetapw$.  We leave this for
future work.

Finally, let us return to leave-one-out CV. In this case, $w_{n}-1$ is nonzero
for exactly one entry. Again, we can choose to leave out the
adversarially-chosen $n^{*}$ as in the first pathological example.  However,
unlike the pathological example, the leave-one-out CV weights are constrained to
be closer to $\onevec$---specifically, we set $w_{n^{*}}=0$, and let $w$ be one
elsewhere. Then \condref{paper_uniform_bound} requires
$\sup_{\theta\in\Omega_{\theta}}\norm{\frac{1}{N}g_{n^{*}}\left(\theta\right)}_{1}  \le\delta.$
In contrast to the pathological example, this supremum will get smaller as $N$
increases as long as $\norm{g_{n^{*}}\left(\theta\right)}_{1}$ grows more slowly
than $N$. For this reason, we expect leave-one-out (and, indeed, leave-$k$-out
for fixed $k$) to be accurately approximated by $\thetaij$ in many cases of
interest, as stated in \corollary{paper_k_means_consistent}.

\section{Related Work}\label{sec:relatedwork}
Although the idea of forming a linear approximation to the re-weighting of an
M-estimator has a long history, we nevertheless contribute in a number of ways.
By limiting ourselves to approximating the exact reweighting on a particular
dataset, we both loosen the strict requirements from the statistical literature
and generalize the existing results from the machine learning literature.

The jackknife is often favored over the IJ in the
statistics literature because of the former's simple computational approach, as
well as perceived difficulties in calculating the necessary derivatives when
some of the parameters are implicitly defined via optimization
\citep[Chapter 2.1]{shao:2012:jackknife}
(though exceptions exist; see, e.g., \citet{wager:2014:confidence}).
The brute-force
approach of the jackknife is, however, a liability in large-scale machine
learning problems, which are generally extremely expensive to re-optimize.
Furthermore, and critically, the complexity and tedium of calculating the
necessary derivatives is entirely eliminated by modern automatic differentiation
\citep{baydin:2015:automatic, maclaurin:2015:autograd}.

Our work is based on the proof of the Fr{\'e}chet differentiability of
M-estimators of \citet{clarke:1983:uniqueness}. In classical statistics,
Fr{\'e}chet differentiability is typically used to describe the asymptotic
behavior of functionals of the empirical distribution in terms of a functional
\citep{mises:1947:asymptotic,fernholz:1983:mises}. Since
\citet{clarke:1983:uniqueness} was motivated by such asymptotic questions, he
studied the Fr{\'e}chet derivative evaluated at a continuous probability
distribution for function classes that included delta functions. This focus led
to the requirement of a bounded gradient.  However, unbounded gradients are
ubiquitous in both statistics and machine learning, and an essential
contribution of the current paper is to remove the need for bounded gradients.

There exist proofs of the consistency of the (non-infinitesimal) jackknife that
allow for unbounded gradients.  For example, it is possible that the proofs of
\citet{reeds:1978:jackknifing}, which require a smoothness assumption similar to
our \assuref{paper_lipschitz}, could be adapted to the IJ.
However, the results of \citet{reeds:1978:jackknifing}---as well as those of
\citet{clarke:1983:uniqueness} and subsequent applications such as those of
\citet{shao:2012:jackknife}---are asymptotic and applicable only to IID data. By
providing finite sample results for a fixed dataset and weight set, we are able
to provide a template for proving accuracy bounds for more generic probability
distributions and re-weighting schemes.

A number of recent machine learning papers have derived approximate linear
versions of leave-one-out estimators.  \citet{KohL17} consider approximating the
effect of leaving out one observation at a time to discover influential
observations and construct adversarial examples, but provide little supporting
theory. \citet{BeiramiRST17} provide rigorous proofs for an approximate
leave-one-out CV estimator; however, their estimator requires computing a new
inverse Hessian for each new weight at the cost of a considerable increase in
computational complexity.  Like the classical statistics literature,
\citet{BeiramiRST17} assume that the gradients are bounded for all $N$.  When
$\norm g_{\infty}^2$ in \corollary{paper_k_means_consistent} is finite for all
$N$, we achieve the same $N^{-2}$ rate claimed by \citet{BeiramiRST17} for
leave-one-out CV although we use only a single matrix inverse. \citet{RadM18}
also approximate leave-one-out CV, and prove tighter bounds for the error of
their approximation than we do, but their work is customized to leave-one-out CV
and makes much more restrictive assumptions (e.g., Gaussianity).

\section{Simulated Experiments}

We begin the empirical demonstration of our method on two simple generalized
linear models: logistic and Poisson regression.\footnote{Leave-one-out CV may
not be the most appropriate estimator of generalization error in this setting
\citep{rosset:2018:fixed}, but this section is intended only to provide simple
illustrative examples.} In each case, we generate a synthetic dataset $Z =
\{(x_n, y_n) \}_{n=1}^N$ from parameters $(\theta, b)$, where $\theta
\in \mathbb{R}^{100}$ is a vector of regression coefficients and $b \in
\mathbb{R}$ is a bias term. In each experiment, $x_n \in \mathbb{R}^{100}$ is
drawn from a multivariate Gaussian, and $y_n$ is a scalar drawn from a Bernoulli
distribution with the logit link or from a Poisson distribution with the
exponential link.

\begin{knitrout}
\definecolor{shadecolor}{rgb}{0.969, 0.969, 0.969}\color{fgcolor}\begin{figure}[!h]

{\centering \includegraphics[width=0.98\linewidth,height=0.784\linewidth]{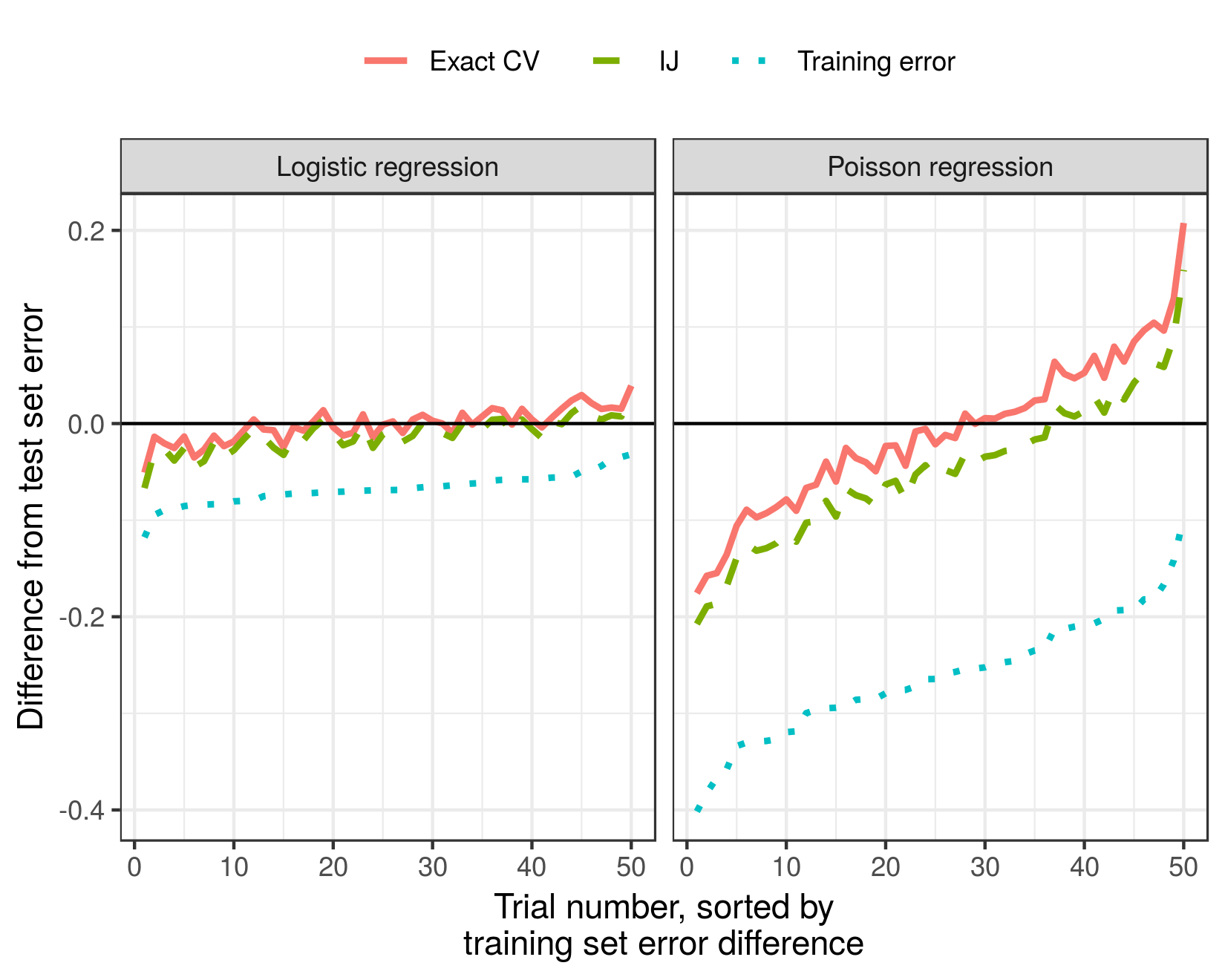} 

}

\caption[Simulated data]{Simulated data: accuracy results.}\label{fig:simulated_experiments_accuracy}
\end{figure}

\end{knitrout}
For a ground truth, we generate a large test set with $N=100{,}000$ datapoints
to measure the true generalization error. We show in
\fig{simulated_experiments_accuracy} that, over 50 randomly generated datasets,
our approximation consistently underestimates the actual error predicted by
exact leave-one-out CV; however, the difference is small relative to the
improvements they both make over the error evaluated on the training set.

\begin{knitrout}
\definecolor{shadecolor}{rgb}{0.969, 0.969, 0.969}\color{fgcolor}\begin{figure}[!h]

{\centering \includegraphics[width=0.98\linewidth,height=0.549\linewidth]{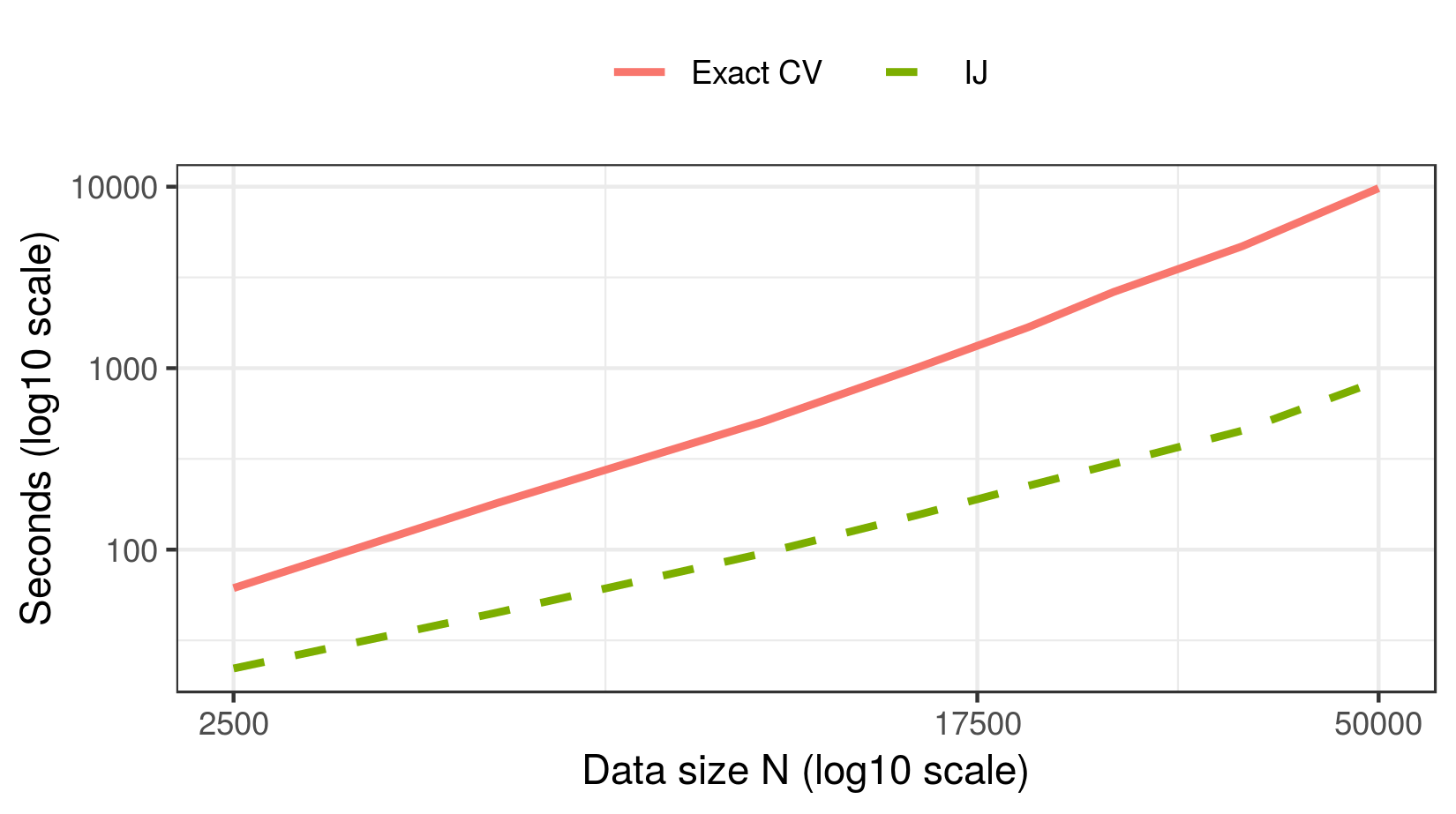} 

}

\caption[Simulated data]{Simulated data: timing results.}\label{fig:simulated_experiments_timing}
\end{figure}

\end{knitrout}
\fig{simulated_experiments_timing} shows the relative timings of our
approximation and exact leave-one-out CV on logistic regression with datasets of
increasing size. The time to run our approximation is roughly an order of
magnitude smaller.

\section{Genomics Experiments\label{sec:genomics}}

\newcommand{\approxopttime}{42}
\newcommand{\approxhesstime}{360}
\newcommand{\fullparamdim}{39,000}

\global\long\def\splinedegree{3}
\global\long\def\ntime{14}
\global\long\def\ngenes{1000}
\global\long\def\nclusters{18}
\global\long\def\covregularization{0.1}

We now consider a genomics application in which we use CV to choose the degree
of a spline smoother when clustering time series of gene expression data. Code and
instructions to reproduce our results can be found in the git repository
\href{https://github.com/rgiordan/AISTATS2019SwissArmyIJ}{rgiordan/AISTATS2019SwissArmyIJ}.
The application is also described in detail in \appsect{appendix_genomics}.

We use a publicly available data set of mice gene expression
\citep{shoemaker:2015:ultrasensitive} in which mice were infected with influenza
virus, and gene expression was assessed several times after infection. The
observed data consists of expression levels $y_{gt}$ for genes $g=1, \ldots,
n_{g}$ and time points $t=1, \ldots, n_{t}$.  In our case $n_{g}=\ngenes$ and
$n_{t}=\ntime$. Many genes behave the same way; thus, clustering the genes by
the pattern of their behavior over time allows dimensionality reduction that can
facilitate interpretation. Consequently, we wish to first fit a smoothed
regression line to each gene and then cluster the results. Following
\cite{Luan:2003:clustering}, we model the time series as a gene-specific
constant additive offset plus a B-spline basis of degree $\splinedegree$, and
the task is to choose the B-spline basis degrees of freedom using
cross-validation on the time points.

Our analysis runs in two stages---first, we regress the genes on the spline
basis, and then we cluster a transformed version of the regression fits. By
modeling in two stages, we both speed up the clustering and allow for the use of
flexible transforms of the fits. We are interested in choosing the smoothing
parameter using CV on the time points. Both the time points and the smoothing
parameter enter the regression objective directly, but they affect the
clustering objective only through the optimal regression parameters. Because the
optimization proceeds in two stages, the fit is not the optimum of any single
objective function.  However, it can still be represented as an M-estimator
(see \appsect{appendix_genomics}).

We implemented the model in \texttt{scipy} \citep{scipy} and computed all
derivatives with \texttt{autograd} \citep{maclaurin:2015:autograd}. We note that
the match between ``exact'' cross-validation (removing time points and
re-optimizing) and the IJ was considerably improved by
using a high-quality second-order optimization method.
In particular, for these
experiments, we employed the Newton conjugate-gradient trust region method
\citep[Chapter 7.1]{wright:1999:optimization} as implemented by the method
\texttt{trust-ncg} in \texttt{scipy.optimize}, preconditioned by the Cholesky
decomposition of an inverse Hessian calculated at an initial approximate
optimum.
The Hessian used for the preconditioner was with respect to
the clustering parameters only and so could be calculated quickly, in contrast
to the $\hone$ matrix used for the IJ, which includes the
regression parameters as well.
 We found that first-order or quasi-Newton
methods (such as BFGS) often got stuck or terminated at points with fairly large
gradients. At such points our method does not apply in theory nor, we found,
very well in practice.

\begin{knitrout}
\definecolor{shadecolor}{rgb}{0.969, 0.969, 0.969}\color{fgcolor}\begin{figure}[!h]

{\centering \includegraphics[width=0.98\linewidth,height=0.784\linewidth]{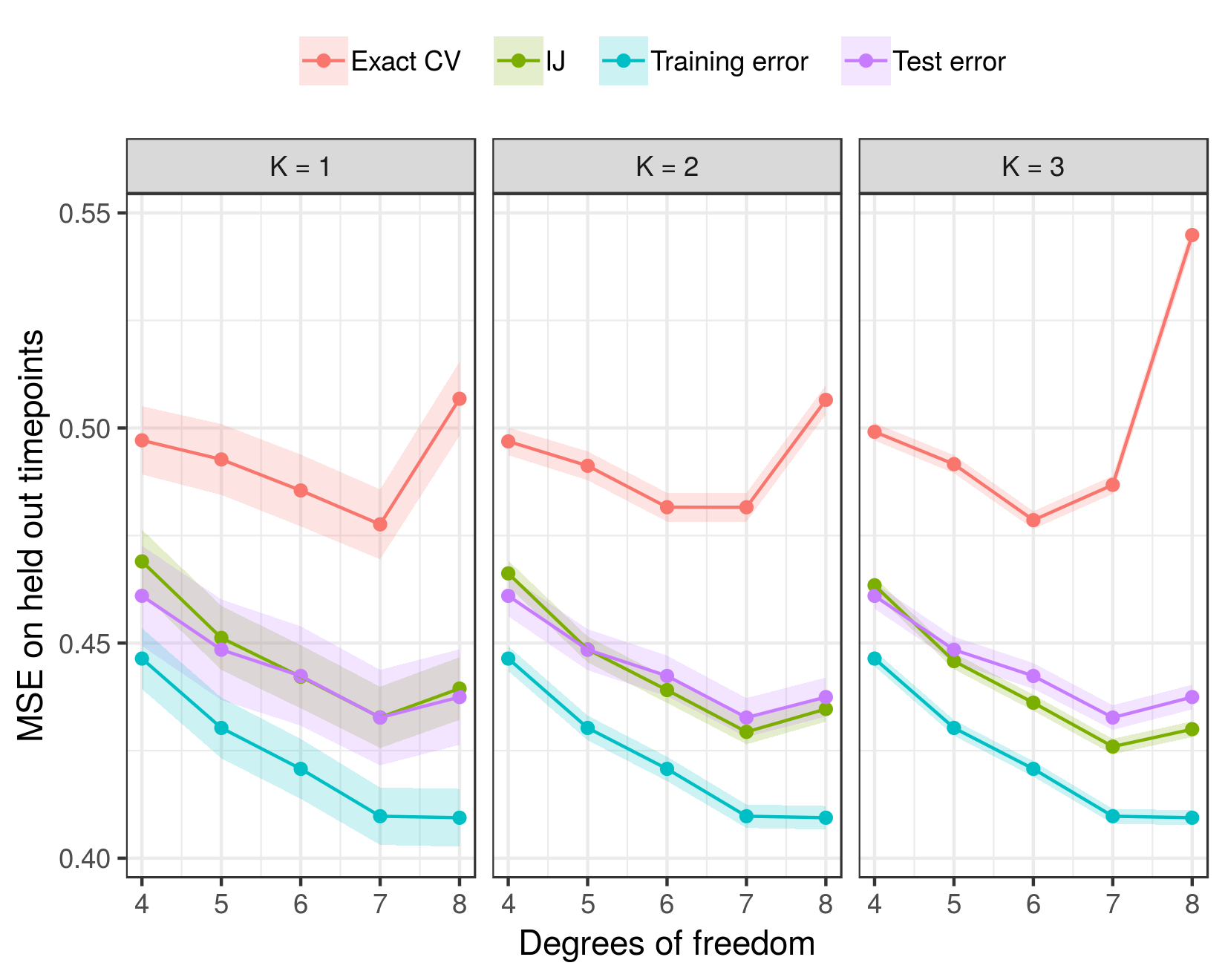} 

}

\caption[Genomics data]{Genomics data: accuracy results.}\label{fig:mse_graph}
\end{figure}

\end{knitrout}
\fig{mse_graph} shows that the IJ is a reasonably good approximation to the test
set error.\footnote{In fact, in this case, the IJ is a better predictor of test
set error than exact CV.  However, the authors have no reason at present to
believe that the IJ is a better predictor of test error than exact CV in
general.} In particular, both the IJ and exact CV capture the increase in test
error for $df=8$, which is not present in the training error. Thus we see that,
like exact CV, the IJ is able to prevent overfitting. Though the IJ
underestimates exact CV, we note that it differs from exact CV by no more than
exact CV itself differs from the true quantity of iterest, the test error.

\begin{knitrout}
\definecolor{shadecolor}{rgb}{0.969, 0.969, 0.969}\color{fgcolor}\begin{figure}[!h]

{\centering \includegraphics[width=0.98\linewidth,height=0.706\linewidth]{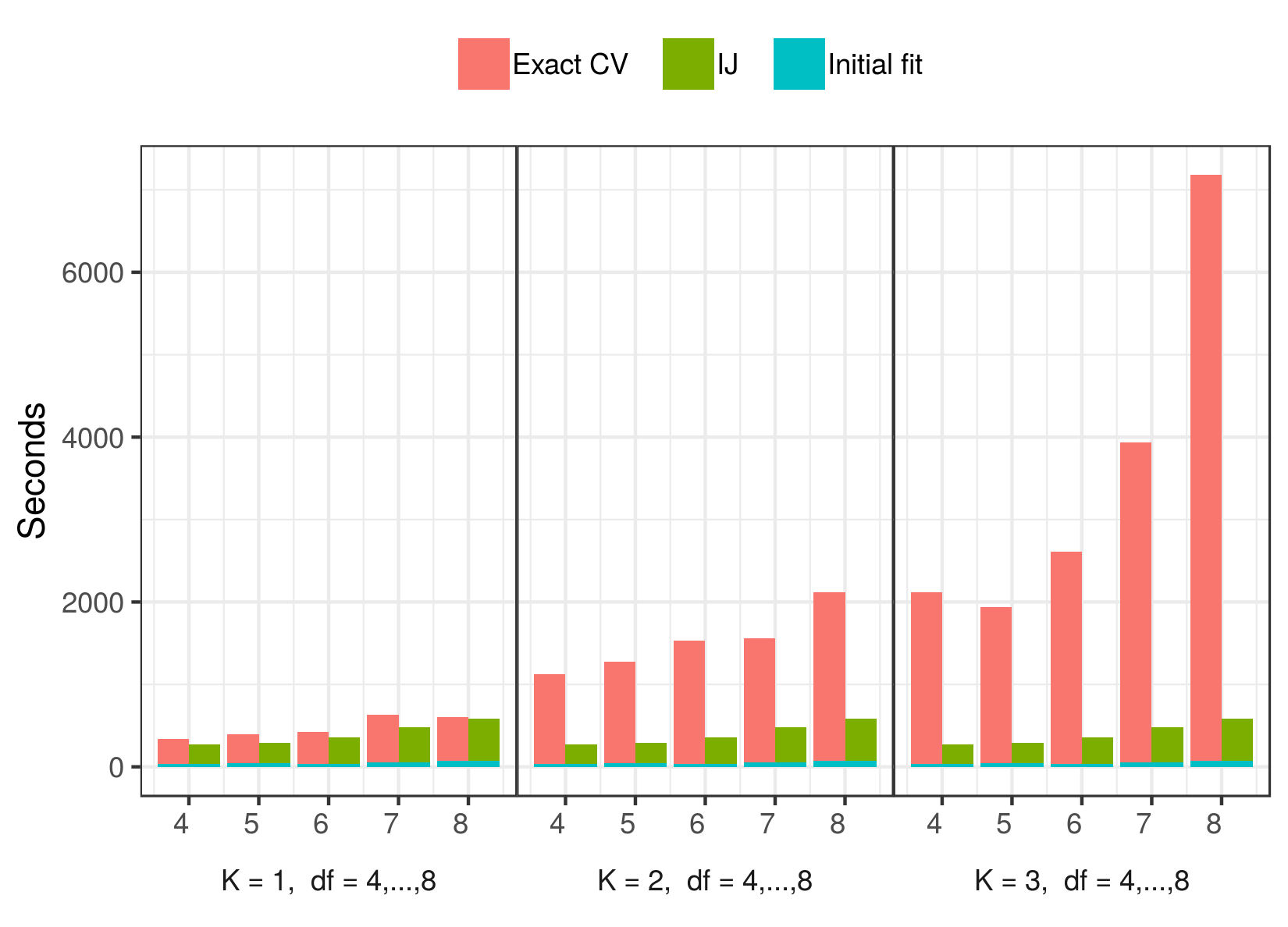} 

}

\caption[Genomics data]{Genomics data: timing results.}\label{fig:timing_graph}
\end{figure}

\end{knitrout}
The timing results for the genomics experiment are shown in \fig{timing_graph}.
For this particular problem with approximately $\fullparamdim$ parameters (the
precise number depends on the degrees of freedom), finding the initial optimum
takes about $\approxopttime$ seconds. The cost of finding the initial optimum is
shared by exact CV and the IJ, and, as shown in
\fig{timing_graph}, is a small proportion of both.

The principle time cost of the IJ is the
computation of $\hone$. Computing and inverting a dense matrix of size
$\fullparamdim$ would be computationally prohibitive.
But, for the regression objective, $\hone$ is extremely sparse and block diagonal, so computing
$\hone$ in this case took only around $\approxhesstime$ seconds.  Inverting $\hone$
took negligible time. Once we have $\hone^{-1}$, obtaining the
subsequent IJ approximations is nearly instantaneous.

The cost of refitting the model for exact CV varies by degrees of freedom
(increasing degrees of freedom increases the number of parameters) and the
number of left-out points (an increasing number of left-out datapoints increases
the number of refits). As can be seen in \fig{timing_graph}, for low degrees of
freedom and few left-out points, the cost of re-optimizing is approximately the
same as the cost of computing $\hone$.  However, as the degrees of freedom and
number of left-out points grow, the cost of exact CV increases to as much as an
order of magnitude more than that of the IJ.

\section{Conclusion}
We recommend consideration of the Swiss Army infinitesimal jackknife for modern
machine learning problems.  The large size of modern data both increases the
need for fast approximations and renders such approximations more accurate.
Furthermore, modern automatic differentiation renders many past practical
difficulties obsolete. By stepping back from the strict requirements of
classical statistical theory, we can see that the value of the infinitesimal
jackknife extends beyond its traditional application areas, while retaining
desirable generality in other respects.

\newpage
\paragraph{Acknowledgements.} We thank anonymous reviewers for their helpful
comments and suggestions. We are grateful to Nelle Varoquaux for her help with
the genomics experiments and to Pang Wei Koh for pointing out and helping to
correct an error in an earlier version of our proofs. This research was
supported in part by DARPA (FA8650-18-2-7832), an ARO YIP award, an NSF CAREER
award, and the CSAIL-MSR Trustworthy AI Initiative. Ryan Giordano was supported
by the Gordon and Betty Moore Foundation through Grant GBMF3834 and by the
Alfred P. Sloan Foundation through Grant 2013-10-27 to the University of
California, Berkeley. Runjing Liu was supported by the NSF Graduate Research
Fellowship.

\bibliography{references}
\bibliographystyle{plainnat}

\newpage
\onecolumn
\appendix

\newcommand{\globalassum}{Assumptions \ref{assu:paper_smoothness}---\ref{assu:paper_bounded} }

\section{Detailed assumptions, lemmas, and proofs\label{sec:appendix_proofs}}

\subsection{Tools}

We begin by stating two general propositions that will be useful.
First, we show that a version of Cauchy-Schwartz can be applied to
weighted sums of tensors.
%
\begin{prop}
\label{propref:tensor_cauchy_schwartz}Tensor array version of H{\"o}lder's
inequality. Let $w$ be an array of scalars and let $a=\left(a_{1},...,a_{N}\right)$
be an array of tensors, were each $a_{n}$ is indexed by $i=1,\ldots,D_{A}$
($i$ may be a multi-index---e.g., if $A$ is a $D\times D$ matrix,
then $i=\left(j,k\right)$, for $j,k\in\left[D\right]$ and $D_{A}=D^{2}$).
Let $p,q\in\left[1,\infty\right]$ be two numbers such that $p^{-1}+q^{-1}=1$.
Then
\begin{align*}
\norm{\frac{1}{N}\sum_{n=1}^{N}w_{n}a_{n}}_{1} & \le\frac{D_{A}^{\frac{1}{p}}}{N}\norm w_{p}\norm a_{q}.
\end{align*}
In particular, with $p=q=2$,
\begin{align*}
\norm{\frac{1}{N}\sum_{n=1}^{N}w_{n}a_{n}}_{1} & \le\sqrt{D_{A}}\frac{\norm w_{2}}{\sqrt{N}}\frac{\norm a_{2}}{\sqrt{N}}.
\end{align*}
\end{prop}
\begin{proof}
The conclusion follows from the ordinary H{\"o}lder's inequality
applied term-wise to $n$ and Jensen's inequality applied to the indices
$i$.

\begin{align*}
\norm{\frac{1}{N}\sum_{n=1}^{N}w_{n}a_{n}}_{1} & =\sum_{i=1}^{D_{A}}\left|\frac{1}{N}\sum_{n=1}^{N}w_{n}\left(a_{n}\right)_{i}\right|\\
 & \le\frac{1}{N}\sum_{i=1}^{D_{A}}\left|\left(\sum_{n=1}^{N}\left|w_{n}\right|^{p}\right)^{\frac{1}{p}}\left(\sum_{n=1}^{N}\left|\left(a_{n}\right)_{i}\right|^{q}\right)^{\frac{1}{q}}\right|\text{(H{\"o}lder)}\\
 & =\frac{1}{N}\norm w_{p}\frac{D_{A}}{D_{A}}\sum_{i=1}^{D_{A}}\left(\sum_{n=1}^{N}\left|\left(a_{n}\right)_{i}\right|^{q}\right)^{\frac{1}{q}}\\
 & \le\frac{1}{N}\norm w_{p}D_{A}\left(\frac{1}{D_{A}}\sum_{i=1}^{D_{A}}\sum_{n=1}^{N}\left|\left(a_{n}\right)_{i}\right|^{q}\right)^{\frac{1}{q}}\textrm{ (Jensen applied to }i\textrm{)}\\
 & =\frac{1}{N}\norm w_{p}D_{A}\left(\frac{1}{D_{A}}\sum_{n=1}^{N}\norm{a_{n}}_{q}^{q}\right)^{\frac{1}{q}}\\
 & =\frac{1}{N}\norm w_{p}D_{A}^{1-\frac{1}{q}}\norm a_{q}\\
 & =\frac{D_{A}^{\frac{1}{p}}}{N}\norm w_{p}\norm a_{q}.
\end{align*}
\end{proof}

%
Next, we prove a relationship between the term-wise difference between
matrices and the difference between their operator norms. It is well-known
that the minimum eigenvalue of a non-singular matrix is continuous
in the entries of the matrix. In the next proposition, we quantify
this continuity for the $L_{1}$ norm.


\begin{prop}
\label{propref:operator_norm_continuity}
Let $A$ and $B$ be two square matrices of the same size.
Let $\norm{A^{-1}}_{op}\le \constop$ for some finite $\constop$, and  Then
\begin{align*}
\norm{A-B}_{1}\le \frac{1}{2} (\constop)^{-1} &
    \quad\Rightarrow\quad\norm{B^{-1}}_{op} \le 2 \constop.
\end{align*}
\begin{proof}
We will use the results stated in Theorem 4.29 of \citet{schott:2016:matrix} and
the associated discussion in Example 4.14, which establish the following result.
Let $A$ be a square, nonsigular matrix, and let $I$ be the identity matrix of
the same size.  Let $\norm{\cdot}$ denote any matrix norm satisfying $\norm{I} =
1$.  Let $D$ be a matrix of the same size as $A$ satisfying
\begin{align}\eqlabel{ab_matrix_condition}
\norm{A^{-1}} \norm{D} \le 1.
\end{align}
Then
\begin{align}\label{eq:matrix_norm_continuity}
    \norm{A^{-1} - (A + D)^{-1}} \le
    \frac{\norm{A^{-1}}\norm{D}}{1 - \norm{A^{-1}\norm{D}}} \norm{A^{-1}}.
\end{align}
We will apply \eqref{matrix_norm_continuity} using the operator norm
$\norm{\cdot}_{op}$, for which $\norm I_{op}=1$ and with $D := B - A$.
Because $\norm{A^{-1}}_{op}\le \constop$, $A$ is invertible.

Assume that $\norm{A-B}_{1}\le \frac{1}{2} (\constop)^{-1}$.  First, note that
\begin{align}\label{eq:ab_matrix_condition_fulfilled}
\norm{A^{-1}}_{op} \norm{D}_{op} &=
    \norm{A^{-1}}_{op}\norm{B - A}_{op} \nonumber \\
&\le\norm{A^{-1}}_{op}\norm{B - A}_{1}
    & \textrm{(ordering of matrix norms)}\nonumber \\
 & \le \constop \frac{1}{2} (\constop)^{-1}
    & \textrm{(by assumption)} \nonumber \\
&= \frac{1}{2}  < 1,
\end{align}
so \eqref{ab_matrix_condition} is satisfied and we can apply
\eqref{matrix_norm_continuity}. Then
\begin{align*}
\norm{B^{-1}}_{op}
 & \le \norm{B^{-1}-A^{-1}}_{op} + \norm{A^{-1}}_{op}
    & \textrm{ (triangle inequality)}\\
 & \le \frac{\norm{A^{-1}}_{op}\norm{B - A}_{op}}
            {1 - \norm{A^{-1}}_{op}\norm{B - A}_{op}}
        \norm{A^{-1}}_{op} + \norm{A^{-1}}_{op}
     & \textrm{(\eqref{matrix_norm_continuity})}\\
 & \le \frac{\frac{1}{2}}{1-\frac{1}{2}}\norm{A^{-1}}_{op} +
    \norm{A^{-1}}_{op}
    &\textrm{(\eqref{ab_matrix_condition_fulfilled})} \\
 & \le 2 \constop.&\textrm{(by assumption)}
\end{align*}
\end{proof}
\end{prop}

Next, we define the quantities needed to make use of the integral form of
the Taylor series remainder.\footnote{We are indebted to Pang Wei Koh for
pointing out the need to use the integral form of the remainder for
Taylor series expansions of vector-valued functions.}

\begin{prop}
\label{propref:taylor_series_remainder}
For any $\theta \in \Omega_{\theta}$ and any $\tilde{w} \in W$,
\begin{align*}
G(\theta, \tilde{w}) - G(\thetaone, \tilde{w}) =&
    \left(\int_0^1 H(\thetaone + t (\theta - \thetaone), w) dt\right)
    \left(\theta - \thetaone\right)
\end{align*}
\end{prop}
\begin{proof}
Let $G_d(\theta, \tilde{w})$ denote the $d$-th component of the vector
$G(\theta, \tilde{w})$, and define the function $f_d(t) := G_d(\thetaone + t
(\thetaone - \theta), \tilde{w})$. The proposition follows by taking the
integral remainder form of the zero-th order Taylor series expansion of $f_d(t)$
around $t=0$ \citep[Appendix B.2]{dudley:2018:analysis}, and stacking the result
into a vector.
\end{proof}


The Taylor series residual of \proprefref{taylor_series_remainder} will show up
repeatedly, so we will give it a concise name in the following definition.


\begin{defn}
\label{defref:hess_integral}
For a fixed weight $w$ and a fixed parameter $\theta$, define the Hessian
integral
\begin{align*}
\hint(\theta, w) :=&
    \int_0^1 H(\thetaone + t (\theta - \thetaone), w) dt.
\end{align*}
\end{defn}



\subsection{Lemmas}

We now prove some useful consequences of our assumptions. The proof
roughly proceeds for all $w\in W_{\delta}$ by the following steps:
\begin{enumerate}
\item When $\delta$ is small we can make $\norm{{{\thetaw}}-\thetaone}_{2}$
small. (\lemref{theta_difference_bound} below.)
\item When $\norm{\theta-\thetaone}_{2}$ is small, then the derivatives
$H\left(\theta,w\right)$ are close to their optimal value, $H\left(\thetaone,\onevec\right)$.
(\lemref{bounded_continuous} and \lemref{gh_difference_from_one}
below.)
\item When the derivatives are close to their optimal values, then $H\left(\theta,w\right)$
is uniformly non-singular. (\lemref{continuous_invertibility} below.)
\item When the derivatives are close to their optimal values and $H\left(\theta,w\right)$
is uniformly non-singular we can control the error in $\thetaij-\thetaw$
in terms of $\delta$. (\thmrefref{taylor_error_first} below.)
\end{enumerate}
We begin by showing that the difference between $\thetaw$ and $\thetaone$
for $w\in W_{\delta}$ can be made small by making $\delta$ from
\condref{paper_uniform_bound} small.  First, however, we need to prove that
operator norm bounds on $H(\theta, w)$ also apply to the integrated Hessian
$\hint(\theta, w)$.


\begin{lem}\label{lem:hess_integral_invertible}
Invertibility of the integrated Hessian.
If, for some domain $\Omega$ and some constant $C$, $\sup_{\theta \in \Omega}
\norm{H(\theta, w)^{-1}}_{op} \le C$, then
$\sup_{\theta \in \Omega} \norm{\hint(\theta, w)^{-1}}_{op} \le C$.
\end{lem}

\begin{proof}
By definition of the operator norm,
\begin{align*}
\norm{\hint(\theta, w)^{-1}}_{op}^{-1} =&
    \min_{v \in \mathbb{R}^D: \norm{v}_2 = 1} v^T \hint(\theta, w) v \\
=& \min_{v \in \mathbb{R}^D: \norm{v}_2 = 1}
    \int_0^1 v^T H(\thetaone + t (\theta - \thetaone), w) v dt \\
\ge& \int_0^1 \min_{v \in \mathbb{R}^D: \norm{v}_2 = 1}
    v^T H(\thetaone + t (\theta - \thetaone), w) v dt \\
\ge& \inf_{\theta \in \Omega} \min_{v \in \mathbb{R}^D: \norm{v}_2 = 1}
        v^T H(\theta, w) v \\
\ge& C^{-1}.
\end{align*}
The result follows by inverting both sides of the inequality.
\end{proof}


\begin{lem}
\label{lem:theta_difference_bound}Small parameter changes. Under
\globalassum and \condref{paper_uniform_bound},
\begin{align*}
\textrm{for all }w\in W_{\delta},\quad\norm{{{\thetaw}}-\thetaone}_{2} &
    \le\constop\delta.
\end{align*}
\end{lem}
\begin{proof}
Applying \proprefref{taylor_series_remainder} with $\theta = \thetaw$
and $\tilde{w} = \onevec$ gives
\begin{align*}
G\left(\thetaw,\onevec\right) &
    =G\left(\thetaone,\onevec\right) +
    \hint\left(\thetaw, \onevec\right)
    \left(\thetaw-\thetaone\right).
\end{align*}
By \assuref{paper_hessian} and
\lemref{hess_integral_invertible},
$\sup_{\theta \in \Omega_\theta}
\norm{\hint(\theta,\onevec)^{-1}} \le \constop$.
In particular, $\hint(\theta,\onevec)$ is non-singular.
A little manipulation, together with the fact that
$G\left(\thetaw,w\right)=G\left(\thetaone,\onevec\right)=0$ gives
\begin{align*}
\thetaw-\thetaone & =
    \hint\left(\thetaw,\onevec\right)^{-1}
    \left(G\left(\thetaw, \onevec\right) - G\left(\thetaw,w\right)\right).
\end{align*}
Taking the norm of both sides gives
\begin{align*}
\norm{{{\thetaw}}-\thetaone}_{2} & =
    \norm{\hint\left(\thetaw,\onevec\right)^{-1}
        \left(G\left({{\thetaw}},\onevec\right) -
             G\left({{\thetaw}},w\right)\right)}_{2}\\
& \le\norm{\hint\left(\thetaw,\onevec\right)^{-1}}_{op}
    \norm{\left(G\left({{\thetaw}},\onevec\right) -
                G\left({{\thetaw}},w\right)\right)}_{2}\\
& \le\constop\norm{G\left({{\thetaw}},\onevec\right) -
                   G\left({{\thetaw}},w\right)}_{2}
    \textrm{ (Lemma \ref{lem:hess_integral_invertible})}\\
& \le\constop\norm{
    G\left({{\thetaw}},\onevec\right) -
    G\left({{\thetaw}},w\right)}_{1}\textrm{ (relation between norms)}\\
& \le\constop\sup_{\theta\in\Omega_{\theta}}
    \norm{G\left(\theta,\onevec\right)-G\left(\theta,w\right)}_{1}\\
& \le\constop\delta.\textrm{ (Condition \ref{cond:paper_uniform_bound}).}
\end{align*}
\end{proof}
%
Because we will refer to it repeatedly, we give the set of $\theta$
defined in \lemref{theta_difference_bound} a name.
%
\begin{defn}
For a given $\delta$, define the region around $\thetaone$ given
by \lemref{theta_difference_bound} as

\begin{align*}
\thetadeltaball & :=\left\{ \theta:\norm{\theta-\thetaone}_{2}\le\constop\delta\right\} \bigcap\Omega_{\theta}.
\end{align*}
\end{defn}
%
%
In other words, \lemref{theta_difference_bound} states that \condref{paper_uniform_bound}
implies $\thetaw\in\thetadeltaball$ when $w\in W_{\delta}$.

Next, we show that closeness in $\theta$ will mean closeness in $H\left(\theta,w\right)$.

\begin{lem}
\label{lem:bounded_continuous} Boundedness and continuity. Under
\paperallcoreassum and \condref{paper_uniform_bound},
\begin{align*}
\textrm{for all }\theta\in\thetaball,\quad\sup_{w\in W}
    \norm{H\left(\theta,w\right) - H\left(\thetaone,w\right)}_{1}
    &\le D\constw\liph\norm{\theta-\thetaone}_{2}.
\end{align*}
\end{lem}
\begin{proof}
For $\theta\in\thetaball$,
\begin{align*}
\sup_{w\in W}\norm{H\left(\theta,w\right)-H\left(\thetaone,w\right)}_{1} &=
    \sup_{w\in W}\norm{\frac{1}{N}
    \sum_{n=1}^{N}w_{n}\left(h_{n}\left(\theta\right) -h_{n}\left(\thetaone\right)\right)}_{1}\textrm{ (by definition)}\\
 & \le D\sup_{w\in W}\frac{\norm w_{2}}{\sqrt{N}}
    \frac{\norm{h\left(\theta\right) -
                h\left(\thetaone\right)}_{2}}{\sqrt{N}}
        \textrm{ (Proposition \ref{propref:tensor_cauchy_schwartz})}\\
 & \le D\constw\frac{\norm{h\left(\theta\right) -
                           h\left(\thetaone\right)}_{2}}{\sqrt{N}}
       \textrm{ (Assumption \ref{assu:paper_weight_bounded})}\\
 & \le D\constw\liph\norm{\theta-\thetaone}_{2}
    \textrm{ (Assumption \ref{assu:paper_lipschitz} and }
    \theta\in\thetaball\textrm{)}.
\end{align*}
\end{proof}
%
%
We now combine \lemref{theta_difference_bound} and \lemref{bounded_continuous}
to show that $H\left(\theta,w\right)$ is close to its value at the
solution $H\left(\thetaone,\onevec\right)$ for sufficiently small
$\delta$ and for all $\theta\in\thetadeltaball$.
%

\begin{lem}
\label{lem:gh_difference_from_one}Bounds for difference in parameters.
Under \paperallcoreassum and \condref{paper_uniform_bound}, if $\delta\le\thetasize\constop[-1]$,
then
\begin{align*}
\sup_{\theta\in\thetadeltaball}\sup_{w\in W_{\delta}} &
    \norm{H\left(\theta,w\right)-H\left(\thetaone,\onevec\right)}_{1}
    \le\left(1 + D\constw\liph\constop\right)\delta.
\end{align*}
\end{lem}
\begin{proof}
By \lemref{theta_difference_bound}, $\delta\le\thetasize\constop[-1]$
implies that $\constop\delta\le\thetasize$ and so $\thetadeltaball\subseteq\thetaball$.
Consequently, we can apply \lemref{bounded_continuous}:
\begin{align*}
\sup_{\theta\in\thetadeltaball}\sup_{w\in W_{\delta}}\norm{H\left(\theta,w\right)-H\left(\thetaone,w\right)}_{1} & \le\sup_{\theta\in\thetaball}\sup_{w\in W_{\delta}}\norm{H\left(\theta,w\right)-H\left(\thetaone,w\right)}_{1}\\
 & \le D\constw\liph\norm{\theta-\thetaone}_{2}\textrm{ (Lemma \ref{lem:bounded_continuous})}\\
 & \le D\constw\liph\constop\delta\quad\textrm{ (because }\theta\in\thetadeltaball\textrm{)}.
\end{align*}
Next, we can use this to write
\begin{align*}
\sup_{\theta\in\thetadeltaball}\sup_{w\in W_{\delta}} & \norm{H\left(\theta,w\right)-H\left(\thetaone,\onevec\right)}_{1}\\
 & =\sup_{\theta\in\thetadeltaball}\sup_{w\in W_{\delta}}\norm{H\left(\theta,w\right)-H\left(\theta,\onevec\right)+H\left(\theta,\onevec\right)-H\left(\thetaone,\onevec\right)}_{1}\\
 & \le\sup_{\theta\in\thetadeltaball}\sup_{w\in W_{\delta}}\norm{H\left(\theta,w\right)-H\left(\theta,\onevec\right)}_{1}+\sup_{\theta\in\thetadeltaball}\sup_{w\in W_{\delta}}\norm{H\left(\theta,\onevec\right)-H\left(\thetaone,\onevec\right)}_{1}\\
 & \le\sup_{\theta\in\Omega_{\theta}}\sup_{w\in W_{\delta}}\norm{H\left(\theta,w\right)-H\left(\theta,\onevec\right)}_{1}+\sup_{\theta\in\thetadeltaball}\sup_{w\in W_{\delta}}\norm{H\left(\theta,\onevec\right)-H\left(\thetaone,\onevec\right)}_{1}\\
 & \le\delta+\sup_{\theta\in\thetadeltaball}\sup_{w\in W_{\delta}}\norm{H\left(\theta,\onevec\right)-H\left(\thetaone,\onevec\right)}_{1}\textrm{ (Condition \ref{cond:paper_uniform_bound})}\\
 & \le\delta+D\constw\liph\constop\delta.
\end{align*}
\end{proof}

%
The constant that appears multiplying $\delta$ at the end of the proof of
\lemref{gh_difference_from_one} will appear often in what follows, so we give it
the special name $\constij$ in \defrefref{constants_definition}.

Note that \lemref{gh_difference_from_one} places a condition on how small
$\delta$ must be in order for our regularity conditions to apply.
\lemref{theta_difference_bound} will guarantee that $\thetaw\in\thetadeltaball$,
but if we are not able to make $\delta$ arbitrarily small in
\condref{paper_uniform_bound}, then we are not guaranteed to ensure that
$\thetadeltaball\subseteq\thetaball$, will not be able to assume Lipschitz
continuity, and none of our results will apply.

Next, using \lemref{gh_difference_from_one}, we can extend the operator bound on
$\hone^{-1}$ from \assuref{paper_hessian} to $H\left(\theta,w\right)^{-1}$ for
all $w\in W_{\delta}$.


\begin{lem}
\label{lem:continuous_invertibility}Uniform invertibility of the
Hessian. Under \paperallcoreassum and \condref{paper_uniform_bound}, if $\delta\le\min\left\{ \thetasize\constop[-1],\frac{1}{2}\constij^{-1}\constop[-1]\right\} $,
then
\begin{align*}
\sup_{\theta\in\thetadeltaball}\sup_{w\in W_{\delta}}\norm{H\left(\theta,w\right)^{-1}}_{op} & \le2\constop.
\end{align*}
\end{lem}
\begin{proof}
By \assuref{paper_hessian}, $\norm{H\left(\thetaone,\onevec\right)^{-1}}_{op}\le\constop$.
So by \proprefref{operator_norm_continuity}, it will suffice to select
$\delta$ so that
\begin{align}
\sup_{\theta\in\thetadeltaball}\sup_{w\in W_{\delta}}\norm{H\left(\theta,w\right)-H\left(\thetaone,\onevec\right)}_{1} & \le\frac{1}{2}\constop[-1].\label{eq:h_bound}
\end{align}
When we can apply \lemref{gh_difference_from_one}, we have
\begin{align*}
\sup_{\theta\in\thetadeltaball}\sup_{w\in W_{\delta}}\norm{H\left(\theta,w\right)-H\left(\thetaone,\onevec\right)}_{1} & \le\constij\delta.
\end{align*}
So $H\left(\theta,w\right)$ will satisfy \eqref{h_bound} if we can
apply \lemref{gh_difference_from_one} and if
\begin{align*}
\delta\le & \frac{1}{2}\constop[-1]\constij^{-1}.
\end{align*}
To apply \lemref{gh_difference_from_one} we additionally require
that $\delta\le\thetasize\constop[-1]$. By taking $\delta\le\min\left\{ \thetasize\constop[-1],\frac{1}{2}\constop[-1]\constij^{-1}\right\} $,
we satisfy \eqref{h_bound} and the result follows.
\end{proof}


Next, we show that a version of \lemref{gh_difference_from_one} also applies
to the integrated Hessian $\hint(\theta, w)$ when $\theta \in \thetadeltaball$.


\begin{lem}
\label{lem:int_h_difference_from_one}
Bounds for difference of the integrated Hessian.
Under \paperallcoreassum and \condref{paper_uniform_bound}, if $\delta\le\thetasize\constop[-1]$ and $\theta \in \thetadeltaball$,
\begin{align*}
\sup_{w\in W_{\delta}}\norm{
    \hint\left(\theta, w\right) -
    H(\thetaone, \onevec)}_{1}
\le& \left(1 + D\constw\liph\constop\right)\delta.
\end{align*}
\end{lem}
\begin{proof}
\begin{align*}
\MoveEqLeft
\sup_{w\in W_{\delta}}\norm{
    \hint\left(\theta, w\right) -
    H(\thetaone, \onevec)}_{1} &\\
=& \sup_{w\in W_{\delta}}\norm{
    \int_0^1 \left(
        H(\thetaone + t (\theta - \thetaone), w) dt -
        H(\thetaone, \onevec)\right)}_{1}
        &\textrm{ (Definition \ref{defref:hess_integral})} \\
\le&
    \sup_{w\in W_{\delta}}
    \int_0^1 \norm{
        H(\thetaone + t (\theta - \thetaone), w) -
        H(\thetaone, \onevec)}_{1} dt
        &\textrm{ (Jensen's inequality)} \\
\le&
    \sup_{\theta\in\thetadeltaball} \sup_{w\in W_{\delta}}
    \norm{H(\theta), w) - H(\thetaone, \onevec)}_{1} &\\
\le&
    \left(1 + D\constw\liph\constop\right)\delta
    &\textrm{ (Lemma \ref{lem:gh_difference_from_one})}
\end{align*}
\end{proof}


With these results in hand, the upper bound on $\delta$ will at last be
sufficient to control the error terms in our approximation. For compactness, we
give it the upper bound on $\delta$ the name $\deltasize$ in
\defrefref{constants_definition}.

Finally, we state a result that will allow us to define derivatives
of $\thetaw$ with respect to $w$.
\begin{lem}
\label{lem:implicit_function_theorem}Inverse function theorem. Under
\paperallcoreassum and \condref{paper_uniform_bound}, and for $\delta\le\deltasize$,
there exists a continuous, differentiable function of $w$, $\thetahat\left(w\right)$,
such that, for all $w\in W$, G$\left(\thetahat\left(w\right),w\right)=0$.
\end{lem}
\begin{proof}
This follows from \lemref{continuous_invertibility} and the implicit
function theorem.
\end{proof}
By definition, $\thetahat\left(\onevec\right)=\thetaone$.

\subsection{Bounding the errors in a Taylor expansion}

We are now in a position to use \paperallcoreassum and \condref{paper_uniform_bound}
to bound the error terms in a first-order Taylor expansion of $\thetaw$.
We begin by simply calculating the derivative $d\thetahat\left(w\right)/dw$.
\begin{prop}
\label{propref:theta_w_first_derivative}For any $w\in W$ for which
$H\left(\thetaw,w\right)$ is invertible, and for any vector $a\in\mathbb{R}^{N}$,
\begin{align*}
\frac{d\thetaw}{dw^{T}}\at wa & =-H\left(\thetaw,w\right)^{-1}G\left(\thetaw,a\right).
\end{align*}
\end{prop}
\begin{proof}
Because $G\left(\thetaw,w\right)=0$ for all $w\in W$, by direct
calculation,
\begin{align*}
0 & =\frac{d}{dw^{T}}G\left(\thetaw,w\right)\at wa\\
 & =\left(\frac{\partial G}{\partial\theta^{T}}\frac{d\hat{\theta}}{dw^{T}}+\frac{\partial G}{\partial w^{T}}\right)\at{_{w}}a\\
 & =H\left(\thetaw,w\right)\frac{d\hat{\theta}}{dw^{T}}\at{_{w}}a+\left(\frac{\partial}{\partial w^{T}}\frac{1}{N}\sum_{n=1}^{N}w_{n}g_{n}\left(\theta\right)\right)\at{_{w}}a\\
 & =H\left(\thetaw,w\right)\frac{d\hat{\theta}}{dw^{T}}\at{_{w}}a+\frac{1}{N}\sum_{n=1}^{N}g_{n}\left(\thetaw\right)a\\
 & =H\left(\thetaw,w\right)\frac{d\hat{\theta}}{dw^{T}}\at{_{w}}a+G\left(\thetaw,a\right).
\end{align*}
Because $H\left(\thetaw,w\right)$ is invertible by assumption, the
result follows.
\end{proof}
\begin{defn}
\label{defref:theta_infinitesimal_jackknife}Define
\begin{align*}
\thetaij\left(w\right) & :=\thetaone+\frac{d\thetaw}{dw^{T}}\at{\onevec}\left(w-\onevec\right)\\
 & =\thetaone-\hone^{-1}G\left(\thetaone,w\right).\textrm{ (because }G\left(\thetaone,\onevec\right)=0\textrm{)}
\end{align*}
\end{defn}
$\thetaij\left(w\right)$ in \defrefref{theta_infinitesimal_jackknife}
is the first term in a Taylor series expansion of $\thetaw$ as a
function of $w$. We want to bound the error, $\thetaij\left(w\right)-\thetaw$.
%
\begin{thm}
\label{thmref:taylor_error_first}

Under \paperallcoreassum and \condref{paper_uniform_bound},
when $\delta\le\deltasize$,
\begin{align*}
\sup_{w\in W_{\delta}}\norm{\thetaij\left(w\right)-\thetahat\left(w\right)}_{2} & \le2\constop[2]\constij\delta^{2}.
\end{align*}
\end{thm}
\begin{proof}
Applying \proprefref{taylor_series_remainder} with $\theta = \thetaw$ and
$\tilde{w} = w$, we have
\begin{align*}
0=G\left(\thetaw,w\right) & =
    G\left(\thetaone, w\right)
    +\hint\left(\thetaw, w\right) \left(\thetaw[w] - \thetaone\right).
\end{align*}
Because $\delta\in W_{\delta}$, \lemref{theta_difference_bound} implies that
$\thetaw\in\thetadeltaball$ so, \lemref{continuous_invertibility}
and \lemref{hess_integral_invertible} imply that
$\hint\left(\thetaw,w\right)$ is invertible and we can solve for
$\thetaw - \thetaone$.
\begin{align*}
\thetaw[w]-\thetaone & =
    -\hint\left(\thetaw, w\right)^{-1}G\left(\thetaone,w\right)\\
 & =\left(-\hint\left(\thetawiggle,w\right)^{-1} +
          H\left(\thetaone,\onevec\right)^{-1} -
          H\left(\thetaone,\onevec\right)^{-1}\right)
            G\left(\thetaone,w\right)\\
 & =\left(H\left(\thetaone,\onevec\right)^{-1} -
          \hint\left(\thetaw,w\right)^{-1}\right) G\left(\thetaone,w\right)+
    \thetaij\left(w\right)-\thetaone.
\end{align*}
Eliminating $\thetaone$ and taking the supremum of both sides gives
\begin{align*}
\sup_{w\in W_{\delta}} &
    \norm{\thetaij\left(w\right)-{{\thetaw[w]}}}_{2}\\
 & =\sup_{w\in W_{\delta}}\norm{
    \left(H\left(\thetaone,\onevec\right)^{-1} -
          \hint\left(\theta, w\right)^{-1}\right)
            G\left(\thetaone,w\right)}_{2}\\
 & =\sup_{w\in W_{\delta}}\norm{
    \hint\left(\thetaw, w\right)^{-1}\left(
        \hint\left(\thetaw,w\right) -
        H\left(\thetaone,\onevec\right)\right)
            H\left(\thetaone,\onevec\right)^{-1}
                G\left(\thetaone,w\right)}_{2}\\
 & \le2\constop\sup_{w\in W_{\delta}}\norm{
    \left(\hint\left(\thetaw,w\right) -
          H\left(\thetaone,\onevec\right)\right)
            H\left(\thetaone,\onevec\right)^{-1}
                G\left(\thetaone,w\right)}_{2}\textrm{ (Lemma
                    \ref{lem:hess_integral_invertible})}\\
 & \le2\constop\sup_{w\in W_{\delta}}\norm{
    \hint\left(\thetaw,w\right) -
    H\left(\thetaone,\onevec\right)}_{op}
    \norm{H\left(\thetaone,\onevec\right)^{-1}G\left(\thetaone,w\right)}_{2}\\
 & \le2\constop\sup_{w\in W_{\delta}}\norm{
    \hint\left(\thetaw,w\right) -
        H\left(\thetaone,\onevec\right)}_{1}
    \norm{H\left(\thetaone,\onevec\right)^{-1}G\left(\thetaone,w\right)}_{2}
    \textrm{ (ordering of matrix norms)} \\
 & \le 2\constop\constij\delta
    \sup_{w\in W_{\delta}}
        \norm{H\left(\thetaone,\onevec\right)^{-1}
              G\left(\thetaone,w\right)}_{2}
    \textrm{ (Lemma \ref{lem:int_h_difference_from_one})}\\
 & \le2\constop[2]\constij\delta
    \sup_{w\in W_{\delta}}\norm{G\left(\thetaone,w\right)}_{2}
    \textrm{ (Assumption \ref{assu:paper_hessian})}\\
 & =2\constop[2]\constij\delta
    \sup_{w\in W_{\delta}}
        \norm{G\left(\thetaone,w\right) -
              G\left(\thetaone,\onevec\right)}_{2}
    \textrm{ (because }G\left(\thetaone,\onevec\right)=0\textrm{)}\\
 & \le2\constop[2]\constij\delta^{2}
    \textrm{ (Condition \ref{cond:paper_uniform_bound})}.
\end{align*}

\end{proof}
%

\subsection{Use cases\label{sec:use_cases}}

First, let us state a simple condition under which \coreassum
hold. It will help to have a lemma for the Lipschitz continuity.
\begin{lem}
Derivative Cauchy Schwartz. Let $a\left(\theta\right)=\left(a_{1}\left(\theta\right),...,a_{N}\left(\theta\right)\right)$
be an array of tensors with multi-index $i\in\left[D_{A}\right]$,
and let $\frac{\partial a\left(\theta\right)}{\partial\theta}=\left(\frac{\partial}{\partial\theta}a_{1}\left(\theta\right),...,\frac{\partial}{\partial\theta}a_{N}\left(\theta\right)\right)$
be an array of tensors of size $D\times D_{A}$. Then
\begin{align*}
\norm{\frac{\partial}{\partial\theta}\norm{a\left(\theta\right)}_{2}}_{2} & \le D_{A}\norm{\frac{\partial a}{\partial\theta}}_{2}.
\end{align*}
\end{lem}
\begin{proof}
By direct calculation,
\begin{align*}
\norm{\frac{\partial}{\partial\theta}\norm{a\left(\theta\right)}_{2}^{2}}_{2}^{2} & =\sum_{r=1}^{D}\left(\frac{\partial}{\partial\theta_{r}}\sum_{n=1}^{N}\sum_{i=1}^{D_{A}}a_{n,i}\left(\theta\right)^{2}\right)^{2}\\
 & =\sum_{r=1}^{D}\left(\sum_{n=1}^{N}\sum_{i=1}^{D_{A}}2a_{n,i}\left(\theta\right)\frac{\partial a_{n,i}\left(\theta\right)}{\partial\theta_{r}}\right)^{2}\\
 & \le\sum_{r=1}^{D}\left(2\sum_{i=1}^{D_{A}}\left(\sum_{n=1}^{N}a_{n,i}\left(\theta\right)^{2}\right)^{\frac{1}{2}}\left(\sum_{n=1}^{N}\left(\frac{\partial a_{n,i}\left(\theta\right)}{\partial\theta_{r}}\right)^{2}\right)^{\frac{1}{2}}\right)^{2}\\
 & \le\sum_{r=1}^{D}\left(2D_{A}^{2}\left(\frac{1}{D_{A}}\sum_{i=1}^{D_{A}}\sum_{n=1}^{N}a_{n,i}\left(\theta\right)^{2}\right)^{\frac{1}{2}}\left(\frac{1}{D_{A}}\sum_{n=1}^{N}\left(\frac{\partial a_{n,i}\left(\theta\right)}{\partial\theta_{r}}\right)^{2}\right)^{\frac{1}{2}}\right)^{2}\\
 & =4D_{A}^{2}\norm a_{2}^{2}\sum_{r=1}^{D}\norm{\frac{\partial a}{\partial\theta_{r}}}_{2}^{2}\\
 & =4D_{A}^{2}\norm a_{2}^{2}\norm{\frac{\partial a}{\partial\theta}}_{2}^{2}.
\end{align*}
By the chain rule,
\begin{align*}
\norm{\frac{\partial}{\partial\theta}\norm{a\left(\theta\right)}_{2}}_{2}^{2} & =\frac{1}{4\norm{a\left(\theta\right)}_{2}^{2}}\norm{\frac{\partial}{\partial\theta}\norm{a\left(\theta\right)}_{2}^{2}}_{2}^{2}\le D_{A}^{2}\norm{\frac{\partial a}{\partial\theta}}_{2}^{2}.
\end{align*}
\end{proof}
\begin{lem}
\label{lem:lipschitz_helper}Let $a\left(\theta\right)\in\mathbb{R}^{D\times D}$
be a continuously differentiable random matrix with a $D\times D\times D$
derivative tensor. (Note that the function, not $\theta$, is random.
For example, $\mbe\left[a\left(\theta\right)\right]$ is still a function
of $\theta$.) Suppose that $\mbe\left[\norm{a\left(\theta\right)}_{2}\right]$
is finite for all $\theta\in\Omega_{\theta}$. Then, for all $\theta_{1},\theta_{2}\in\Omega_{\theta}$,
\begin{align*}
\left|\mbe\left[\norm{a\left(\theta_{1}\right)}_{2}\right]-\mbe\left[\norm{a\left(\theta_{2}\right)}_{2}\right]\right| & \le\sqrt{\mbe\left[\sup_{\theta\in\Omega_{\theta}}\norm{\frac{\partial a\left(\theta\right)}{\partial\theta}}_{2}^{2}\right]}\norm{\theta_{1}-\theta_{2}}_{2}.
\end{align*}
\end{lem}
\begin{proof}
For any tensor $a$ with multi-index $i$,
\begin{align*}
\norm{\frac{\partial}{\partial\theta}\norm a_{2}^{2}}_{2}^{2} & =\sum_{r=1}^{D}\left(\frac{\partial}{\partial\theta_{r}}\norm a_{2}^{2}\right)^{2}\\
 & =\sum_{r=1}^{D}\left(\frac{\partial}{\partial\theta_{r}}\sum_{i=1}^{D_{A}}a_{i}^{2}\right)^{2}\\
 & =\sum_{r=1}^{D}\left(2\sum_{i=1}^{D_{A}}a_{i}\frac{\partial a_{i}}{\partial\theta_{r}}\right)^{2}\\
 & \le4\sum_{r=1}^{D}\sum_{i=1}^{D_{A}}a_{i}^{2}\sum_{i=1}^{D_{A}}\left(\frac{\partial a_{i}}{\partial\theta_{r}}\right)^{2}\textrm{ (Cauchy-Schwartz)}\\
 & =4\sum_{i=1}^{D_{A}}a_{i}^{2}\sum_{r=1}^{D}\sum_{i=1}^{D_{A}}\left(\frac{\partial a_{i}}{\partial\theta_{r}}\right)^{2}\\
 & =4\norm a_{2}^{2}\norm{\frac{\partial a}{\partial\theta}}_{2}^{2}.
\end{align*}

Consequently,
\begin{align*}
\norm{\frac{\partial}{\partial\theta}\norm{a\left(\theta\right)}_{2}}_{2}^{2} & =\norm{\frac{1}{2\norm{a\left(\theta\right)}_{2}}\frac{\partial}{\partial\theta}\norm{a\left(\theta\right)}_{2}^{2}}_{2}^{2}\\
 & =\frac{1}{4\norm{a\left(\theta\right)}_{2}^{2}}\norm{\frac{\partial}{\partial\theta}\norm{a\left(\theta\right)}_{2}^{2}}_{2}^{2}\\
 & \le\frac{4\norm{a\left(\theta\right)}_{2}^{2}}{4\norm{a\left(\theta\right)}_{2}^{2}}\norm{\frac{\partial}{\partial\theta}a\left(\theta\right)}_{2}^{2}\\
 & =\norm{\frac{\partial a\left(\theta\right)}{\partial\theta}}_{2}^{2}.
\end{align*}
So for any $\theta_{1},\theta_{2}\in\Omega_{\theta}$,
\begin{align*}
\left|\mbe\left[\norm{a\left(\theta_{1}\right)}_{2}\right]-\mbe\left[\norm{a\left(\theta_{2}\right)}_{2}\right]\right| & \le\mbe\left[\left|\norm{a\left(\theta_{1}\right)}_{2}-\norm{a\left(\theta_{2}\right)}_{2}\right|\right]\\
 & \le\mbe\left[\left(\sup_{\theta\in\Omega_{\theta}}\norm{\frac{\partial}{\partial\theta}\norm{a\left(\theta\right)}_{2}}_{2}\right)\right]\norm{\theta_{1}-\theta_{2}}_{2}\textrm{ (}\theta\textrm{ is not random)}\\
 & \le\mbe\left[\left(\sup_{\theta\in\Omega_{\theta}}\norm{\frac{\partial a\left(\theta\right)}{\partial\theta}}_{2}\right)\right]\norm{\theta_{1}-\theta_{2}}_{2}\\
 & \le\sqrt{\mbe\left[\sup_{\theta\in\Omega_{\theta}}\norm{\frac{\partial a\left(\theta\right)}{\partial\theta}}_{2}^{2}\right]}\norm{\theta_{1}-\theta_{2}}_{2}.
\end{align*}
The result follows. Note that the bound still holds (though vacuously)
if $\mbe\left[\sup_{\theta\in\Omega_{\theta}}\norm{\frac{\partial a\left(\theta\right)}{\partial\theta}}_{2}^{2}\right]$
is infinite.
\end{proof}
\begin{prop}
\label{prop:assumptions_hold}Let $\Omega_{\theta}$ be a compact
set. Let $g_{n}\left(\theta\right)$ be twice continuously differentiable
IID random functions. Define
\begin{align*}
h_{n}\left(\theta\right) & :=\frac{\partial g{}_{n}\left(\theta\right)}{\partial\theta}\\
r_{n}\left(\theta\right) & :=\frac{\partial^{2}g{}_{n}\left(\theta\right)}{\partial\theta\partial\theta},
\end{align*}
where $r_{n}\left(\theta\right)$ is a $D\times D\times D$ tensor.
Assume that

1a) $\mbe\left[\sup_{\theta\in\Omega_{\theta}}\norm{g_{n}\left(\theta\right)}_{2}^{2}\right]<\infty$;

1b) $\mbe\left[\sup_{\theta\in\Omega_{\theta}}\norm{h_{n}\left(\theta\right)}_{2}^{2}\right]<\infty$;

1c) $\mbe\left[\sup_{\theta\in\Omega_{\theta}}\norm{r_{n}\left(\theta\right)}_{2}^{2}\right]<\infty;$

2) $\mbe\left[h_{n}\left(\theta\right)\right]$ is non-singular for
all $\theta\in\Omega_{\theta}$;

3) We can exchange expectation and differentiation.

Then $\lim_{N\rightarrow\infty}P\left(\textrm{\coreassum\ hold}\right)=1.$
\end{prop}
\begin{proof}
The proof follows from Theorems 9.1 and
9.2 of \citet{keener:2011:theoretical}. We will first show that the expected values of
the needed functions satisfy \coreassum, and then that the sample versions
converge uniformly.

By Jensen's inequality,
\begin{align*}
\mbe\left[\sup_{\theta\in\Omega_{\theta}}\norm{g_{n}\left(\theta\right)}_{2}\right] & =\mbe\left[\sqrt{\sup_{\theta\in\Omega_{\theta}}\norm{g_{n}\left(\theta\right)}_{2}^{2}}\right]\le\sqrt{\mbe\left[\sup_{\theta\in\Omega_{\theta}}\norm{g_{n}\left(\theta\right)}_{2}^{2}\right]}.
\end{align*}
Also, for the $i^{th}$ component of $g_{n}\left(\theta\right)$
\begin{align*}
\mbe\left[\sup_{\theta\in\Omega_{\theta}}\left|g_{n,i}\left(\theta\right)\right|\right] & \le\mbe\left[\sup_{\theta\in\Omega_{\theta}}\norm{g_{n}\left(\theta\right)}_{\infty}\right]\le\mbe\left[\sup_{\theta\in\Omega_{\theta}}\norm{g_{n}\left(\theta\right)}_{2}\right].
\end{align*}
By Theorem 9.1 of \citet{keener:2011:theoretical}, $\mbe\left[\norm{g_{n}\left(\theta\right)}_{2}^{2}\right]$
, $\mbe\left[\norm{g_{n}\left(\theta\right)}_{2}\right]$, and $\mbe\left[g_{n}\left(\theta\right)\right]$
are continuous functions of $\theta$, and because $\Omega_{\theta}$
is compact, they are each bounded. Similar reasoning applies to $h_{n}\left(\theta\right)$
and $r_{n}\left(\theta\right)$. Consequently we can define
\begin{align*}
\sup_{\theta\in\Omega_{\theta}}\mbe\left[\norm{g_{n}\left(\theta\right)}_{2}^{2}\right] & =:Q_{g}^{2}<\infty\\
\sup_{\theta\in\Omega_{\theta}}\mbe\left[\norm{h_{n}\left(\theta\right)}_{2}^{2}\right] & =:Q_{h}^{2}<\infty.
\end{align*}
Below, these constants will be used to satisfy \assuref{paper_smoothness}
and \assuref{paper_bounded} with high probability.

Because $\Omega_{\theta}$ is compact, $\mbe\left[h_{n}\left(\theta\right)\right]$
is continuous, $\mbe\left[h_{n}\left(\theta\right)\right]$ is non-singular,
and the operator norm is a continuous function of $\mbe\left[h_{n}\left(\theta\right)\right]$,
we can also define
\begin{align*}
\sup_{\theta\in\Omega_{\theta}}\norm{\mbe\left[h_{n}\left(\theta\right)\right]^{-1}}_{op} & =:Q_{op}<\infty.
\end{align*}
Below, this constant be used to satisfy \assuref{paper_hessian}
with high probability.

Finally, we turn to the Lipschitz condition. \lemref{lipschitz_helper}
implies that
\begin{align*}
\left|\mbe\left[\norm{h_{n}\left(\theta_{1}\right)}_{2}\right]-\mbe\left[\norm{h_{n}\left(\theta_{2}\right)}_{2}\right]\right| & \le\sqrt{\mbe\left[\sup_{\theta\in\Omega_{\theta}}\norm{r_{n}\left(\theta\right)}_{2}^{2}\right]}\norm{\theta_{1}-\theta_{2}}_{2}.
\end{align*}
Define
\begin{align*}
\Lambda_{h} & =\sqrt{\mbe\left[\sup_{\theta\in\Omega_{\theta}}\norm{r_{n}\left(\theta\right)}_{2}^{2}\right]},
\end{align*}
so that we have shown that $\mbe\left[\norm{h_{n}\left(\theta\right)}_{2}\right]$
is Lipschitz in $\Omega_{\theta}$ with constant $\Lambda_{h}$, which
is finite by assumption.

We have now shown, essentially, that the expected versions of the
quantities we wish to control satisfy \coreassum with $N=1$.
We now need to show that the sample versions satisfy \coreassum
with high probability, which will follow from the fact that the sample
versions converge uniformly to their expectations by
Theorem 9.2 of \citet{keener:2011:theoretical}.

First, observe that \assuref{paper_smoothness} holds with probability
one by assumption. For the remaining assumption choose an $\epsilon>0$,
and define
\begin{align*}
\constg & :=\sqrt{Q_{g}^{2}+\epsilon}\\
\consth & :=\sqrt{Q_{h}^{2}+\epsilon}\\
\constop & :=2Q_{op}\\
\liph & :=\sqrt{D^{4}\Lambda_{h}^{2}+\epsilon}.
\end{align*}

By \citet{keener:2011:theoretical} Theorem 9.2,
\begin{align*}
\sup_{\theta\in\Omega_{\theta}}\left|\frac{1}{N}\sum_{n=1}^{N}\norm{g_{n}\left(\theta\right)}_{2}^{2}-\mbe\left[\norm{g_{n}\left(\theta\right)}_{2}^{2}\right]\right| & \xrightarrow[N\rightarrow\infty]{p}0.
\end{align*}
Because
\begin{align*}
\sup_{\theta\in\Omega_{\theta}}\left|\frac{1}{N}\sum_{n=1}^{N}\norm{g_{n}\left(\theta\right)}_{2}^{2}\right| & >Q_{g}^{2}+\epsilon\ge\sup_{\theta\in\Omega_{\theta}}\mbe\left[\norm{g_{n}\left(\theta\right)}_{2}^{2}\right]+\epsilon\Rightarrow\\
\sup_{\theta\in\Omega_{\theta}}\left|\frac{1}{N}\sum_{n=1}^{N}\norm{g_{n}\left(\theta\right)}_{2}^{2}-\mbe\left[\norm{g_{n}\left(\theta\right)}_{2}^{2}\right]\right| & >\epsilon,
\end{align*}

we have
\begin{align*}
 & P\left(\sup_{\theta\in\Omega_{\theta}}\left|\frac{1}{N}\sum_{n=1}^{N}\norm{g_{n}\left(\theta\right)}_{2}^{2}\right|\ge Q_{g}^{2}+\epsilon\right)\le\\
 & \quad P\left(\sup_{\theta\in\Omega_{\theta}}\left|\frac{1}{N}\sum_{n=1}^{N}\norm{g_{n}\left(\theta\right)}_{2}^{2}-\mbe\left[\norm{g_{n}\left(\theta\right)}_{2}^{2}\right]\right|\le\epsilon\right),
\end{align*}
so
\begin{align*}
P\left(\sup_{\theta\in\Omega_{\theta}}\left|\frac{1}{N}\sum_{n=1}^{N}\norm{g_{n}\left(\theta\right)}_{2}^{2}\right|\ge\constg^{2}\right) & \xrightarrow[N\rightarrow\infty]{}0.
\end{align*}
An analogous argument holds for $\frac{1}{N}\norm{h_{n}\left(\theta\right)}_{2}^{2}$.
Consequently, $P\left(\textrm{Assumption \ref{assu:paper_bounded} holds}\right)\xrightarrow[N\rightarrow\infty]{}1$.

We now consider \assuref{paper_hessian}. Again, by \citet{keener:2011:theoretical} Theorem 9.2
applied to each element of the matrix $h_{n}\left(\theta\right)$,
using a union bound over each of the $D^{2}$ entries,
\begin{align*}
\sup_{\theta\in\Omega_{\theta}}\norm{\frac{1}{N}\sum_{n=1}^{N}h_{n}\left(\theta\right)-\mbe\left[h_{n}\left(\theta\right)\right]}_{1} & \xrightarrow[N\rightarrow\infty]{p}0.
\end{align*}
By the converse of \proprefref{operator_norm_continuity}, because
$\norm{\mbe\left[h_{n}\left(\theta\right)\right]^{-1}}_{op}\le Q_{op}$,
\begin{align*}
\norm{\left(\frac{1}{N}\sum_{n=1}^{N}h_{n}\left(\theta\right)\right)^{-1}}_{op} & >2Q_{op}=\constop\Rightarrow\\
\norm{\frac{1}{N}\sum_{n=1}^{N}h_{n}\left(\theta\right)-\mbe\left[h_{n}\left(\theta\right)\right]}_{1} & >\frac{1}{2}Q_{op}^{-1}.
\end{align*}
Consequently,
\begin{align*}
 & P\left(\norm{\left(\frac{1}{N}\sum_{n=1}^{N}h_{n}\left(\theta\right)\right)^{-1}}_{op}\ge\constop\right)\le\\
 & \quad P\left(\norm{\frac{1}{N}\sum_{n=1}^{N}h_{n}\left(\theta\right)-\mbe\left[h_{n}\left(\theta\right)\right]}_{1}\right)\xrightarrow[N\rightarrow\infty]{p}0,
\end{align*}
and $P\left(\textrm{Assumption \ref{assu:paper_hessian} holds}\right)\xrightarrow[N\rightarrow\infty]{}1.$

Finally, applying \lemref{lipschitz_helper} to $\frac{1}{\sqrt{N}}\norm{h\left(\theta_{2}\right)}_{2}$,
\begin{align*}
\left|\frac{1}{\sqrt{N}}\norm{h\left(\theta_{1}\right)}_{2}-\frac{1}{\sqrt{N}}\norm{h\left(\theta_{2}\right)}_{2}\right| & \le\sup_{\theta\in\Omega_{\theta}}\norm{\frac{\partial}{\partial\theta}\frac{1}{\sqrt{N}}\norm{h\left(\theta\right)}_{2}}_{2}\norm{\theta_{1}-\theta_{2}}_{2}\\
 & \le\frac{D^{2}}{\sqrt{N}}\sup_{\theta\in\Omega_{\theta}}\norm{r\left(\theta\right)}_{2}\norm{\theta_{1}-\theta_{2}}_{2}\\
 & =D^{2}\sqrt{\sup_{\theta\in\Omega_{\theta}}\frac{1}{N}\norm{r\left(\theta\right)}_{2}^{2}}\norm{\theta_{1}-\theta_{2}}_{2}.
\end{align*}
Consequently,
\begin{align*}
\left|\frac{1}{\sqrt{N}}\norm{h\left(\theta_{1}\right)}_{2}-\frac{1}{\sqrt{N}}\norm{h\left(\theta_{2}\right)}_{2}\right| & \ge\liph\norm{\theta_{1}-\theta_{2}}_{2}\Rightarrow\\
D^{2}\sqrt{\sup_{\theta\in\Omega_{\theta}}\frac{1}{N}\norm{r\left(\theta\right)}_{2}^{2}} & \ge\liph\Rightarrow\\
\sup_{\theta\in\Omega_{\theta}}\frac{1}{N}\norm{r\left(\theta\right)}_{2}^{2}-\sup_{\theta\in\Omega_{\theta}}\mbe\left[\norm{r_{n}\left(\theta\right)}_{2}^{2}\right] & \ge\frac{\liph^{2}}{D^{4}}-\sup_{\theta\in\Omega_{\theta}}\mbe\left[\norm{r_{n}\left(\theta\right)}_{2}^{2}\right]\Rightarrow\\
\sup_{\theta\in\Omega_{\theta}}\left|\frac{1}{N}\norm{r\left(\theta\right)}_{2}^{2}-\mbe\left[\norm{r_{n}\left(\theta\right)}_{2}^{2}\right]\right| & \ge\frac{\liph^{2}}{D^{4}}-\Lambda_{h}^{2}=\epsilon.
\end{align*}
However, again by \citet{keener:2011:theoretical} Theorem 9.2,
\begin{align*}
\sup_{\theta\in\Omega_{\theta}}\left|\frac{1}{N}\norm{r\left(\theta\right)}_{2}^{2}-\mbe\left[\norm{r_{n}\left(\theta\right)}_{2}^{2}\right]\right| & \xrightarrow[N\rightarrow\infty]{p}0,
\end{align*}
so $P\left(\textrm{Assumption \ref{assu:paper_lipschitz} holds}\right)\xrightarrow[N\rightarrow\infty]{}1$.
\end{proof}

\section{Genomics Experiments Details}\label{sec:appendix_genomics}

We demonstrate the Python and R code used to run and analyze the experiments on
the genomics data in a sequence of Jupyter notebooks. The output of these
notebooks are included below, though they are best viewed in their original
notebook form. The notebooks, as well as scripts and instructions for
reproducing our analysis in its entirety, can be found in the git repository
\href{https://github.com/rgiordan/AISTATS2019SwissArmyIJ}{rgiordan/AISTATS2019SwissArmyIJ}.



\section*{Supplementary material}

# fit_model_and_save

February 21, 2019

## 1 Genomics experiment details.

We demonstrate the infinitesimal jackknife on a publicly available data set of mice gene expression in Shoemaker et al. [2015].

Mice were infected with influenza virus, and gene expression was assessed several times after infection, so the observed data consists of expression levels $y_{gt}$ for genes $g = 1, ..., n_g$ and time points $t = 1, ..., n_t$, where in this case $n_g = 1000$ and $n_t = 42$.

This notebook contains the first of three steps in the analysis. In this notebook, we will first load the data and define a basis with a hyperparameter we wish to select with cross validation. We then describe the two stages of our analysis: a regression stage and a clustering stage. We then save the data for further analysis by the notebooks `load_and_refit` and `calculate_prediction_error`.

This notebook assumes you have already followed the instructions in `README.md` to install the necessary packages and create the dataset.

## 2 Step 1: Initial fit.

```
In [1]: import matplotlib.pyplot as plt
        %matplotlib inline

        import numpy as np
        import inspect
        import os
        import sys
        import time

        np.random.seed(3452453) # nothing special about this seed (we hope)!

In [2]: from aistats2019_ij_paper import regression_mixture_lib as rm_lib
        from aistats2019_ij_paper import regression_lib as reg_lib
        from aistats2019_ij_paper import sensitivity_lib as sens_lib
        from aistats2019_ij_paper import spline_bases_lib
        from aistats2019_ij_paper import transform_regression_lib as trans_reg_lib
        from aistats2019_ij_paper import loading_data_utils
        from aistats2019_ij_paper import saving_gmm_utils
        from aistats2019_ij_paper import mse_utils

        import plot_utils_lib
```

## 2.1 The first stage: Regression

### 2.1.1 Load data

```
In [3]: # Set bnp_data_repo to be the location of a clone of the repo
        # https://github.com/NelleV/genomic_time_series_bnp
        bnp_data_repo = '../../genomic_time_series_bnp'
        y_train, y_test, train_indx, timepoints = loading_data_utils.load_genomics_data(
            bnp_data_repo,
            split_test_train = True,
            train_indx_file = '../fits/train_indx.npy')
```

```
Loading data from:   ../../genomic_time_series_bnp/data/shoemaker2015reprocessed
```

```
In [4]: n_train = np.shape(y_train)[0]
        print('number of genes in training set: \n', n_train)

        n_test = np.shape(y_test)[0]
        print('number of genes in test set: \n', n_test)

        n_genes = n_train + n_test

        test_indx = np.setdiff1d(np.arange(n_genes), train_indx)
        gene_indx = np.concatenate((train_indx, test_indx))
```

```
number of genes in training set:
 700
number of genes in test set:
 300
```

Each gene $y_g$ has 42 observations. Observations are made at 14 timepoints, with 3 replicates at each timepoints.

```
In [5]: n_t = len(timepoints)
        n_t_unique = len(np.unique(timepoints))

        print('timepoints: \n ', timepoints, '\n')
        print('Distinct timepoints: \n', np.sort(np.unique(timepoints)), '\n')
        print('Number of distinct timepoints:', n_t_unique)
```

```
timepoints:
  [0, 0, 0, 3, 3, 3, 6, 6, 6, 9, 9, 9, 12, 12, 12, 18, 18, 18, 24, 24, 24, 30, 30, 30, 36, 36, 3

Distinct timepoints:
 [  0   3   6   9  12  18  24  30  36  48  60  72 120 168]

Number of distinct timepoints: 14
```

2

Here is the raw data for a few randomly chosen genes.

```
In [6]: f, axarr = plt.subplots(2, 3, figsize=(15,8))

        gene_indx = np.sort(np.random.choice(n_train, 6))

        for i in range(6):
            n = gene_indx[i]
            this_plot = axarr[int(np.floor(i / 3)), i % 3]
            this_plot.plot(timepoints, y_train[n, :].T, '+', color = 'blue');
            this_plot.set_ylabel('gene expression')
            this_plot.set_xlabel('time')
            this_plot.set_title('gene number {}'.format(n))

        f.tight_layout()
```

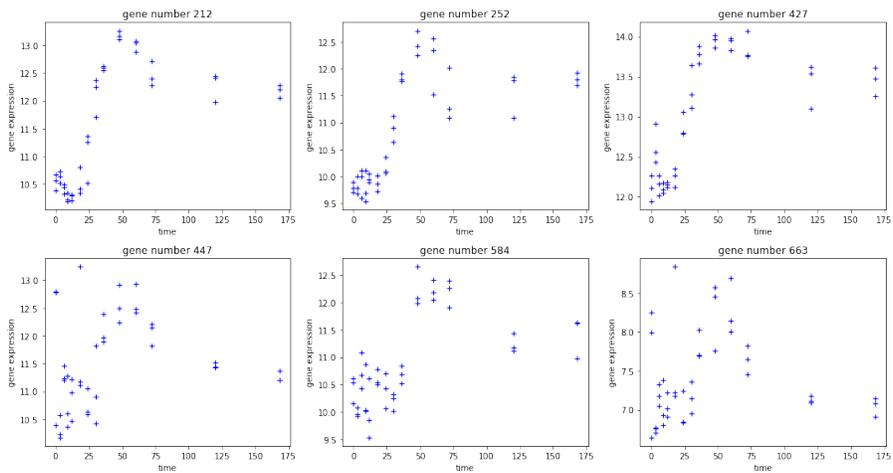

### 2.1.2 Define regressors

We model the time course using cubic B-splines. Let $\alpha$ be the degrees of freedom of the B-splines, and this is the parameter we seek to choose using cross-validation.

For a given degrees of freedom, the B-spline basis is given by an $n_t \times n_x$ matrix $X_{df}$, where the each column of $X_{df}$ is a B-spline basis vector evaluated at the $n_t$ timepoints. Note that $n_x$ increases with increasing degrees of freedom.

Note that we only use B-splines to smooth the first 11 timepoints. For the last three timepoints, $t = 72, 120, 168$, we use indicator functions on each timepoint as three extra basis vectors. In other words, we append to the regressor matrix three columns, where each column is 1 if $t = 72, 120$, or 168, respectively, and 0 otherwise. We do this to avoid numerical issues in the matrix $X^T X$. Because the later timepoints are more spread out, the B-spline basis are close to zero at the later timepoints, leading to matrices close to being singular.

3

```
In [7]: # Simulate passing arguments in on the command line so that the notebook
        # looks more like those in ``cluster_scripts``.
        class Args():
            def __init__(self):
                pass

        args = Args()
        args.df = 7
        args.degree = 3
        args.num_components = 10

In [8]: regressors = spline_bases_lib.get_genomics_spline_basis(
            timepoints, df=args.df, degree=3)

        regs = reg_lib.Regressions(y_train, regressors)
```

We plot the B-spline matrix for several degrees of freedom below:

```
In [9]: f, axarr = plt.subplots(2, 3, figsize=(15,8))

        i = 0
        for df in [4, 5, 6, 7, 8, 9]:
            _regressors = spline_bases_lib.get_genomics_spline_basis(
                    timepoints, exclude_num=3, df=df)

            this_plot = axarr[int(np.floor(i / 3)), i % 3]
            this_plot.plot(timepoints, _regressors);
            this_plot.set_xlabel('time')
            this_plot.set_ylabel('B-spline value')
            this_plot.set_title('B-spliine basis when df = {}'.format(df))

            i += 1

        f.tight_layout()
```

4

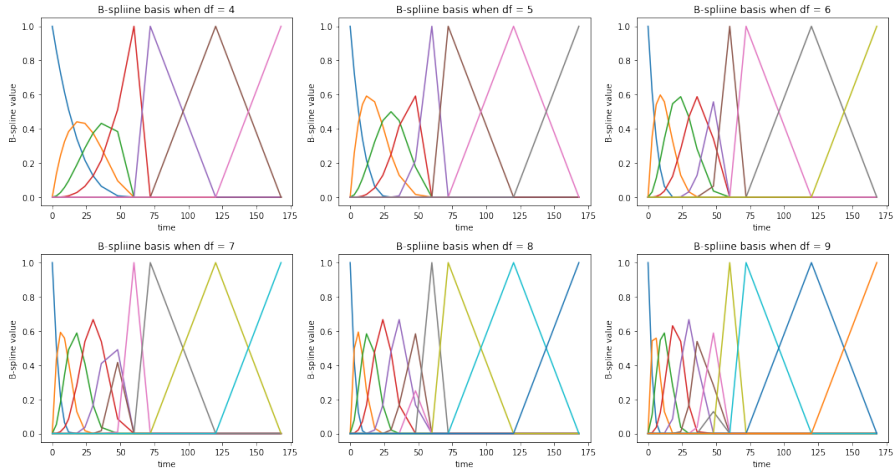

We display the regressor matrix below.

```
In [10]: plt.matshow(regs.x.T)
         plt.ylabel('basis')
         plt.xlabel('timepoint and replicate')
         plt.title('The (transposed) regressor matrix when df = {}\n'.format(args.df));
```

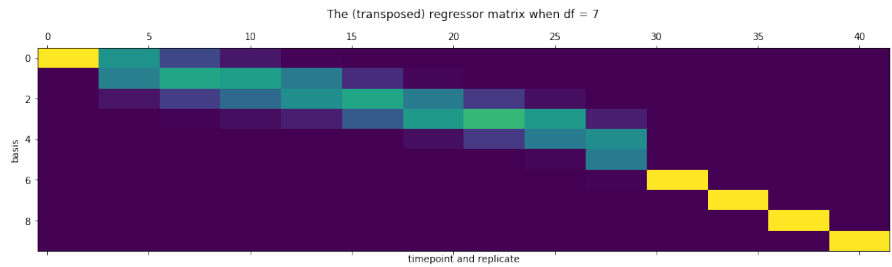

With the regressor $X$ defined above, for each gene $g$ we model $P\left(y_g|\beta_g, \sigma_g^2\right) = \mathcal{N}\left(y_g|X\beta_g, \sigma_g^2\right)$. In the second stage, we will want to cluster $\beta_g$ taking into account its uncertainty on each gene. To do this, we wish to estimate the posterior mean $\mathbb{E}[\beta_g|y_g]$ and covariance $\text{Cov}(\beta_g|y_g)$ with flat priors for both $\beta_g$ and $\sigma_g^2$.

For each gene, we estimate the posterior with a mean field variational Bayes (MFVB) approximation $q\left(\sigma_g^2, \beta_g; \hat{\eta}_g\right)$ to the posterior $P\left(\beta_g, \sigma_g^2|y_g\right)$.

In particular, we take $q\left(\sigma_g^2, \beta_g; \hat{\eta}_g\right) = q^*\left(\sigma_g^2\right) q^*\left(\beta_g\right)$, where $q^*\left(\sigma_g^2\right)$ is a dirac delta function, and we optimize over its a location parameter; $q^*\left(\beta_g\right)$ is a Gaussian density and we optimize over its mean and covariance.

5

The optimal variational approximation has a closed form that is formally identical to the standard frequentist mean and covariance estimate for linear regression. Explicitly, the optimal variational distribution is,

$$q^*(\beta_g) = \mathcal{N}\left(\beta_g \mid (X^T X)^{-1} X^T y_g, \ \hat{\tau}_g (X^T X)^{-1}\right)$$
$$q^*(\sigma_g^2) = \delta\{\sigma_g^2 = \hat{\tau}_g\}$$

where $\hat{\tau}_g = \frac{1}{n_t - n_x} \|y_g - X(X^T X)^{-1} X^T y_g\|_2^2$.

The advantage of the MVFB construction is that $\hat{\eta}_g$ for $g = 1, ..., n_g$ satisfies set of $n_g$ independent M-estimation objectives, allowing us to apply our infinitesimal jackknife results. Specifically, defining $\theta_{reg} := (\eta_1, ..., \eta_{n_g})$, we wish to minimize

$$F_{reg}(\theta_{reg}, \alpha) = \sum_{g=1}^{n_g} KL\left(q\left(\sigma_g^2, \beta_g; \eta_g\right) \| P\left(\beta_g, \sigma_g^2 | y_g\right)\right)$$
$$= -\sum_{g=1}^{n_g} \mathbb{E}_q\left[\log P\left(\beta_g, \sigma_g^2 | y_g\right)\right] + \mathbb{E}_q\left[\log q\left(\beta_g, \sigma_g^2 | \eta_g\right)\right]$$
$$:= \sum_{g=1}^{n_g} F_{reg,g}(\eta_g, \alpha).$$

Our M-estimator, then, is

$$\frac{\partial F_{reg}(\theta_{reg}, \alpha)}{\partial \theta_{reg}} = 0.$$

The class regs can calculate the optimal variational parameters for each gene. In particular, the variational parameters $\eta_g$ consist of a variational mean and covariance for $\beta_g$, as well as a location estimate for $\sigma_g^2$.

```
In [11]: reg_time = time.time()
         opt_reg_params = regs.get_optimal_regression_params()
         reg_time = time.time() - reg_time
         print('Regression time: {} seconds'.format(reg_time))
```

```
Regression time: 0.029132366180419922 seconds
```

Here are what some of the fits look like. Each regression produces a prediction $\hat{y}_g := X\mathbb{E}_q[\beta_g]$, plotted with the heavy red line above. The light red are predictions when $\beta_g$ is drawn from $q^*(\beta_g)$; the spread of the light red is intended to give a sense of the covariance of $\beta_g$.

```
In [12]: f, axarr = plt.subplots(2, 3, figsize=(15,8))

         for i in range(6):
             n = gene_indx[i]
             this_plot = axarr[int(np.floor(i / 3)), i % 3]
             plot_utils_lib.PlotRegressionLine(
                 timepoints, regs, opt_reg_params, n, this_plot=this_plot)

         f.tight_layout()
```

6

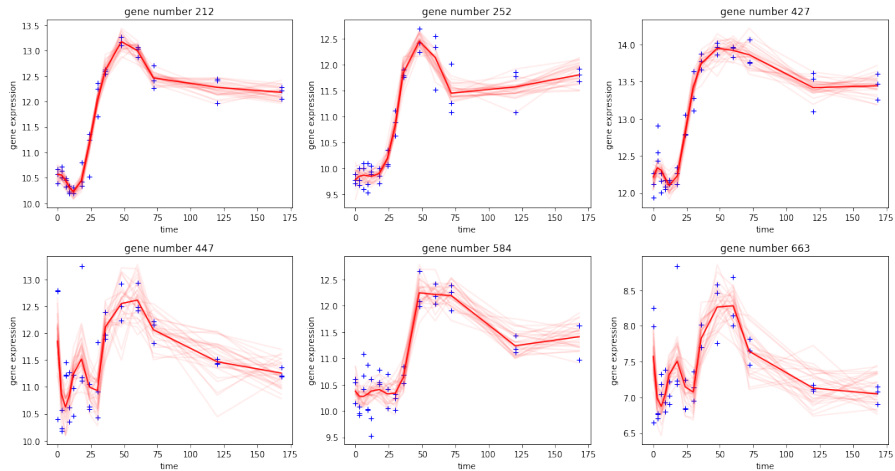

We also define and save data for the test regressions, which we will use later to evaluate out-of-sample performance. The training regressions will be saved below with the rest of the fit.

```
In [13]: regs_test = reg_lib.Regressions(y_test, regressors)
         test_regression_outfile = '../fits/test_regressions.json'
         with open(test_regression_outfile, 'w') as outfile:
             outfile.write(regs_test.to_json())
```

## 2.2 The second stage: fit a mixture model.

### 2.2.1 Transform the parameters before clustering

We are interested in the pattern of gene expression, not the absolute level, so we wish to cluster $\hat{y}_g - \bar{\hat{y}}_g$, where $\bar{\hat{y}}_g$ is the average over time points. Noting that the $n_t \times n_t$ matrix $\text{Cov}_q\left(\hat{y}_g - \bar{\hat{y}}_g\right)$ is rank-deficient because we have subtracted the mean, the final step is to rotate $\hat{y}_g - \bar{\hat{y}}_g$ into a basis where the zero eigenvector is a principle axis and then drop that component.

Call these transformed regression coefficients $\gamma_g$ and observe that $\text{Cov}_q\left(\gamma_g\right)$ has a closed form and is full-rank. It is these $\gamma_g$s that we will cluster in the second stage.

We briefly note that the re-centering operation could have been equivalently achieved by making a constant one of the regressors. We chose this implementation because it also allows the user to cluster more complex, non-linear transformations of the regression coefficients, though we leave this extension for future work.

We note that the transformations described in this section are done automatically in the GMM class. We are only calculating these transformations here for exposition.

```
In [14]: # Get the matrix that does the transformation.
         transform_mat, unrotate_transform_mat = \
             trans_reg_lib.get_reversible_predict_and_demean_matrix(regs.x)
         trans_obs_dim = transform_mat.shape[0]
```

If $T$ is the matrix that effects the transformation, then

$$\mathbb{E}_q[\gamma_g] = T\mathbb{E}_q[\beta_g]$$
$$\text{Cov}_q(\gamma_q) = T\text{Cov}_q(\beta_g)T^T$$

The transformed parameters are also regression parameters, just in a different space.

```
In [15]: # Apply the transformation
         transformed_reg_params = \
             trans_reg_lib.multiply_regression_by_matrix(
                 opt_reg_params, transform_mat)
```

We now visualize the transformed coefficients and their uncertainty.

```
In [16]: f, axarr = plt.subplots(2, 3, figsize=(15,8))

         transformed_beta = transformed_reg_params['beta_mean']
         transformed_beta_info = transformed_reg_params['beta_info']

         for i in range(6):
             n = gene_indx[i]
             this_plot = axarr[int(np.floor(i / 3)), i % 3]
             this_plot.plot(transformed_beta[n, :], color = 'red');
             this_plot.set_ylabel('transformed coefficient')
             this_plot.set_xlabel('index')
             this_plot.set_title('gene number {}'.format(n))

             # draw from the variational distribution, to plot uncertainties
             for j in range(30):
                 transformed_beta_draw = np.random.multivariate_normal(
                     transformed_beta[n, :], \
                      np.linalg.inv(transformed_beta_info[n]))

                 axarr[int(np.floor(i / 3)), i % 3].plot(transformed_beta_draw,
                                                          color = 'red', alpha = 0.08);

         f.tight_layout()
```

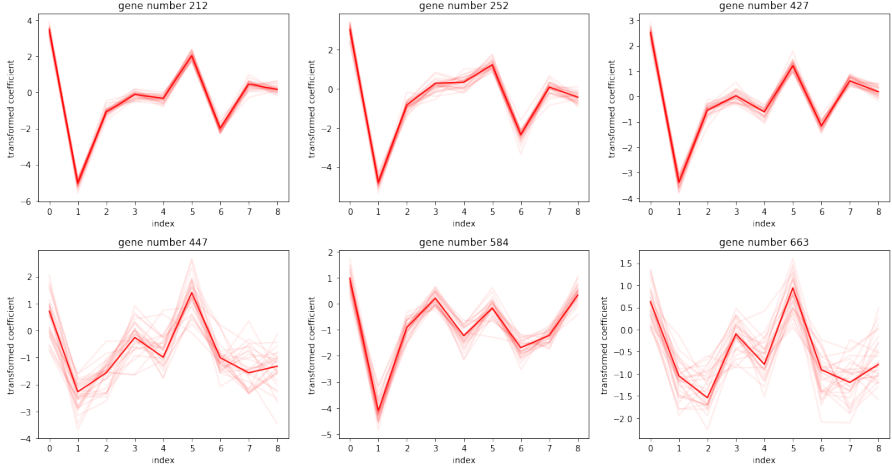

The heavy red lines are the means of the transformed regression coefficients; shaded lines are draws from the variational distribution.

It is these transformed coefficients, $\gamma_g$, that we cluster in the second stage.

### 2.2.2 Estimate an optimal clustering.

We now define a clustering problem for the $\gamma_g$. Let $n_k$ be the number of clusters, and $\mu_1, ..., \mu_{n_k}$ be the cluster centers. Also let $z_{gk}$ be the binary indicator for the $g$th gene belonging to cluster $k$. We then define the following generative model

$$P(\pi) = Dirichlet(\omega)$$
$$P(\mu_k) = \mathcal{N}(\mu_k|0, \Sigma_0) \quad \text{for} \quad k = 1, ..., n_k$$
$$P(z_{gk} = 1|\pi_k) = \pi_k \quad \text{for} \quad k = 1, ..., n_k; \ n = 1, ..., n_g$$
$$P(\gamma_g|z_{gk} = 1, \mu_k, \eta_g) = \mathcal{N}(\gamma_g|\mu_k, \text{Cov}_q(\gamma_g) + \epsilon I_{n_t-1}) \quad \text{for} \quad k = 1, ..., n_k; \ n = 1, ..., n_g.$$

where $\epsilon$ is a small regularization parameter, which helped our optimization produce more stable results.

We will estimate the clustering using the maximum a posteriori (MAP) estimator of $\theta_{clust} := (\mu, \pi)$. This defines an optimization objective that we seek to minimize:

$$F_{clust}(\theta_{clust}, \theta_{reg}) = -\sum_{g=1}^{n_g} E_{q_z^*}\left\{\log P(\gamma_g|\eta_g, \mu, \pi, z_g) - \log P(z_g|\pi)\right\} - \log P(\mu) - \log P(\pi)$$

which, for every value of $\theta_{reg}$, we expect to satisfy

$$\frac{\partial F_{clust}(\theta_{clust}, \theta_{reg})}{\partial \theta_{clust}} = 0.$$

9

Note that $\theta_{clust}$ involves only the "global" parameters $\mu$ and $\pi$. We did take a variational distribution for the $z_{gk}$s, represented by independent Bernoulli distribution, but the optimal $q_z^*$ can be written as a function of $\mu$ and $\pi$. Hence, our optimization objective only involves these global parameters.

```
In [17]: # Define prior parameters.
         num_components = args.num_components
         epsilon = 0.1
         loc_prior_info_scalar = 1e-5

         trans_obs_dim = regs.x.shape[1] - 1
         prior_params = \
             rm_lib.get_base_prior_params(trans_obs_dim, num_components)
         prior_params['probs_alpha'][:] = 1
         prior_params['centroid_prior_info'] = loc_prior_info_scalar * np.eye(trans_obs_dim)

In [18]: gmm = rm_lib.GMM(args.num_components,
                          prior_params, regs, opt_reg_params,
                          inflate_coef_cov=None,
                          cov_regularization=epsilon)
```

In our experiment, the number of clusters $n_k$ was chosen to be 10. We set $\omega$ to be the ones vector of length $n_k$. The prior info for the cluster centers $\Sigma_0$ is 1e-05$\times I$. $\epsilon$ was set to be 0.1.

Let us examine the optimization objective. First, we'll inspect the likelihood terms. What follows is the likelihood given that gene $g$ belongs to cluster $k$.

```
In [19]: print(inspect.getsource(rm_lib.get_log_lik_nk))

def get_log_lik_nk(centroids, probs, x, x_infos):
    loc_log_lik = \
        -0.5 * (-2 * np.einsum('ni,kj,nij->nk', x, centroids, x_infos) +
                np.einsum('ki,kj,nij->nk', centroids, centroids, x_infos))

    log_probs = np.log(probs[0, :])
    log_lik_by_nk = loc_log_lik + log_probs.T

    return log_lik_by_nk
```

We can then optimize for $q_z^*$, which can be parametrized by its mean $\mathbb{E}_{q_z^*}[z]$. We note that this update has a closed form given $\theta_{clust}$, so there is no need to solve an optimization problem to find $q_z^*(z)$. We additionally note that we do not use the EM algorithm, which we found to have exhibit extremely poor convergence rates. Rather, we set $q_z^*(z)$ to its optimal value given $\theta_{clust}$ and return the objective as a function of $\theta_{clust}$ alone, allowing the use of more general and higher-quality optimization routines.

```
In [20]: print(inspect.getsource(rm_lib.get_e_z))
```

10

```python
def get_e_z(log_lik_by_nk):
    log_const = paragami.simplex_patterns.logsumexp(log_lik_by_nk, axis=1)
    e_z = np.exp(log_lik_by_nk - log_const)
    return e_z
```

With the optimal parameters for $z_{nk}$, we combine the likelihood term with the prior and entropy terms.

```
In [21]: print(inspect.getsource(rm_lib.wrap_get_loglik_terms))
         print(inspect.getsource(rm_lib.wrap_get_kl))
```
```python
def wrap_get_loglik_terms(gmm_params, transformed_reg_params):
    log_lik_by_nk = get_log_lik_nk(
        centroids=gmm_params['centroids'],
        probs=gmm_params['probs'],
        x=transformed_reg_params['beta_mean'],
        x_infos=transformed_reg_params['beta_info'])

    e_z = get_e_z(log_lik_by_nk)

    return log_lik_by_nk, e_z

def wrap_get_kl(gmm_params, transformed_reg_params, prior_params):
    log_lik_by_nk, e_z = \
        wrap_get_loglik_terms(gmm_params, transformed_reg_params)
    log_prior = get_log_prior(
        gmm_params['centroids'], gmm_params['probs'], prior_params)
    return get_kl(log_lik_by_nk, e_z, log_prior)
```

This objective function is wrapped in the `GMM` class method `get_params_kl`.

```
In [22]: print(inspect.getsource(gmm.get_params_kl))
```
```python
    def get_params_kl(self, gmm_params):
        """Get the optimization objective as a function of the mixture
        parameters.
        """
        return wrap_get_kl(
            gmm_params, self.transformed_reg_params, self.prior_params)
```

### 2.2.3 Optimization

For optimization we make extensive use of the autograd library for automatic differentiation and the paragami library for parameter packing and sparse Hessians. These packages' details are beyond the scope of the current notebook.

First, we do a k-means initialization.

```
In [23]: print('Running k-means init.')
         init_gmm_params = \
             rm_lib.kmeans_init(gmm.transformed_reg_params,
                                gmm.num_components, 50)
         print('Done.')
         init_x = gmm.gmm_params_pattern.flatten(init_gmm_params, free=True)

Running k-means init.
Done.
```

We note that the match between "exact" cross-validation (removing time points and re-optimizing) and the IJ was considerably improved by using a high-quality second-order optimization method. In particular, for these experiments, we employed the Newton conjugate-gradient trust region method (Chapter 7.1 of Wright et al [1999]) as implemented by the method `trust-ncg` in `scipy.optimize`, preconditioned by the Cholesky decomposition of an inverse Hessian calculated at an initial approximate optimum.

We found that first-order or quasi-Newton methods (such as BFGS) often got stuck or terminated at points with fairly large gradients. At such points our method does not apply in theory nor, we found, very well in practice.

The inverse Hessian used for the preconditioner was with respect to the clustering parameters only and so could be calculated quickly, in contrast to the $H_1$ matrix used for the IJ, which includes the regression parameters as well.

First, run with a low tolerance to get a point at which to evaluate an initial preconditioner.

```
In [24]: gmm.conditioned_obj.reset() # Reset the logging and iteration count.
         gmm.conditioned_obj.set_print_every(1)

         opt_time = time.time()
         gmm_opt, init_x2 = gmm.optimize(init_x, gtol=1e-2)
         opt_time = time.time() - opt_time

Iter 0: f = -159.11834165
Iter 1: f = -159.67926278
Iter 2: f = -159.97782885
Iter 3: f = -160.15878320
Iter 4: f = -159.59447036
Iter 5: f = -160.19209687
Iter 6: f = -160.27259154
Iter 7: f = -160.29486553
Iter 8: f = -160.33460656
Iter 9: f = -160.34154288
Iter 10: f = -160.32382096
Iter 11: f = -160.34447865
Iter 12: f = -160.34634639
Iter 13: f = -160.34692896
```

Next, set the preconditioner using the square root inverse Hessian at the point `init_x2`.

```
In [25]: tic = time.time()
         h_cond = gmm.update_preconditioner(init_x2)
         opt_time += time.time() - tic
```

The method `optimize_fully` repeats this process of optimizing and re-calculating the preconditioner until the optimal point does not change.

```
In [26]: gmm.conditioned_obj.reset()
         tic = time.time()
         gmm_opt, gmm_opt_x = gmm.optimize_fully(
             init_x2, verbose=True)
         opt_time += time.time() - tic
         print('Optimization time: {} seconds'.format(opt_time))

Preconditioned iteration 1
  Running preconditioned optimization.
Iter 0: f = -160.34692896
Iter 1: f = -160.34694250
Iter 2: f = -160.34694250
Preconditioned iteration 2
  Getting Hessian and preconditioner.
  Running preconditioned optimization.
Iter 3: f = -160.34694250
Iter 4: f = -160.34694250
Converged.
Optimization time: 8.438910484313965 seconds
```

`paragami` patterns allow conversion between unconstrained vectors and dictionaries of parameter values. After "folding" the optimal `gmm_opt_x`, `opt_gmm_params` contains a dictionary of optimal cluster centroids and cluster probabilities.

```
In [27]: opt_gmm_params = gmm.gmm_params_pattern.fold(gmm_opt_x, free=True)
         print(opt_gmm_params.keys())
         print(np.sort(opt_gmm_params['probs']))

odict_keys(['centroids', 'probs'])
[[0.01567608 0.04016882 0.06955236 0.07427946 0.09373695 0.0947442
  0.09653288 0.12626624 0.15739176 0.23165127]]
```

Each gene's regression line has an inferred cluster membership given by $\mathbb{E}_{q_z^*}[z_g]$, and an expected posterior centroid given by $\sum_k \mathbb{E}_{q_z^*}[z_{gk}]\mu_k$. This expected posterior centroid can be untransformed to give a prediction for the observation.

It is the difference between this prediction line — which is a function of the clustering — and the actual data that we consider to be the "error" of the model.

```
In [28]: gmm_pred = mse_utils.get_predictions(gmm, opt_gmm_params, opt_reg_params)

         f, axarr = plt.subplots(2, 3, figsize=(15,8))

         for i in range(6):
             n = gene_indx[i]
             this_plot = axarr[int(np.floor(i / 3)), i % 3]
             plot_utils_lib.PlotRegressionLine(
                 timepoints, regs, opt_reg_params, n, this_plot=this_plot)
             plot_utils_lib.PlotPredictionLine(
                 timepoints, regs, gmm_pred, n, this_plot=this_plot)

         f.tight_layout()
```

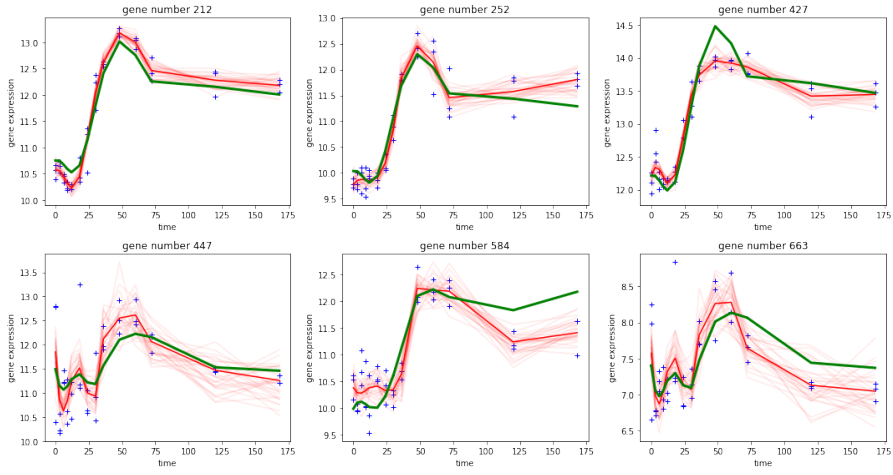

### 2.2.4 Calculating $H_1$ for the IJ

We seek to choose the degrees of freedom $\alpha$ for the B-splines using cross-validation. We leave out one or more timepoints, and fit using only the remaining timepoints. We then estimate the test error by predicting the value of the genes at the held out timepoints.

To do this, we define time weights $w_t$ by observing that, for each $g$, the term $\mathbb{E}_q \left[ \log P \left( \beta_g, \sigma_g^2 | y_g \right) \right]$ decomposes into a sum over time points:

$$F_{reg,g}\left(\eta_g, \alpha, w\right) := -\sum_{t=1}^{n_t} w_t \left( -\frac{1}{2}\sigma_g^{-2} \left(y_{g,t} - (X\beta_g)_t\right)^2 - \frac{1}{2}\log \sigma_g^2 \right) + \mathbb{E}_q\left[\log q\left(\beta_g, \sigma_g^2 | \eta_g\right)\right].$$

We naturally define $F_{reg}\left(\theta_{reg}, \alpha, w\right) := \sum_{g=1}^{n_g} F_{reg,g}\left(\eta_g, \alpha, w\right)$.

14

By defining $\theta = (\theta_{clust}, \theta_{reg})$, we then have an M-estimator

$$G(\theta, w, \alpha) := \begin{pmatrix} \frac{\partial F_{reg}(\theta_{reg}, w, \alpha)}{\partial \theta_{reg}} \\ \frac{\partial F_{clust}(\theta_{clust}, \theta_{reg})}{\partial \theta_{clust}} \end{pmatrix} = 0.$$

We can then apply the IJ to approximate the leaving out of various timepoints.

Note that what we call the "Hessian" for this two-step procedure is not really a Hessian, as it is not symmetric. More precisely, it is the Jacobian of $G$, or what we defined as $H_1$ in the text.

Calculating $H_1$ is the most time-consuming part of the infinitesimal jackknife, since the $H_1$ matrix is quite large (though sparse). However, once $H_1$ is computed, calculating each $\theta_{IJ}(w)$ is extremely fast.

$H_1$ can be computed in blocks:

$$H_1 = \begin{pmatrix} \nabla^2_{\theta_{reg}} F_{reg} & 0 \\ \nabla_{\theta_{reg}} \nabla_{\theta_{clust}} F_{clust} & \nabla^2_{\theta_{clust}} F_{clust} \end{pmatrix}$$

The code refers to $\nabla^2_{\theta_{clust}} F_{clust}$ as the "GMM Hessian". It refers to $\nabla_{\theta_{reg}} \nabla_{\theta_{clust}} F_{clust}$ as the "cross Hessian". And it refers to $\nabla^2_{\theta_{reg}} F_{reg}$ as the "regression Hessian", which itself is block diagonal, with each block an observation. Due to details of the implementation of block sparse Hessians using forward mode automatic differnetiation in the class `vittles.SparseBlockHessian`, the code below confusingly refers to each regression parameter as a "block".

When the `FitDerivatives` class is initialized, it calculates these blocks separately and stacks them into the attribute `full_hess`, which is a sparse matrix representing $H_1$.

```
In [29]: # Even though $H_1$ is not a Hessian, by force of habit we call the time to
         # compute it ``hess_time``.
         hess_time = time.time()
         fit_derivs = sens_lib.FitDerivatives(
             opt_gmm_params, opt_reg_params,
             gmm.gmm_params_pattern, regs.reg_params_pattern,
             gmm=gmm, regs=regs,
             print_every=10)
         hess_time = time.time() - hess_time
         print('Total hessian time: {} seconds'.format(hess_time))
```

```
Initializing FitDerivatives.
Getting t Jacobian.
Getting full Hessian.
   Getting GMM Hessian...
   GMM Hessian time: 2.1917014122009277
   Getting cross Hessian...
   Cross Hessian time: 34.25235605239868
   Getting regression Hessian...
Block index 0 of 66.
Block index 10 of 66.
Block index 20 of 66.
Block index 30 of 66.
```

15

```
Block index 40 of 66.
Block index 50 of 66.
Block index 60 of 66.
Done differentiating.
   Regression Hessian time: 121.74362897872925
Done with full Hessian.
Total hessian time: 169.59288716316223 seconds
```

### 2.2.5 Save results as a compressed file.

The results, including $H_1$, are now saved. To calculate the exact CV, these results (including the preconditioner) will be loaded and the model will be refit with timepoints left out. To calculate the IJ, the same results will be loaded and $H_1$ will be used to calculate the IJ.

```
In [30]: extra_metadata = dict()
         extra_metadata['opt_time'] = opt_time
         extra_metadata['reg_time'] = reg_time
         extra_metadata['hess_time'] = hess_time
         extra_metadata['df'] = args.df
         extra_metadata['degree'] = args.degree

         npz_outfile = '../fits/initial_fit.npz'
         saving_gmm_utils.save_initial_optimum(
             npz_outfile,
             gmm=gmm,
             regs=regs,
             timepoints=timepoints,
             fit_derivs=fit_derivs,
             extra_metadata=extra_metadata)
```

# load_and_refit

February 21, 2019

## 1 Step 2: Refit.

In this notebook, we calculate the parameters used for exact CV by refitting the model initially fit in step one, the notebook `fit_model_and_save`.

For expository purposes this notebook calculates the refit for only one weight vector. To compute exact CV, one would perform the corresponding computation for all leave-k-out weight vectors.

```python
In [1]: from copy import deepcopy
        import inspect
        import matplotlib.pyplot as plt
        %matplotlib inline
        import numpy as np
        import sys
        import time

        np.random.seed(3452453)

        import paragami

        from aistats2019_ij_paper import regression_mixture_lib as rm_lib
        from aistats2019_ij_paper import saving_gmm_utils
        from aistats2019_ij_paper import mse_utils

        import plot_utils_lib

In [2]: # Load the initial fit.
        # This file was produced by the notebook ``fit_model_and_save``.
        initial_fit_infile = '../fits/initial_fit.npz'
        full_fit, gmm, regs, metadata = \
            saving_gmm_utils.load_initial_optimum(initial_fit_infile)
        timepoints = metadata['timepoints']

Initializing FitDerivatives.
Using provided t_jac.
Using provided full_hess.
```

First, choose some timepoints to leave out.

```
In [3]: # Simulate passing arguments in on the command line.
        class Args():
            def __init__(self):
                pass

        args = Args()
        args.num_times = 1
        args.which_comb = 1
        args.max_num_timepoints = 7
```

The number of points left out (that is, *k*) is given by `num_times`, which is 1. The largest timepoint we leave out is given by `max_num_timepoints`, which is 7. Because later timepoints are not affected by the smoothing, there is no reason to leave them out.

There are a certain number of ways to leave *k* out of 7 timepoints, and `which_comb` chooses one of them in the order given by the function `itertools.combinations`. Of course, when $k = 1$, `which_comb` simply chooses which timepoint to leave out. `mse_utils.get_indexed_combination` maps `which_comb` to particular timepoints in a consistent way.

Full exact CV would run this script for all 7 choose *k* values of `which_comb`.

Because we have repeated measurements at each timepoint, leaving out a single timepoint will correspond to leaving out multiple row of the observation matrix. Those rows are determined by `mse_utils.get_time_weight`, which also returns a weight vector setting these observations' weights to zero.

```
In [4]: lo_inds = mse_utils.get_indexed_combination(
            num_times=args.num_times, which_comb=args.which_comb,
            max_num_timepoints=args.max_num_timepoints)
        new_time_w, full_lo_inds = mse_utils.get_time_weight(lo_inds, timepoints)

        print('Left out timepoint: {}'.format(lo_inds))
        print('Left out observations: {}'.format(full_lo_inds))
        print('Leave-k-out weights: {}'.format(new_time_w))

Left out timepoint: [1]
Left out observations: [3 4 5]
Leave-k-out weights: [1 1 1 0 0 0 1 1 1 1 1 1 1 1 1 1 1 1 1 1 1 1 1 1 1 1 1 1 1 1 1 1 1 1 1
 1 1 1 1 1]
```

We now re-optimize with the new weights.

Note that we could either start the optimization at the initial optimum (a "warm start") or do a fresh start from k-means. A fresh start is more time consuming but a more stringent test for the accuracy of the IJ. We calculate both, but report results from the fresh start in the paper. In the notebook `examine_and_save_results`, you can choose to examine either set of results.

Here, for consistency with the paper, we re-initialize with k-means.

```
In [5]: regs.time_w = deepcopy(new_time_w)
        reg_params_w = regs.get_optimal_regression_params()
```

2

```python
            gmm.set_regression_params(reg_params_w)

            init_gmm_params = \
                rm_lib.kmeans_init(gmm.transformed_reg_params,
                                   gmm.num_components, 50)
            init_x = gmm.gmm_params_pattern.flatten(init_gmm_params, free=True)

            opt_time = time.time()
            gmm_opt, init_x2 = gmm.optimize(init_x, gtol=1e-2)

            print('\tUpdating preconditioner...')
            kl_hess = gmm.update_preconditioner(init_x2)

            print('\tRunning preconditioned optimization...')
            gmm.conditioned_obj.reset()
            reopt, gmm_params_free_w = gmm.optimize_fully(init_x2, verbose=True)
            print(gmm_opt.message)
            opt_time = time.time() - opt_time

            print('Refit time: {} seconds'.format(opt_time))
```

```
Iter 0: f = -153.38003431
Iter 1: f = -152.49438715
Iter 2: f = -153.69147895
Iter 3: f = -153.83779915
Iter 4: f = -154.02397812
Iter 5: f = -153.41393391
Iter 6: f = -154.10396420
Iter 7: f = -154.14366282
Iter 8: f = -154.14261201
Iter 9: f = -154.16417745
Iter 10: f = -154.18307547
Iter 11: f = -154.20711481
Iter 12: f = -154.22118064
Iter 13: f = -154.27402715
Iter 14: f = -154.28739474
Iter 15: f = -154.33849929
Iter 16: f = -154.03580241
Iter 17: f = -154.35421130
Iter 18: f = -154.36910489
Iter 19: f = -154.36872458
Iter 20: f = -154.37238982
Iter 21: f = -154.37722095
Iter 22: f = -154.38186985
Iter 23: f = -154.38410992
        Updating preconditioner...
        Running preconditioned optimization...
Preconditioned iteration 1
```

3

```
   Running preconditioned optimization.
Iter 0: f = -154.38410992
Iter 1: f = -154.38423176
Iter 2: f = -154.38584092
Iter 3: f = -154.21889674
Iter 4: f = -154.42200228
Iter 5: f = -154.39603234
Iter 6: f = -154.39957947
Iter 7: f = -154.41374585
Iter 8: f = -154.43397491
Iter 9: f = -154.43484046
Iter 10: f = -154.43484816
Iter 11: f = -154.43484816
Preconditioned iteration 2
   Getting Hessian and preconditioner.
   Running preconditioned optimization.
Iter 12: f = -154.43484816
Iter 13: f = -154.43484816
Converged.
Optimization terminated successfully.
Refit time: 14.35115647315979 seconds
```

We now save the results.

```
In [6]: gmm_params_w = \
            full_fit.comb_params_pattern['mix'].fold(
                gmm_params_free_w, free=True)
        refit_comb_params = {
            'mix': gmm_params_w,
            'reg': reg_params_w }
        refit_comb_params_free = \
            full_fit.comb_params_pattern.flatten(refit_comb_params, free=True)

In [7]: save_filename = \
            '../fits/refit__num_times{}__which_comb{}.npz'.format(
                args.num_times, args.which_comb)
        print('Saving to {}'.format(save_filename))
        saving_gmm_utils.save_refit(
            outfile=save_filename,
            comb_params_free=refit_comb_params_free,
            comb_params_pattern=full_fit.comb_params_pattern,
            initial_fit_infile=initial_fit_infile,
            time_w=new_time_w,
            lo_inds=lo_inds,
            full_lo_inds=full_lo_inds)

Saving to ../fits/refit__num_times1__which_comb1.npz
```

# calculate_prediction_errors

February 21, 2019

## 1 Step 3: Calculate the IJ and prediction errors.

In this notebook, for a single weight vector, we calculate the IJ itself as well as the prediction errors for exact CV and IJ. This notebook uses the output of the notebooks `load_and_refit` and `fit_model_and_save`.

```
In [1]: import numpy as np
        import paragami
        import vittles
        import scipy as sp
        from scipy import sparse
        import time

        import seaborn as sns
        import pandas as pd

        import matplotlib.pyplot as plt
        %matplotlib inline

        np.random.seed(3452453)

        from aistats2019_ij_paper import regression_lib as reg_lib
        from aistats2019_ij_paper import sensitivity_lib as sens_lib
        from aistats2019_ij_paper import saving_gmm_utils
        from aistats2019_ij_paper import mse_utils

        import plot_utils_lib

In [2]: # Simulate passing arguments in on the command line.
        class Args():
            def __init__(self):
                pass

        args = Args()
        args.num_times = 1
        args.which_comb = 1
        args.max_num_timepoints = 7
```

1

```
In [3]: ###############################
        # Load the original fit.

        print('Loading original fit.')
        initial_fit_infile = '../fits/initial_fit.npz'
        full_fit, gmm, regs, initial_metadata = \
            saving_gmm_utils.load_initial_optimum(initial_fit_infile)

        opt_comb_params = full_fit.get_comb_params()
Loading original fit.
Initializing FitDerivatives.
Using provided t_jac.
Using provided full_hess.

In [4]: ###############################
        # Load the test data

        test_regression_infile = '../fits/test_regressions.json'
        with open(test_regression_infile) as infile:
            regs_test = reg_lib.Regressions.from_json(infile.read())

        ##########################################
        # Load a refit as specfified by ``args``.

        refit_filename = \
            '../fits/refit__num_times{}__which_comb{}.npz'.format(
                args.num_times, args.which_comb)

        comb_params_free_refit, comb_params_pattern_refit, refit_metadata = \
            saving_gmm_utils.load_refit(refit_filename)

        time_w = refit_metadata['time_w']
        lo_inds = refit_metadata['lo_inds']
        full_lo_inds = refit_metadata['full_lo_inds']

        assert(comb_params_pattern_refit == full_fit.comb_params_pattern)
        comb_params_refit = comb_params_pattern_refit.fold(
            comb_params_free_refit, free=True)

        time_w = refit_metadata['time_w']
        lo_inds = refit_metadata['lo_inds']
        full_lo_inds = refit_metadata['full_lo_inds']
```

The objects named `comb_params` refer to both the regression and clustering parameters. The name `free` refers to the unconstrained flat value for the parameters as calculated by `paragami`.

```
In [5]: print('Regression pattern: ',
              comb_params_pattern_refit['reg'])
```

2

```
        print('Clustering pattern: ',
              comb_params_pattern_refit['mix'])

Regression pattern:  OrderedDict:
        [beta_mean] = NumericArrayPattern (700, 10) (lb=-inf, ub=inf)
        [beta_info] = PatternArray (700,) of PDMatrix 10x10 (diag_lb = 0.0)
        [y_info] = NumericArrayPattern (700,) (lb=0.0, ub=inf)
Clustering pattern:  OrderedDict:
        [centroids] = NumericArrayPattern (10, 9) (lb=-inf, ub=inf)
        [probs] = SimplexArrayPattern (1,) of 10-d simplices
```

#### 1.0.1 Calculate the infinitesimal jackknife.

The vittles package makes it easy to calculate linear approximations to the sensitivity of M-estimators to hyperparameters, of which the IJ is a special case. Here, the HyperparameterSensitivityLinearApproximation uses the sparse value of $H_1$ calculated earlier.

Note that $H_1$ is factorized during the initialization of weight_sens, and that it takes relatively little time.

```
In [6]: # Note that if you don't cast the jacobian to a numpy array from
        # a numpy matrix, the output is a 2d-array, causing confusion later.
        weight_sens = vittles.HyperparameterSensitivityLinearApproximation(
            objective_fun=lambda: 0,
            opt_par_value=full_fit.comb_params_free,
            hyper_par_value=regs.time_w,
            hessian_at_opt=sp.sparse.csc_matrix(full_fit.full_hess),
            cross_hess_at_opt=np.array(full_fit.t_jac.todense()))
```

We now use the weight_sens object to approximate the "free" value of the combined parameters at time_w. The IJ operates in unconstrained space, so we use paragami to fold the unconstrained vector back into a dictionary of parameters.

```
In [7]: # Get the infinitesimal jackknife for the refit weight vector.
        lr_time = time.time()
        comb_params_free_lin = \
            weight_sens.predict_opt_par_from_hyper_par(time_w)
        lr_time = time.time() - lr_time
        print('Infinitesimal jackknife time: {}'.format(lr_time))

        comb_params_lin = full_fit.comb_params_pattern.fold(comb_params_free_lin, free=True)
```

```
Infinitesimal jackknife time: 0.0011603832244873047
```

#### 1.0.2 Calculate various prediction errors.

Recall that the prediction error is the difference between the data and the posterior expected cluster centroid for a particular gene. Let us consider the original optimal clustering parameters,

3

`opt_comb_params['mix']`. To get the test set error on gene $g$ for these parameters, we need to do the following steps:

1. Run the regression for gene $g$ in the test set
2. Classify the regression, calculating $\mathbb{E}_{q_z^*}[z_g]$. This is a function of the clustering parameters and the regression line for gene $g$.
3. Calculate the expected posterior cluster centroid for gene $g$, which is $\mu_g^* = \sum_k \mathbb{E}_{q_z^*}[z_{gk}] \mu_k$.
4. Because the transformation discards the mean information, compare the de-meaned data to the estimated centroid: $error_{gt} = \left(y_{gt} - \frac{1}{T}\sum_{t'=1}^{T} y_{gt'}\right) - \mu_{gt}^*$.

Note that step one could re-run the regression either with the original weights or the new weights. We found that this decision does not matter qualitatively. Here and in the paper, we simply classify the original regression, but the notebook `examine_and_save_results` can produce results for oth the original and re-weighted regressions.

We will examine prediction error on the time points that are left out, that is, for observations in `full_lo_inds`.

```
In [8]: print('Calculating prediction error.')

        # Get the training set error on the full data.
        train_error = mse_utils.get_lo_err_folded(
            opt_comb_params,
            keep_inds=full_lo_inds,
            mse_regs=regs,
            mse_reg_params=opt_comb_params['reg'],
            gmm=gmm)

        ############
        # Original fit.

        # Get the optimal test set regressions.
        reg_params_test = regs_test.get_optimal_regression_params()

        # Get the test error for the original fit.
        orig_test_error = mse_utils.get_lo_err_folded(
            opt_comb_params,
            keep_inds=full_lo_inds,
            mse_regs=regs_test,
            mse_reg_params=reg_params_test,
            gmm=gmm)

        orig_pred = mse_utils.get_predictions(
            gmm, opt_comb_params['mix'], reg_params_test)

        # Get the test error for the CV refit.
        cv_error = mse_utils.get_lo_err_folded(
            comb_params_refit,
            keep_inds=full_lo_inds,
```

4

```
            mse_regs=regs_test,
            mse_reg_params=reg_params_test,
            gmm=gmm)

        cv_pred = mse_utils.get_predictions(
            gmm, comb_params_refit['mix'], reg_params_test)

        # Get the test error for the IJ approximation.
        ij_error = mse_utils.get_lo_err_folded(
            comb_params_lin,
            keep_inds=full_lo_inds,
            mse_regs=regs_test,
            mse_reg_params=reg_params_test,
            gmm=gmm)

        ij_pred = mse_utils.get_predictions(
            gmm, comb_params_lin['mix'], reg_params_test)

Calculating prediction error.
```

### 1.0.3 Selected results.

We now make a cursory comparison of the results. For a more detailed analysis, including the results that went into the paper, see the notebook examine_and_save_results.

```
In [9]: cv_excess_error = cv_error - orig_test_error
        ij_excess_error = ij_error - orig_test_error

        def GetColDf(col):
            return pd.DataFrame(
                {'cv_error': cv_error[:, col],
                 'cv_excess': cv_excess_error[:, col],
                 'ij_error': ij_error[:, col],
                 'ij_excess': ij_excess_error[:, col],
                 'col': col})

        result = pd.concat([ GetColDf(col) for col in range(len(full_lo_inds)) ])
```

If we simply look at the point-by-point error, CV and IJ are highly correlated.

```
In [10]: sns.jointplot(x='cv_error', y='ij_error', data=result);
```

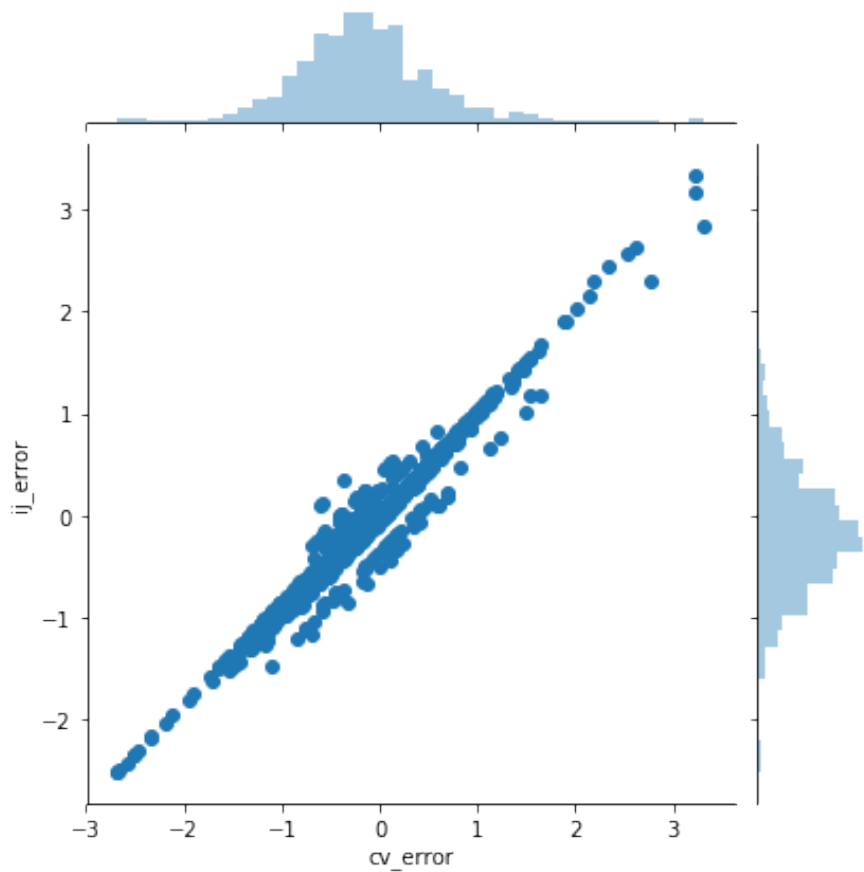

However, this is because the error in each point is dominated by the error at the original optimum. To meaningfully compare the IJ to CV, we should compare the difference between the IJ and CV error and the error at the original optimum. The distribution of these "difference-in-difference" errors is shown in the next plot.

Some clear outliers can be seen. However, note that, in this case, overplotting makes IJ looks worse than it is – in the histograms you can see that most differences are very small.

```
In [11]: sns.jointplot(x='cv_excess', y='ij_excess', data=result);
```

6

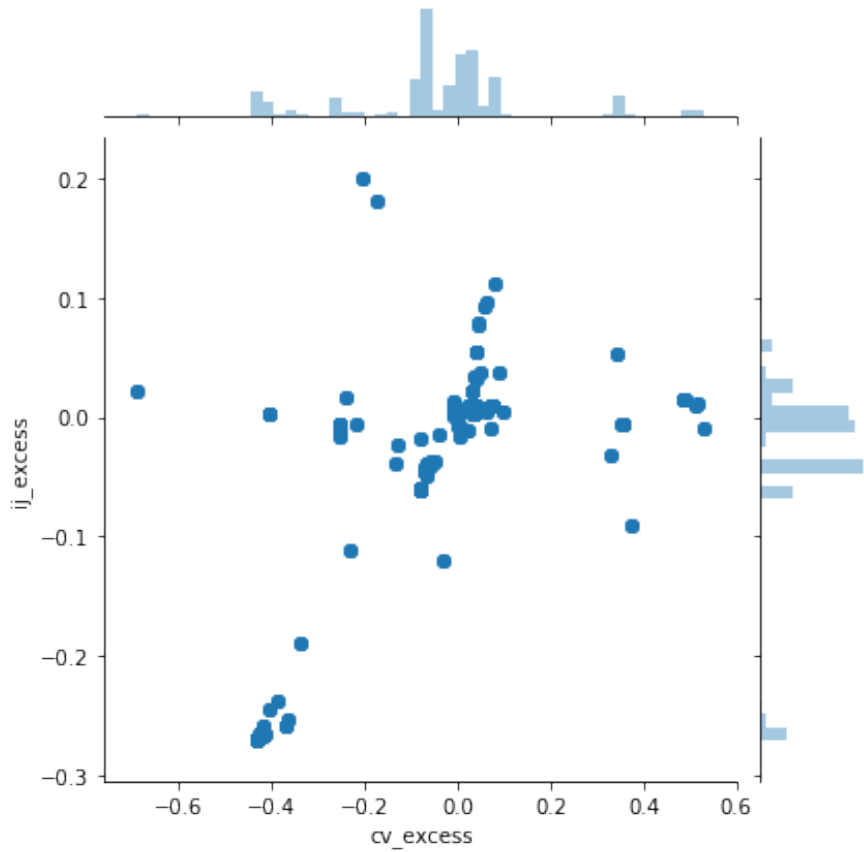

As you might expect from a linear approximation, the IJ does the worst when the predicted change for CV is large.

```
In [12]: misfit = np.max(np.abs(cv_excess_error - ij_excess_error), axis=1)
         abs_cv_excess_error = np.max(np.abs(cv_excess_error), axis=1)

         sns.jointplot(abs_cv_excess_error, misfit)

Out[12]: <seaborn.axisgrid.JointGrid at 0x7f3f74fe1908>
```

7

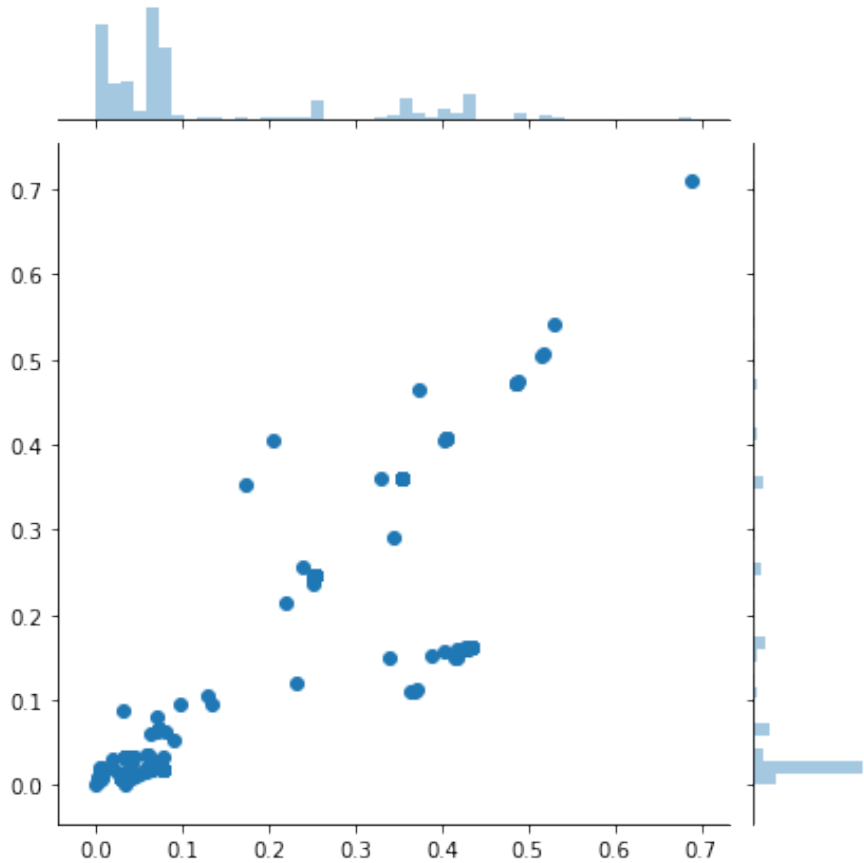

Finally, we visualize some of the genes where IJ badly misestimates the CV error. Clearly, in these cases, re-fitting with the left-out points (shown with large dots) produced large changes that the IJ did not capture. In general, it appears that the IJ errs relative to CV by not moving far enough from the original optimum.

Despite the poor fit on these extreme genes, we stress that most genes exhibited small changes in both CV and IJ. For these genes, IJ performs well enough to capture salient aspects of the estimated out-of-sample error. For more detailed analysis of this point, see the notebook `examine_and_save_results`.

```
In [13]: timepoints = initial_metadata['timepoints']
         timepoints_stretch = np.sqrt(timepoints)

         def PlotGenePredictions(gene_ind):
             _, figs = plt.subplots(1, 3, figsize=(15,6))
```

8

```python
                for i in range(3):
                    np.random.seed(42)
                    plot_utils_lib.PlotRegressionLine(
                        timepoints_stretch, regs_test, reg_params_test, gene_ind, this_plot=figs[i]
                    figs[i].plot(timepoints_stretch[full_lo_inds],
                                 regs_test.y[gene_ind, full_lo_inds], 'o', markersize=10)

                plot_utils_lib.PlotPredictionLine(
                    timepoints_stretch, regs_test, orig_pred, gene_ind, this_plot=figs[0])
                figs[0].set_title('Gene {} original fit'.format(gene_ind))

                plot_utils_lib.PlotPredictionLine(
                    timepoints_stretch, regs_test, ij_pred, gene_ind, this_plot=figs[1])
                figs[1].set_title('Gene {} IJ fit'.format(gene_ind))

                plot_utils_lib.PlotPredictionLine(
                    timepoints_stretch, regs_test, cv_pred, gene_ind, this_plot=figs[2])
                figs[2].set_title('Gene {} CV fit'.format(gene_ind))

In [14]: worst_fits = np.argsort(-1 * misfit)

         for gene in worst_fits[0:5]:
             PlotGenePredictions(gene)
```

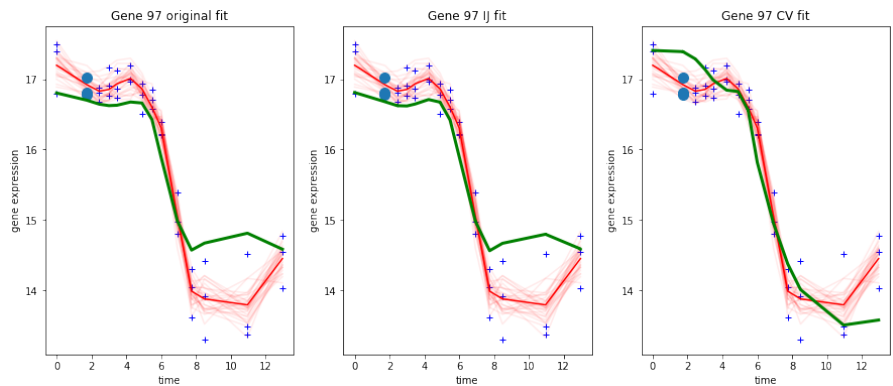

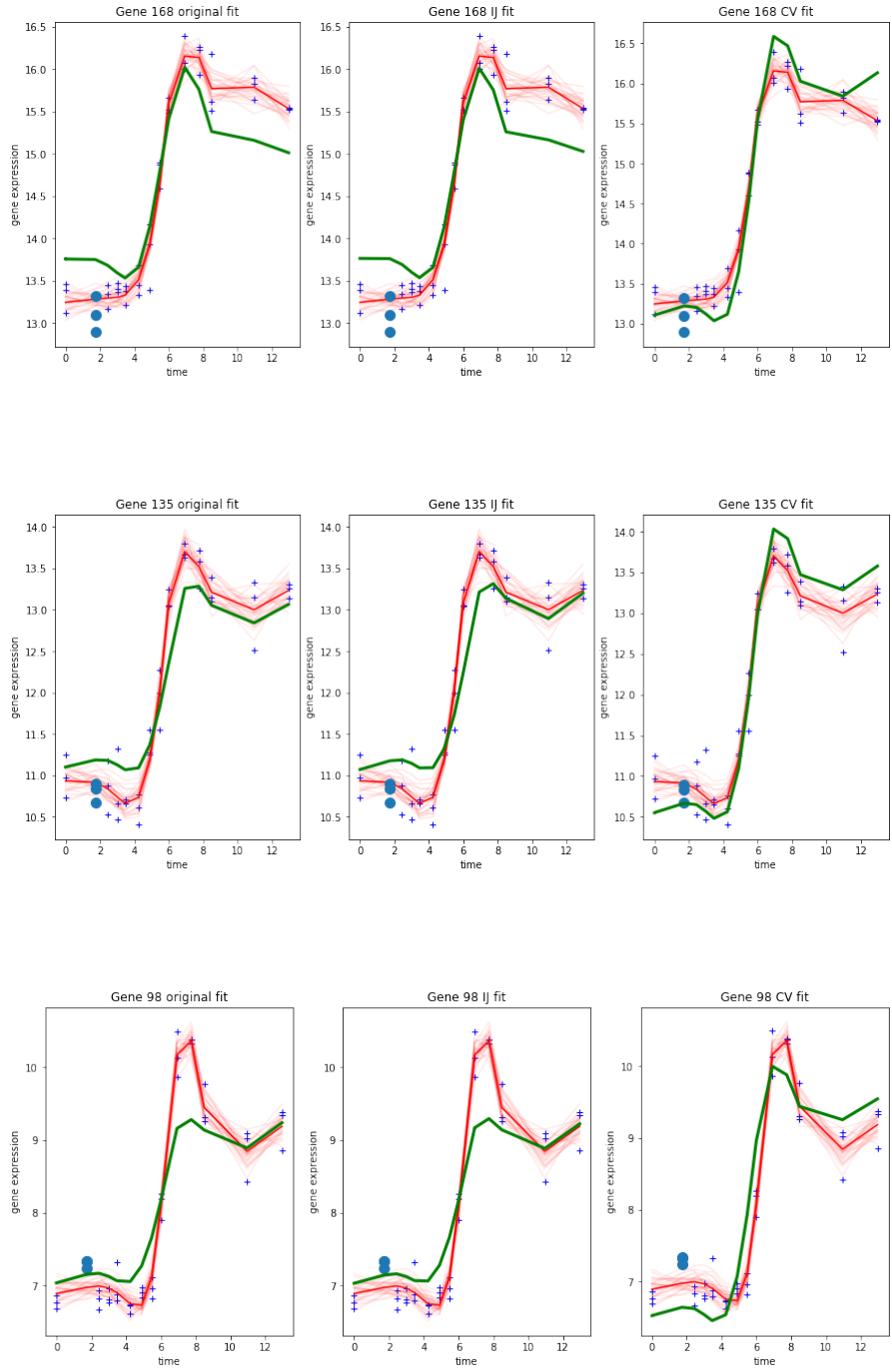

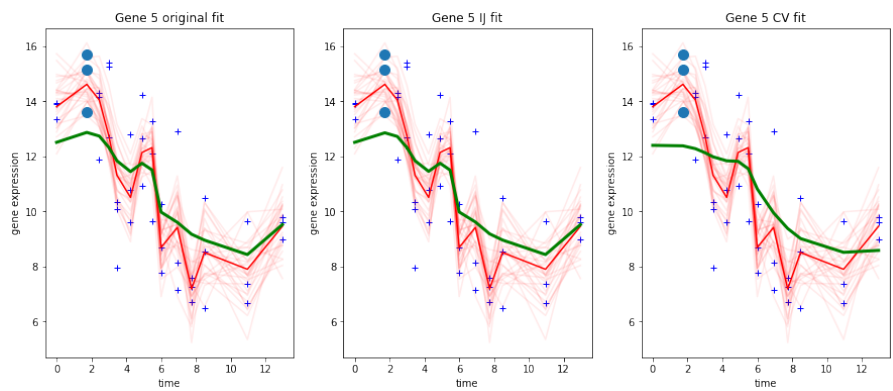

# examine_and_save_results

February 21, 2019

## 1 Detailed analysis of results.

This notebook loads the output of the scripts in the directory `cluster_scripts` (particularly, the final script, `run_slurm_pred_error.py`). It produces the Rdata file that is used for the graphs in the paper as well as a number of supplemental analyses.

```
In [1]: library(tidyverse)
        library(gridExtra)
        library(repr) # For setting plot sizes
        source("load_python_data_lib.R")
        py_main <- InitializePython()

 Attaching packages   tidyverse 1.2.1
 ggplot2 3.1.0       purrr   0.2.5
 tibble  1.4.2       dplyr   0.7.8
 tidyr   0.8.1       stringr 1.3.1
 readr   1.1.1       forcats 0.3.0
 Conflicts   tidyverse_conflicts()
 dplyr::filter() masks stats::filter()
 dplyr::lag()    masks stats::lag()

Attaching package: gridExtra

The following object is masked from package:dplyr:

    combine


Attaching package: reshape2

The following object is masked from package:tidyr:

    smiths



In [2]: # Choose the initialization method.
        init_method <- "kmeans" # This is the choice for the paper.
```

1

```
        #init_method <- "warm"

        # Choose whether or not to re-run the regressions before calculating test error.
        use_rereg <- FALSE # This is the choice for the paper.
        #use_rereg <- TRUE

        # This is the file that is used in the paper's knitr.
        save_dir <- "../../fits"
        save_filename <- sprintf("paper_results_init_%s_rereg_%s.Rdata", init_method, use_rereg)
```

### 1.0.1   Load the saved data for all dfs and k

```
In [3]: dfs <- list()
        metadata_dfs <- list()

        for (lo_num_times in 1:3) {
          cat("lo_num_times ", lo_num_times)
          for (df in 4:8) {
            cat(".")
            load_res <- LoadPredictionError(df, lo_num_times, init_method)
            this_refit_err_df <- load_res$refit_err_df
            this_metadata_df <- load_res$metadata_df
            this_refit_err_melt <- MeltErrorColumns(this_refit_err_df)
            dfs[[length(dfs) + 1]] <- this_refit_err_melt
            metadata_dfs[[length(metadata_dfs) + 1]] <- this_metadata_df
          }
          cat("\n")
        }
        cat("Done.\n")
        refit_err_melt <- do.call(bind_rows, dfs)
        metadata_df <- do.call(bind_rows, metadata_dfs)

lo_num_times  1...
lo_num_times  2...
lo_num_times  3...
Done.
```

### 1.0.2   Metadata (timing, parameter dimensions)

Make a tidy dataframe with the metadata. The parameter length, Hessian time, and initial optimization time are all reported in the text of the paper. Their values will be derived from this dataframe in knitr.

```
In [4]: metadata_df <-
            metadata_df %>%
            mutate(lr_hess_time=total_lr_time + initial_hess_time,
                   avg_lr_time=total_lr_time / num_comb,
                   avg_refit_time=total_refit_time / num_comb,
```

2

```
                  param_length=gmm_param_length + reg_param_length)
        print(names(metadata_df))

        select(metadata_df, df, param_length) %>%
            group_by(df) %>%
            summarize(param_length=unique(param_length))

        select(metadata_df, df, initial_hess_time, initial_opt_time) %>%
            group_by(df) %>%
            summarize(initial_hess_time=median(initial_hess_time),
                      initial_opt_time=median(initial_opt_time))

        round(median(metadata_df$initial_opt_time), digits=-1)
```

```
 [1] "num_comb"         "total_lr_time"     "total_refit_time"
 [4] "initial_opt_time" "initial_reg_time"  "initial_hess_time"
 [7] "gmm_param_length" "reg_param_length"  "df"
[10] "lo_num_times"     "init_method"       "lr_hess_time"
[13] "avg_lr_time"      "avg_refit_time"    "param_length"
```

| df | param_length |
|---|---|
| 4 | 25325 |
| 5 | 31643 |
| 6 | 38661 |
| 7 | 46379 |
| 8 | 54797 |

| df | initial_hess_time | initial_opt_time |
|---|---|---|
| 4 | 275.7295 | 31.44656 |
| 5 | 295.0325 | 41.84182 |
| 6 | 359.6855 | 35.11145 |
| 7 | 478.7345 | 50.88843 |
| 8 | 584.4987 | 77.02919 |

40

Make a dataframe for the timing plot from the metadata.

```
In [5]: metadata_graph_df <-
          metadata_df %>%
            select(df, lo_num_times, total_refit_time, lr_hess_time, initial_opt_time) %>%
            melt(id.vars=c("lo_num_times", "df"))
        head(metadata_graph_df)
```

| lo_num_times | df | variable | value |
|---|---|---|---|
| 1 | 4 | total_refit_time | 338.1638 |
| 1 | 5 | total_refit_time | 391.6006 |
| 1 | 6 | total_refit_time | 423.8322 |
| 1 | 7 | total_refit_time | 632.2635 |
| 1 | 8 | total_refit_time | 599.0894 |
| 2 | 4 | total_refit_time | 1123.7316 |

3

```
In [6]: ggplot(metadata_graph_df) +
            geom_bar(aes(x=lo_num_times, y=value, fill=variable, group=variable),
                     stat="identity", position=position_dodge()) +
            facet_grid( ~ df)
```

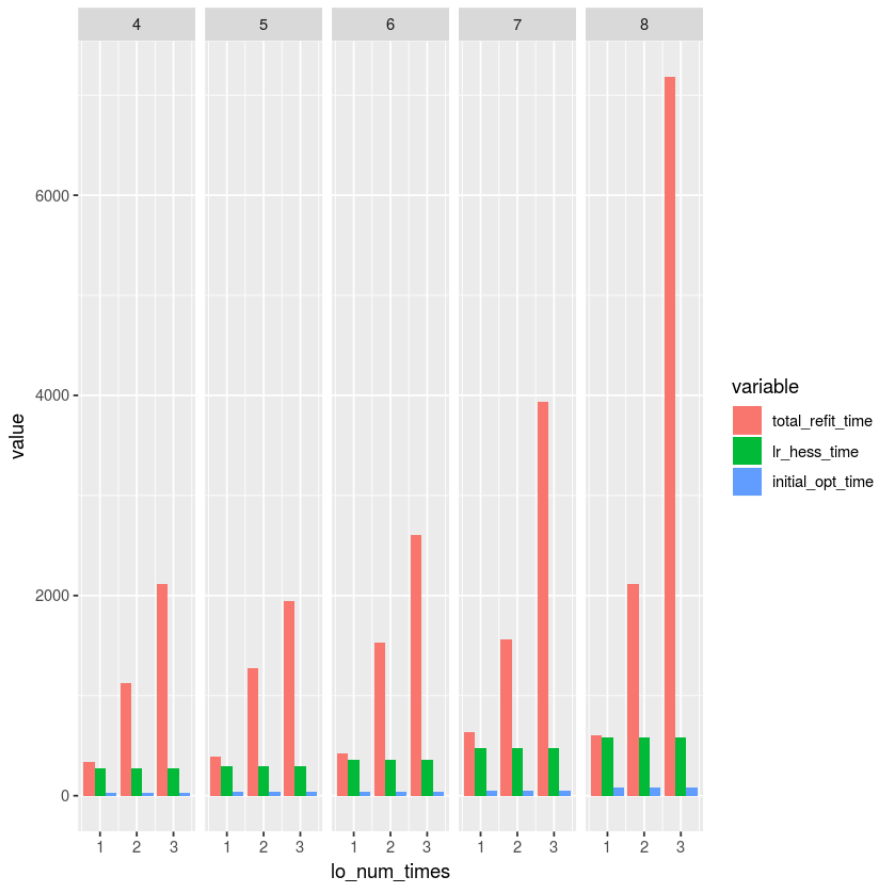

### 1.0.3  Calculate prediction errors

Make summaries of prediction error for various methods and datasets.

```
In [7]: # In-sample IJ error.
        lr_df <-
            refit_err_melt %>%
            filter(rereg==use_rereg, method=="lin", test==FALSE, measure=="err") %>%
```

4

```r
        rename(error=value) %>%
        mutate(output="lin_in_sample")

    # In-sample CV error.
    cv_df <-
      refit_err_melt %>%
      filter(rereg==use_rereg, method=="ref", test==FALSE, measure=="err") %>%
      rename(error=value) %>%
      mutate(output="cv_in_sample")

    # In-sample training error (no points left out).
    train_df <-
      refit_err_melt %>%
      filter(rereg==use_rereg, method=="ref", test==FALSE, measure=="train_err") %>%
      rename(error=value) %>%
      mutate(output="train_error")

    # Out-of-sample test error.
    test_df <-
      refit_err_melt %>%
      filter(rereg==use_rereg, method=="ref", test==TRUE, measure=="train_err") %>%
      rename(error=value) %>%
      mutate(output="test_error")

    refit_for_df_choice <- bind_rows(
      lr_df, cv_df, test_df, train_df)
```

`In [8]:` `head(refit_for_df_choice)`

| test | method | comb | rereg | gene | df | lo_num_times | time | measure | error | output |
|---|---|---|---|---|---|---|---|---|---|---|
| FALSE | lin | 0 | FALSE | 0 | 4 | 1 | 0 | err | 1.0088933 | lin_in_samp |
| FALSE | lin | 0 | FALSE | 1 | 4 | 1 | 0 | err | 0.1243607 | lin_in_samp |
| FALSE | lin | 0 | FALSE | 2 | 4 | 1 | 0 | err | -0.4340983 | lin_in_samp |
| FALSE | lin | 0 | FALSE | 3 | 4 | 1 | 0 | err | -0.2203431 | lin_in_samp |
| FALSE | lin | 0 | FALSE | 4 | 4 | 1 | 0 | err | 1.9032786 | lin_in_samp |
| FALSE | lin | 0 | FALSE | 5 | 4 | 1 | 0 | err | -0.2876837 | lin_in_samp |

Make a tidy dataframe for choosing `df`. The graph in the paper will be based on this dataframe.

Note that most of the signal for choosing `df` is already in the training data error. However, there is an uptick in error in both CV and IJ for `df=8` which is not captured by the training data error.

```r
In [9]: refit_err_summary <-
        refit_for_df_choice %>%
        group_by(output, df, lo_num_times) %>%
        mutate(esize=abs(error)) %>%
        summarize(med=median(esize),
                  mean=mean(esize),
                  n_obs=n(),
                  se=sd(esize) / sqrt(n_obs),
```

5

```
                    qlow=quantile(esize, 0.25),
                    qhigh=quantile(esize, 0.75))

    ggplot(refit_err_summary) +
      geom_line(aes(x=df, y=mean, group=output, color=output)) +
      geom_errorbar(aes(x=df, ymin=mean - 2 * se, ymax=mean + 2 * se,
                        group=output, color=output)) +
      facet_grid(~ lo_num_times) +
      ggtitle(sprintf("%d times left out", lo_num_times))
```

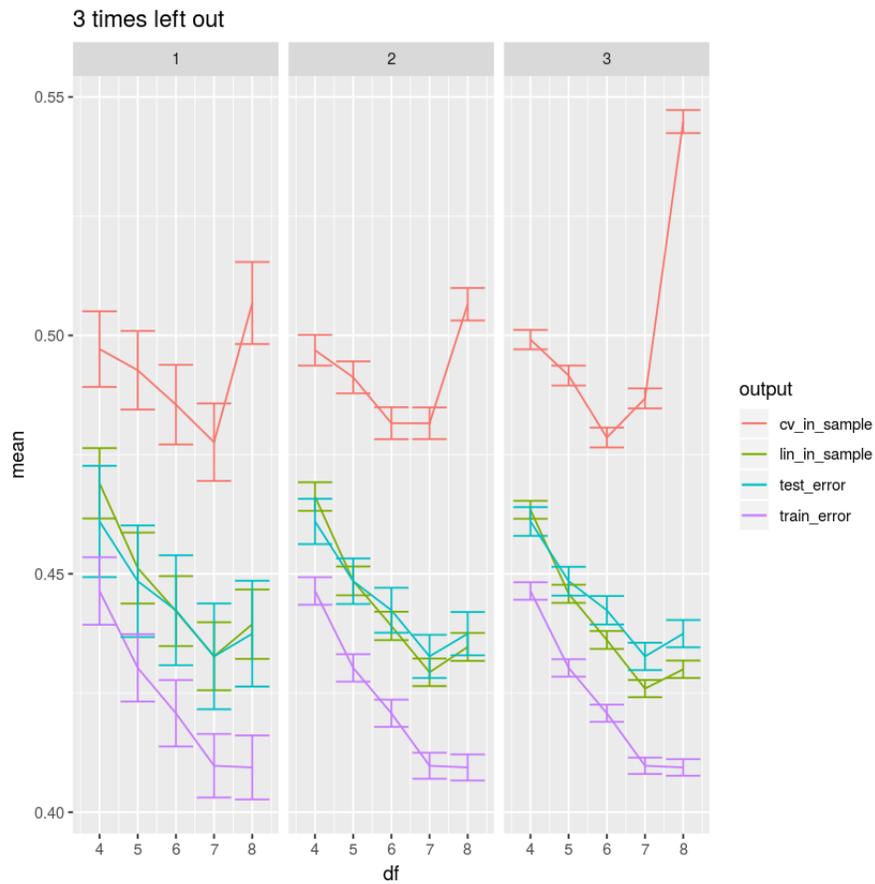

### 1.0.4 Gene-by-gene accuracy measures.

```
In [10]: refit_err_plot <-
             refit_err_melt %>%
```

6

```r
        filter(rereg==use_rereg) %>%
        dcast(df + lo_num_times + test + comb + rereg + gene + time ~ method + measure,
              value.var=c("value"))
```

We now look at the correlation between the CV and IJ prediction errors across genes. For each `df` and `k`, there are a number of different combinations of left-out points. We report the median, min, and max correlation coefficients across these combinations of left-out points.

First, we show the correlation between the raw prediction errors. Although the correlation is quite high, this is because the training error at the original optimum is the principle source of variation in the errors across genes, and this quantity is common to both CV and IJ.

```r
In [11]: err_corr <- refit_err_plot %>%
            filter(test==FALSE, rereg==use_rereg) %>%
            group_by(df, lo_num_times, comb) %>%
            summarize(r=cor(lin_err, ref_err)) %>%
            group_by(df, lo_num_times) %>%
            summarize(med_r=median(r), min_r=min(r), max_r=max(r))

         print("Correlation between error: ")
         print(err_corr)

[1] "Correlation between error: "
# A tibble: 15 x 5
# Groups:   df [?]
      df lo_num_times med_r min_r max_r
   <int>        <int> <dbl> <dbl> <dbl>
 1     4            1 0.974 0.949 0.984
 2     4            2 0.975 0.902 0.992
 3     4            3 0.967 0.871 0.991
 4     5            1 0.963 0.856 0.983
 5     5            2 0.966 0.860 0.984
 6     5            3 0.947 0.759 0.981
 7     6            1 0.980 0.807 0.985
 8     6            2 0.968 0.835 0.986
 9     6            3 0.929 0.759 0.983
10     7            1 0.962 0.794 0.974
11     7            2 0.952 0.737 0.976
12     7            3 0.914 0.599 0.974
13     8            1 0.962 0.703 0.971
14     8            2 0.941 0.663 0.974
15     8            3 0.829 0.251 0.958
```

A more meaningful measure is the correlation in the excess error for IJ and CV over the error at the original fit.

```r
In [12]: diff_corr <- refit_err_plot %>%
            filter(test==FALSE, rereg==use_rereg) %>%
            group_by(df, lo_num_times, comb) %>%
```

7

```r
            summarize(r=cor(lin_e_diff, ref_e_diff)) %>%
            group_by(df, lo_num_times) %>%
            summarize(med_r=median(r), min_r=min(r), max_r=max(r))

        print("Correlation between difference from train error: ")
        print(diff_corr)
```

```
[1] "Correlation between difference from train error: "
# A tibble: 15 x 5
# Groups:   df [?]
      df lo_num_times med_r   min_r max_r
   <int>        <int> <dbl>   <dbl> <dbl>
 1     4            1 0.483  0.0956 0.844
 2     4            2 0.577  0.277  0.828
 3     4            3 0.605  0.303  0.833
 4     5            1 0.464  0.143  0.728
 5     5            2 0.510  0.330  0.709
 6     5            3 0.510  0.312  0.671
 7     6            1 0.655  0.368  0.783
 8     6            2 0.588  0.218  0.845
 9     6            3 0.499  0.0701 0.737
10     7            1 0.660  0.512  0.760
11     7            2 0.564  0.224  0.863
12     7            3 0.491  0.0344 0.801
13     8            1 0.744  0.380  0.900
14     8            2 0.646  0.166  0.862
15     8            3 0.214 -0.226  0.767
```

For higher degrees of freedom, increasing the number of left-out points seems to decrease the IJ's accuracy, as you might expect.

```r
In [13]: ggplot(diff_corr) +
           geom_bar(aes(x=paste(df, lo_num_times, sep=","),
                        y=med_r, fill=as.character(df)), stat="identity")
```

8

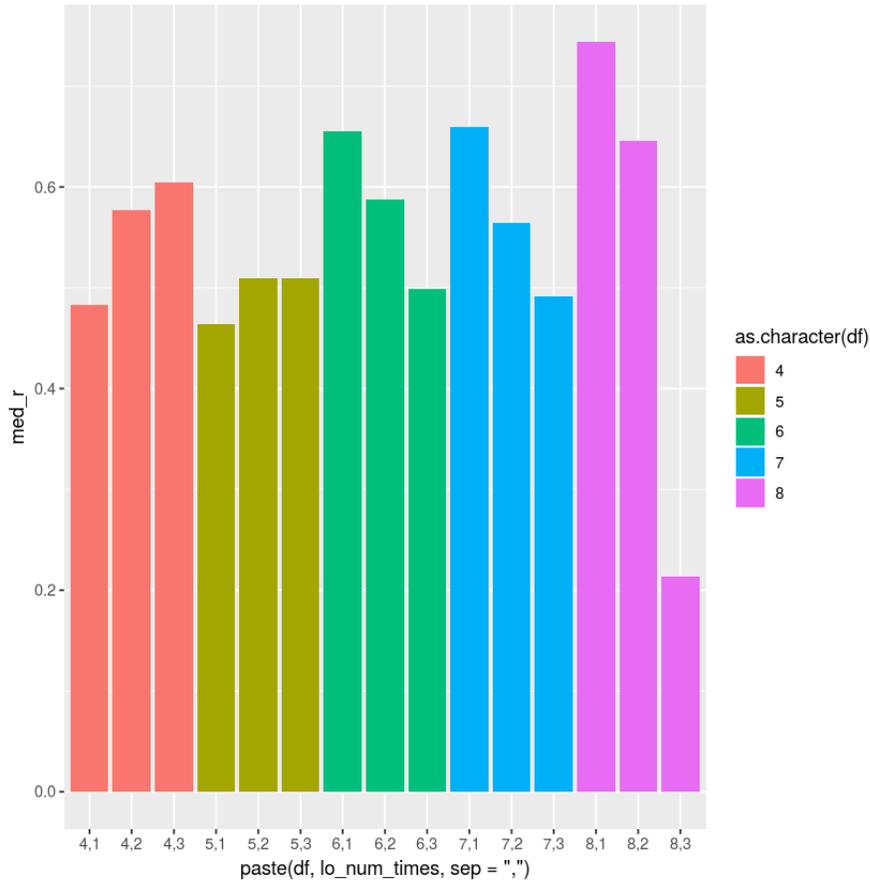

Plot the densities of the IJ and CV with points to show outliers. This is a graphical version of the results summarized by the correlation tables above.

```
In [14]: # There are a few outliers, so limit the extent of the plot so that
         # the bulk of the distribution is visible.
         qlim <- quantile(refit_err_plot$ref_e_diff, c(0.1, 0.9))

         options(repr.plot.width=4, repr.plot.height=20)

         # This plot, or ones like it, is probably the best measure of
         # the accuracy of the IJ.
         ggplot(filter(refit_err_plot, test == FALSE, lo_num_times==1)) +
             geom_point(aes(x=ref_e_diff, y=lin_e_diff), alpha=0.01) +
             geom_density2d(aes(x=ref_e_diff, y=lin_e_diff)) +
```

9

```
              geom_abline(aes(slope=1, intercept=0)) +
              facet_grid(df ~ rereg) +
              xlim(qlim[1], qlim[2]) + ylim(qlim[1], qlim[2])

Warning message:
Removed 10770 rows containing non-finite values (stat_density2d).Warning message:
Removed 10770 rows containing missing values (geom_point).
```

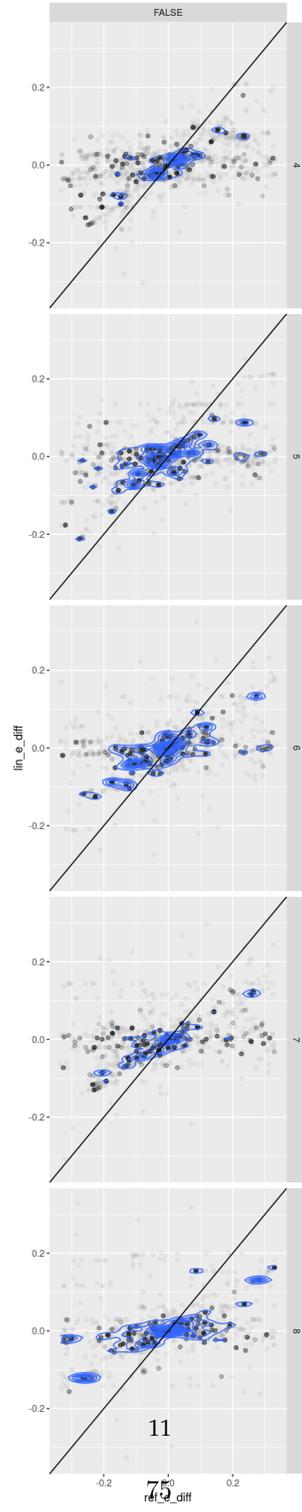

75

### 1.0.5 Save results for plotting in the paper.

```
In [15]: print(sprintf("Saving to %s", file.path(save_dir, save_filename)))
         save(refit_err_summary,
              metadata_df,
              diff_corr,
              err_corr,
              file=file.path(save_dir, save_filename))

[1] "Saving to ../../fits/paper_results_init_kmeans_rereg_FALSE.Rdata"
```



\end{document}